\documentclass[11pt,a4paper,reqno,twoside]{amsart}%
\usepackage{amsfonts}
\usepackage{amscd}
\usepackage{amssymb}
\usepackage{amsthm}
\usepackage{amsmath}
\usepackage{latexsym}
\usepackage{graphicx}
\usepackage{subfigure}
\usepackage{float}
\usepackage[a4paper,asymmetric]{geometry}
\usepackage{times}
\providecommand{\U}[1]{\protect\rule{.1in}{.1in}}
\newtheorem{theorem}{Theorem}
\theoremstyle{plain}

\newtheorem{assumption}{Assumption}

\newtheorem{corollary}[theorem]{Corollary}

\newtheorem{example}[theorem]{Example}

\newtheorem{lemma}[theorem]{Lemma}

\newtheorem{proposition}[theorem]{Proposition}
\newtheorem{remark}[theorem]{Remark}

\numberwithin{equation}{section}
\numberwithin{theorem}{section}

\allowdisplaybreaks
\begin{document}
\title{A pricing measure to explain the risk premium in power markets}
\date{\today }
\author[Benth]{Fred Espen Benth}
\address[Fred Espen Benth]{\\
Centre of Mathematics for Applications \\
University of Oslo\\
P.O. Box 1053, Blindern\\
N--0316 Oslo, Norway}
\email[]{fredb\@@math.uio.no}
\urladdr{http://folk.uio.no/fredb/}
\author[Ortiz-Latorre]{Salvador Ortiz-Latorre}
\address[Salvador Ortiz-Latorre]{\\
Centre of Mathematics for Applications \\
University of Oslo\\
P.O. Box 1053, Blindern\\
N--0316 Oslo, Norway}
\email[]{salvador.ortiz-latorre\@@cma.uio.no}
\thanks{We are grateful for the financial support from the project "Energy Markets:
Modeling, Optimization and Simulation (EMMOS)", funded by the Norwegian
Research Council under grant Evita/205328.}

\begin{abstract}
In electricity markets, it is sensible to use a two-factor model with mean
reversion for spot prices. One of the factors is an Ornstein-Uhlenbeck (OU)
process driven by a Brownian motion and accounts for the small variations. The
other factor is an OU process driven by a pure jump L\'{e}vy process and
models the characteristic spikes observed in such markets. When it comes to
pricing, a popular choice of pricing measure is given by the Esscher transform
that preserves the probabilistic structure of the driving L\'{e}vy processes,
while changing the levels of mean reversion. Using this choice one can
generate stochastic risk premiums (in geometric spot models) but with
(deterministically) changing sign. In this paper we introduce a pricing change
of measure, which is an extension of the Esscher transform. With this new
change of measure we also can slow down the speed of mean reversion and
generate stochastic risk premiums with stochastic non constant sign, even in
arithmetic spot models. In particular, we can generate risk profiles with
positive values in the short end of the forward curve and negative values in
the long end. Finally, our pricing measure allows us to have a stationary spot
dynamics while still having randomly fluctuating forward prices for contracts
far from maturity.

\end{abstract}
\maketitle

\section{Introduction}

In modelling and analysis of forward and futures prices in commodity markets,
the \textit{risk premium} plays an important role. It is defined as the
difference between the forward price and the expected commodity spot price at
delivery, and the classical theory predicts a \textit{negative} risk premium.
The economical argument for this is that producers of the commodity is willing
to pay a premium for hedging their production (see Geman~\cite{G} for a
discussion, as well as a list of references).

Geman and Vasicek ~\cite{GV} argued that in power markets, the consumers may
hedge the price risk using forward contracts which are close to delivery, and
thus creating a \textit{positive} premium. Power is a non-storable commodity,
and as such may experience rather large price variations over short time
(sometimes referred to as spikes). One might observe a risk premium which may
be positive in the short end of the forward market, and negative in the long
end where the producers are hedging their power generation. A theoretical and
empirical foundation for this is provided in, for example, Bessembinder and
Lemon~\cite{BL} and Benth, Cartea and Kiesel \cite{BCK}.

When deriving the forward price, one specifies a pricing probability and
computes the forward price as the conditional expected spot at delivery. In
the power market, this pricing probability is not necessarily a so-called
equivalent martingale measure, or a risk neutral probability (see Bingham and
Kiesel \cite{BK}), as the spot is not tradeable in the usual sense. Thus, a
pricing probability can a priori be any equivalent measure, and in effect is
an indirect specification of the risk premium. In this paper we suggest a new
class of pricing measures which gives a stochastically varying risk premium.

We will focus our considerations on the power market, where typically a spot
price model may take the form as a two-factor mean reversion dynamics. Lucia
and Schwartz~\cite{LS} considered two-factor models for the electricity spot
price dynamics in the Nordic power market NordPool. Both arithmetic and
geometric models where suggested, that is, either directly modelling the spot
price by a two-factor dynamics, or assuming such a model for the logarithmic
spot prices. Their models were based on Brownian motion and, as such, not able
to capture the extreme variations in the power spot markets. Cartea and
Figueroa~\cite{CF} used a compound Poisson process to model spikes, that is,
extreme price jumps which are quickly reverted back to "normal levels". Benth,
\v{S}altyt\.{e} Benth and Koekebakker~\cite{BSBK} give a general account on
multi-factor models based on Ornstein-Uhlenbeck processes driven by both
Brownian motion and L\'{e}vy processes. Empirical studies suggest a stationary
power spot price dynamics after explaining deterministic seasonal variations
(see e.g. Barndorff-Nielsen, Benth and Veraart~\cite{BNBV-spot} for a study of
spot prices at EEX, the German power exchange). We will in this paper focus on
a two-factor model for the spot, where each factor is an Ornstein-Uhlenbeck
process, driven by a Brownian motion and a jump process, respectively. The
first factor models the "normal variations" of the spot price, whereas the
second accounts for sudden jumps (spikes) due to unexpected imbalances in
supply and demand.

The standard approach in power markets is to specify a pricing measure which
is preserving the L\'{e}vy property. This is called the Esscher transform (see
Benth et al.~\cite{BSBK}), and works for L\'{e}vy processes as the Girsanov
transform with a constant parameter for Brownian motion. The effect of doing
such a measure change is to adjust the mean reversion level, and it is known
that the risk premium becomes deterministic and typically either positive or
negative for all maturities along the forward curve.

We propose a class of measure changes which slows down the \textit{speed of
mean reversion} of the two factors. As it turns out, in conjunction with an
Esscher transform as mentioned above, we can produce a stochastically varying
risk premium, where potential positive premiums in the short end of the market
can be traced back to sudden jumps in the spike factor being slowed down under
the pricing measure. This result holds for arithmetic spot models, whereas the
geometric ones are much harder to analyse under this change of probability.
The class of probabilities preserves the Ornstein-Uhlenbeck structure of the
factors, and as such may be interpreted as a \textit{dynamic structure
preserving measure change}. For the L\'{e}vy driven component, the L\'{e}vy
property is lost in general, and we obtain a rather complex jump process with
state-dependent (random) compensator measure.

We can explicitly describe the density process for our measure change. The
theoretical contribution of this paper, besides the new insight on risk
premium, is a proof that the density process is a true martingale process,
indeed verifying that we have constructed a \textit{probability measure}. This
verification is not straightforward because the kernels used to define the
density process, through stochastic exponentiation, are stochastic and
unbounded. Hence, the usual criterion by L\'{e}pingle-M\'{e}min \cite{LeMe78}
is difficult to apply and, furthermore, it does not provide sharp results. We
follow the same line of reasoning as in a very recent paper by Klebaner and
Lipster \cite{KleLi11}. Although their result is more general than ours in
some respects, it does not apply directly to our case because we need some
additional integrability requirements. The proof is roughly as follows. First,
we reduce the problem to show the uniform integrability of the sequence of
random variables obtained by evaluating at the end of the trading period the
localised density process. This sequence of random variables naturally induces
a sequence of measure changes which, combined with an easy inequality for the
logarithm function, allow us to get rid of the stochastic exponential in the
expression to be bounded. Finally, we can reduce the problem to get an uniform
bound for the second moment of the factors under these new probability measures.

Interestingly, as our pricing probability is reducing the speed of mean
reversion, we might in the extreme situation "turn off" the mean reversion
completely (by reducing it to zero). For example, if we take the Brownian
factor as the case, we can have a stationary dynamics of the "normal
variations" in the market, but when looking at the process under the pricing
probability the factor can be non-stationary, that is, a drifted Brownian
motion. A purely stationary dynamics for the spot will produce constant
forward prices in the long end of the market, something which is not observed
empirically. Hence, the inclusion of non-stationary factors are popular in
modelling the spot-forward markets. In many studies of commodity spot and
forward markets, one is considering a two-factor model with one non-stationary
and one stationary component. The stationary part explains the short term
variations, while the non-stationary is supposed to account for long-term
price fluctuations in the spot (see Gibson and Schwartz~\cite{GS} and Schwartz
and Smith~\cite{SS} for such models applied to oil markets). Indeed, the power
spot models in Lucia and Schwartz~\cite{LS} are of this type. It is hard to
detect the long term factor in spot price data, and one is usually filtering
it out from the forward prices using contracts far from delivery.
Theoretically, such contracts should have a dynamics being proportional to the
long term factor. Contrary to this approach, one may in view of our new
results, suggest a \textit{stationary} spot dynamics and introduce a pricing
measure which turns one of the factors into a non-stationary dynamics. This
would imply that one could directly fit a two-factor stationary spot model to
power data, and next calibrate a measure change to account for the long term
variations in the forward prices by turning off (or significantly slow down)
the speed of mean reversion.

Our results are presented as follows: in the next section we introduce the
basic assumptions and properties satisfied by the factors in our model. Then,
in Section \ref{SecChangeOfMeasure}, we define the new change of measure and
prove the main results regarding the uniform integrability of its density
process. We deal with the Brownian and pure jump case separately. Finally, in
Section \ref{SecStudyOfRiskPrem}, we recall the arithmetic and geometric spot
price models. We compute the forward price processes induced by this change of
measure and we discuss the risk premium profiles that can be obtained.

\section{The mathematical set up\label{SecMathSetup}}

Suppose that $(\Omega,\mathcal{F},\{\mathcal{F}_{t}\}_{t\in\lbrack0,T]},P)$ is
a complete filtered probability space, where $T>0$ is a fixed finite time
horizon. On this probability space there are defined $W$, a standard Wiener
process, and $L,$ a pure jump L\'{e}vy subordinator with finite expectation,
that is a L\'{e}vy process with the following L\'{e}vy-It\^{o} representation
$L(t)=\int_{0}^{t}\int_{0}^{\infty}zN^{L}(ds,dz),t\in\lbrack0,T],$ where
$N^{L}(ds,dz)$ is a Poisson random measure with L\'{e}vy measure $\ell$
satisfying $\int_{0}^{\infty}z\ell(dz)<\infty.$ We shall suppose that $W$ and
$L$ are independent of each other. The following assumption is minimal, having
in mind, on the one hand, that our change of measure extends the Esscher
transform and, on the other hand, that we are going to consider a geometric
spot price model.

\begin{assumption}
\label{Assumption_Exp_Theta_L}We assume that%
\begin{equation}
\Theta_{L}\triangleq\sup\{\theta\in\mathbb{R}_{+}:\mathbb{E}[e^{\theta
L(1)}]<\infty\}, \label{Def_Theta_L}%
\end{equation}
is strictly positive constant, which may be $\infty.$
\end{assumption}

Actually, to have the geometric model well defined we will need to assume
later that $\Theta_{L}>1.$ Some remarks are in order.

\begin{remark}
\label{Remark_Cumulants}In $(-\infty,\Theta_{L})$ the cumulant (or log moment
generating) function $\kappa_{L}(\theta)\triangleq\log\mathbb{E}_{P}[e^{\theta
L(1)}]$ is well defined and analytic. As $0\in(-\infty,\Theta_{L})$, $L$ has
moments of all orders. Also, $\kappa_{L}(\theta)$ is convex, which yields that
$\kappa_{L}^{\prime\prime}(\theta)\geq0$ and, hence, that $\kappa_{L}^{\prime
}(\theta)$ is non decreasing. Finally, as a consequence of $L\geq0,$ a.s., we
have that $\kappa_{L}^{\prime}(\theta)$ is non negative.
\end{remark}

\begin{remark}
\label{Remark_D_Cumulants}Thanks to the L\'{e}vy-Kintchine representation of
$L$ we can express $\kappa_{L}(\theta)$ and its derivatives in terms of the
L\'{e}vy measure $\ell.$ We have that for $\theta\in(-\infty,\Theta_{L})$%
\begin{align*}
\kappa_{L}(\theta)  &  =\int_{0}^{\infty}(e^{\theta z}-1)\ell(dz)<\infty,\\
\kappa_{L}^{(n)}(\theta)  &  =\int_{0}^{\infty}z^{n}e^{\theta z}%
\ell(dz)<\infty,\quad n\in\mathbb{N},
\end{align*}
showing, in fact, that $\kappa_{L}^{(n)}(\theta)>0,n\in\mathbb{N}.$
\end{remark}

Consider the OU processes
\begin{align}
X(t)  &  =X(0)+\int_{0}^{t}(\mu_{X}-\alpha_{X}X(s))ds+\sigma_{X}W(t)\quad
t\in\lbrack0,T],\label{Equ_OU_Brownian}\\
Y(t)  &  =Y(0)+\int_{0}^{t}(\mu_{Y}-\alpha_{Y}Y(s))ds+L(t),\quad t\in
\lbrack0,T], \label{Equ_OU_Levy}%
\end{align}
with $\alpha_{X},\sigma_{X},\alpha_{Y}>0,\mu_{X},X(0)\in\mathbb{R},\mu
_{Y},Y(0)\geq0.$ Note that, in equation $\left(  \ref{Equ_OU_Brownian}\right)
,$ $X$ is written as a sum of a finite variation process and a martingale. We
may also rewrite equation $(\ref{Equ_OU_Levy})$ as a sum of a finite variation
part and pure jump martingale%
\[
Y(t)=Y(0)+\int_{0}^{t}(\mu_{Y}+\kappa_{L}^{\prime}(0)-\alpha_{Y}%
Y(s))ds+\int_{0}^{t}\int_{0}^{\infty}z\tilde{N}^{L}(ds,dz),\quad t\in
\lbrack0,T],
\]
where $\tilde{N}^{L}(ds,dz)\triangleq N^{L}(ds,dz)-ds\ \ell(dz)$ is the
compensated version of $N^{L}(ds,dz)$. In the notation of Shiryaev
\cite{Sh99}, page 669, the predictable characteristic triplets (with respect
to the pseudo truncation function $g(x)=x$) of $X$ and $Y$ are given by
\[
(B^{X}(t),C^{X}(t),\nu^{X}(dt,dz))=(\int_{0}^{t}(\mu_{X}-\alpha_{X}%
X(s))ds,\sigma_{X}^{2}t,0),\quad t\in\lbrack0,T],
\]
and%
\[
(B^{Y}(t),C^{Y}(t),\nu^{Y}(dt,dz))=(\int_{0}^{t}(\mu_{Y}+\kappa_{L}^{\prime
}(0)-\alpha_{Y}Y(s))ds,0,\ell(dz)dt),\quad t\in\lbrack0,T],
\]
respectively. In addition, applying It\^{o} formula to $e^{\alpha_{X}t}X(t)$
and $e^{\alpha_{Y}t}Y(t),$ one can find the following explicit expressions for
$X(t)$ and $Y(t)$%
\begin{align}
X(t)  &  =X(s)e^{-\alpha_{X}(t-s)}+\frac{\mu_{X}}{\alpha_{X}}(1-e^{-\alpha
_{X}(t-s)})+\sigma_{X}\int_{s}^{t}e^{-\alpha_{X}(t-u)}%
dW(u),\label{Equ_X_Explicit_P}\\
Y(t)  &  =Y(s)e^{-\alpha_{Y}(t-s)}+\frac{\mu_{Y}+\kappa_{L}^{\prime}%
(0)}{\alpha_{Y}}(1-e^{-\alpha_{Y}(t-s)})+\int_{s}^{t}\int_{0}^{\infty
}e^{-\alpha_{Y}(t-u)}z\tilde{N}^{L}(du,dz), \label{Equ_Y_Explicit_P}%
\end{align}
where $0\leq s\leq t\leq T.$

\begin{remark}
\label{Remark_X_Gaussian}Using that the stochastic integral of a deterministic
function is Gaussian, one easily gets that $X$ is a Gaussian process and
$X(t)\sim\mathcal{N}(m_{t},\Sigma_{t}^{2})$ with%
\begin{align*}
m_{t}  &  =X(0)e^{-\alpha_{X}t}+\frac{\mu_{X}}{\alpha_{X}}(1-e^{-\alpha_{X}%
t}),\quad t\in\lbrack0,T],\\
\Sigma_{t}^{2}  &  =\frac{\sigma_{X}^{2}}{2\alpha_{X}}(1-e^{-2\alpha_{X}%
t}),\quad t\in\lbrack0,T].
\end{align*}

\end{remark}

\section{The change of measure\label{SecChangeOfMeasure}}

We will consider a parametrized family of measure changes which will allow us
to simultaneously modify the speed and the level of mean reversion in
equations $\left(  \ref{Equ_OU_Brownian}\right)  $ and $\left(
\ref{Equ_OU_Levy}\right)  $. The density processes of these measure changes
will be determined by the stochastic exponential of certain martingales. To
this end consider the following families of kernels%
\begin{align}
G_{\theta_{1},\beta_{1}}(t)  &  \triangleq\sigma_{X}^{-1}\left(  \theta
_{1}+\alpha_{X}\beta_{1}X(t)\right)  ,\quad t\in\lbrack0,T],\label{Equ_G}\\
H_{\theta_{2},\beta_{2}}(t,z)  &  \triangleq e^{\theta_{2}z}\left(
1+\frac{\alpha_{Y}\beta_{2}}{\kappa_{L}^{\prime\prime}(\theta_{2}%
)}zY(t-)\right)  ,\quad t\in\lbrack0,T],z\in\mathbb{R}. \label{Equ_H}%
\end{align}
The parameters $\bar{\beta}\triangleq(\beta_{1},\beta_{2})$ and $\bar{\theta
}\triangleq(\theta_{1},\theta_{2})$ will take values on the following sets
$\bar{\beta}\in\lbrack0,1]^{2},\bar{\theta}\in\bar{D}_{L}\triangleq
\mathbb{R}\times D_{L},$where $D_{L}\triangleq(-\infty,\Theta_{L}/2)$ and
$\Theta_{L}$ is given by equation $\left(  \ref{Def_Theta_L}\right)  .$ By
Assumption $\left(  \ref{Assumption_Exp_Theta_L}\right)  $ and Remarks
\ref{Remark_Cumulants} and \ref{Remark_D_Cumulants} these kernels are well defined.

\begin{remark}
Under the assumption $\int_{0}^{\infty}z^{3}e^{\Theta_{L}z}\ell(dz)<\infty,$
which is stronger than $\int_{0}^{\infty}e^{\Theta_{L}z}\ell(dz)<\infty,$ one
can consider the set \textrm{cl}($D_{L})=(-\infty,\Theta_{L}/2]$ and our
results still hold by changing $\kappa_{L}^{\prime}(\theta),\kappa_{L}%
^{\prime\prime}(\theta)$ and $\kappa_{L}^{(3)}(\theta)$ by its left
derivatives at the rigth end of $D_{L}.$
\end{remark}

\begin{example}
\label{Example_Subordinators} Typical examples of $\ell,\Theta_{L}$ and
$D_{L}$ are the following:

\begin{enumerate}
\item Bounded support: $L$ has a jump of size $a,$ i.e. $\ell=\delta_{a}.$ In
this case $\Theta_{L}=\infty$ and $D_{L}=\mathbb{R}.$

\item Finite activity: $L$ is a compound Poisson process with exponential
jumps, i.e., $\ell(dz)=ce^{-\lambda z}1_{(0,\infty)}\allowbreak dz,$ for some
$c>0$ and $\lambda>0.$ In this case $\Theta_{L}=\lambda$ and $D_{L}%
=(-\infty,\lambda/2).$

\item Infinite activity: $L$ is a tempered stable subordinator, i.e.,
$\ell(dz)=cz^{-(1+\alpha)}\allowbreak e^{-\lambda z}\allowbreak1_{(0,\infty
)}dz,$ for some $c>0,\lambda>0$ and $\alpha\in\lbrack0,1).$ In this case also
$\Theta_{L}=\lambda$ and $D_{L}=(-\infty,\lambda/2).$
\end{enumerate}
\end{example}

Next, for $\bar{\beta}\in\lbrack0,1]^{2},\bar{\theta}\in\bar{D}_{L},$ define
the following family of Wiener and Poisson integrals%
\begin{align}
\tilde{G}_{\theta_{1},\beta_{1}}(t)  &  \triangleq\int_{0}^{t}G_{\theta
_{1},\beta_{1}}(s)dW(s),\quad t\in\lbrack0,T],\label{Equ_G_tilda}\\
\tilde{H}_{\theta_{2},\beta_{2}}(t)  &  \triangleq\int_{0}^{t}\int_{0}%
^{\infty}\left(  H_{\theta_{2},\beta_{2}}(s,z)-1\right)  \tilde{N}%
^{L}(ds,dz),\quad t\in\lbrack0,T], \label{Equ_H_tilda}%
\end{align}
associated to the kernels $G_{\theta_{1},\beta_{1}}$ and $H_{\theta_{2}%
,\beta_{2}},$ respectively.

\begin{remark}
\label{Remark_Stoch_Exp}Let $M$ be a semimartingale on $(\Omega,\mathcal{F}%
,\{\mathcal{F}_{t}\}_{t\in\lbrack0,T]},P)$ and denote by $\mathcal{E}(M)$ the
stochastic exponential of $M,$ that is, the unique strong solution of
\begin{align*}
d\mathcal{E}(M)(t)  &  =\mathcal{E}(M)(t-)dM(t),\quad t\in\lbrack0,T],\\
\mathcal{E}(M)(t)  &  =1.
\end{align*}
When $M$ is a local martingale, $\mathcal{E}(M)$ is also a local martingale.
If $\mathcal{E}(M)$ is positive, then $\mathcal{E}(M)$ is also a
supermartingale and $\mathbb{E}_{P}[\mathcal{E}(M)(t)]\leq1,t\in\lbrack0,T].$
In that case, one has that $\mathcal{E}(M)$ is a true martingale if and only
$\mathbb{E}_{P}[\mathcal{E}(M)(T)]=1.$ If $\mathcal{E}(M)$ is a positive true
martingale, it can be used as a density process to define a new probability
measure $Q,$ equivalent to $P,$ that is, $\left.  \frac{dQ}{dP}\right\vert
_{\mathcal{F}_{t}}=\mathcal{E}(M)(t),t\in\lbrack0,T].$
\end{remark}

The desired family of measure changes is given by $Q_{\bar{\theta},\bar{\beta
}}\sim P,\bar{\beta}\in\lbrack0,1]^{2},\bar{\theta}\in\bar{D}_{L},$ with%
\begin{equation}
\left.  \frac{dQ_{\bar{\theta},\bar{\beta}}}{dP}\right\vert _{\mathcal{F}_{t}%
}\triangleq\mathcal{E}(\tilde{G}_{\theta_{1},\beta_{1}}+\tilde{H}_{\theta
_{2},\beta_{2}})(t),\quad t\in\lbrack0,T], \label{Equ_Def_Q_theta_beta_p}%
\end{equation}
where we are implicitly assuming that $\mathcal{E}(\tilde{G}_{\theta_{1}%
,\beta_{1}}+\tilde{H}_{\theta_{2},\beta_{2}})$ is a strictly positive true
martingale. Then, by Girsanov's theorem for semimartingales (Thm. 1 and 3, p.
702 and 703 in Shiryaev \cite{Sh99}), the process $X(t)$ and $Y(t)$ become
\begin{align}
X(t)  &  =X(0)+B_{Q_{\bar{\theta},\bar{\beta}}}^{X}(t)+\sigma_{X}%
W_{Q_{\bar{\theta},\bar{\beta}}}(t),\quad t\in\lbrack0,T],\nonumber\\
Y(t)  &  =Y(0)+B_{Q_{\bar{\theta},\bar{\beta}}}^{Y}(t)+\int_{0}^{t}\int
_{0}^{\infty}z\tilde{N}_{Q_{\bar{\theta},\bar{\beta}}}^{L}(ds,dz),\quad
t\in\lbrack0,T], \label{Equ_GirsanovY}%
\end{align}
with%
\begin{align}
B_{Q_{\bar{\theta},\bar{\beta}}}^{X}(t)  &  =\int_{0}^{t}(\mu_{X}+\theta
_{1}-\alpha_{X}(1-\beta_{1})X(s))ds,\quad t\in\lbrack
0,T],\label{Equ_X_Drift_Q}\\
B_{Q_{\bar{\theta},\bar{\beta}}}^{Y}(t)  &  =\int_{0}^{t}(\mu_{Y}+\kappa
_{L}^{\prime}(0)-\alpha_{Y}Y(s))ds+\int_{0}^{t}\int_{0}^{\infty}%
z(H_{\theta_{2},\beta_{2}}(s,z)-1)\ell(dz)ds\label{Equ_Y_Drift_Q}\\
&  =\int_{0}^{t}\{(\mu_{Y}+\kappa_{L}^{\prime}(0)-\alpha_{Y}Y(s))+\int
_{0}^{\infty}z(e^{\theta_{2}z}-1)\ell(dz)\nonumber\\
&  \qquad+\frac{\alpha_{Y}\beta_{2}}{\kappa_{L}^{\prime\prime}(\theta_{2}%
)}\int_{0}^{\infty}z^{2}e^{\theta_{2}z}\ell(dz)Y(s-)\}ds\nonumber\\
&  =\int_{0}^{t}\left(  \mu_{Y}+\kappa_{L}^{\prime}(\theta_{2})-\alpha
_{Y}(1-\beta_{2})Y(s)\right)  ds,\quad t\in\lbrack0,T],\nonumber
\end{align}
where $W_{Q_{\bar{\theta},\bar{\beta}}}$ is a $Q_{\bar{\theta},\bar{\beta}}%
$-standard Wiener process and the $Q_{\bar{\theta},\bar{\beta}}$-compensator
measure of $Y$ (and $L$) is
\[
v_{Q_{\bar{\theta},\bar{\beta}}}^{Y}(dt,dz)=v_{Q_{\bar{\theta},\bar{\beta}}%
}^{L}(dt,dz)=H_{\theta_{2},\beta_{2}}(t,z)\ell(dz)dt.
\]
In conclusion, the semimartingale triplet for $X$ and $Y$ under $Q_{\bar
{\theta},\bar{\beta}}$ are given by $(B_{Q_{\bar{\theta},\bar{\beta}}}%
^{X},\sigma_{X}^{2}t,0)$ and $(B_{Q_{\bar{\theta},\bar{\beta}}}^{Y}%
,0,v_{Q_{\bar{\theta},\bar{\beta}}}^{Y}),$ respectively.

\begin{remark}
\label{Remark_Dynamics_Q}Under $Q_{\bar{\theta},\bar{\beta}},$ $X$ and $Y$
still satisfy Langevin equations with different parameters, that is, the
measure change preserves the structure of the equations. The process $L$ is
not a L\'{e}vy process under $Q_{\bar{\theta},\bar{\beta}}$, but it remains a
semimartingale. Therefore, one can use It\^{o} formula again to obtain the
following explicit expressions for $X$ and $Y$%
\begin{align}
X(t)  &  =X(s)e^{-\alpha_{X}(1-\beta_{1})(t-s)}+\frac{\mu_{X}+\theta_{1}%
}{\alpha_{X}(1-\beta_{1})}(1-e^{-\alpha_{X}(1-\beta_{1})(t-s)}%
)\label{Equ_X_Explicit_Q}\\
&  \qquad+\sigma_{X}\int_{s}^{t}e^{-\alpha_{X}(1-\beta_{1})(t-u)}%
dW_{Q_{\bar{\theta},\bar{\beta}}}(u),\nonumber\\
Y(t)  &  =Y(s)e^{-\alpha_{Y}(1-\beta_{2})(t-s)}+\frac{\mu_{Y}+\kappa
_{L}^{\prime}(\theta_{2})}{\alpha_{Y}(1-\beta_{2})}(1-e^{-\alpha_{Y}%
(1-\beta_{2})(t-s)})\label{Equ_Y_Explicit_Q}\\
&  \qquad+\int_{s}^{t}\int_{0}^{\infty}e^{-\alpha_{Y}(1-\beta_{2}%
)(t-u)}z\tilde{N}_{Q_{\bar{\theta},\bar{\beta}}}^{L}(du,dz),\nonumber
\end{align}
where $0\leq s\leq t\leq T.$
\end{remark}

\begin{remark}
Looking at equations $\left(  \ref{Equ_X_Drift_Q}\right)  $ and $\left(
\ref{Equ_Y_Drift_Q}\right)  $, one can see how the values of the parameters
$\bar{\theta}$ and $\bar{\beta}$ change the drift. Setting $\bar{\theta
}=(0,0)$ we keep fixed the level to which the process reverts and change the
speed of mean reversion by changing $\bar{\beta}$. If $\bar{\beta}=(0,0)$ we
fix the speed of mean reversion and change the level by changing $\bar{\theta
}.$ By choosing $\beta_{1}=1$, say, we observe that $X(t)$ in
\eqref{Equ_X_Explicit_Q} becomes (using a limit consideration in the second
term)
\begin{equation}
X(t)=X(s)+(\mu_{X}+\theta_{1})(t-s)+\sigma_{X}(W_{Q_{\bar{\theta},\bar{\beta}%
}}(t)-W_{Q_{\bar{\theta},\bar{\beta}}}(s))\,.
\end{equation}
Hence, $X$ is a drifted Brownian motion and we have a non-stationary dynamics
under the pricing measure with this choice of $\beta_{1}$. Obviously, we can
choose $\beta_{2}=1$ and obtain similarly a non-stationary dynamics for the
jump component as well, however, this will not be driven by a L\'{e}vy process
under $Q_{\bar{\theta},\bar{\beta}}$.
\end{remark}

The previous reasonings rely crucially on the assumption that $Q_{\bar{\theta
},\bar{\beta}}$ is a probability measure. Hence, we have to find sufficient
conditions on the L\'{e}vy process $L$ and the possible values of the
parameters $\bar{\theta}$ and $\bar{\beta}$ that ensure $\mathcal{E}(\tilde
{G}_{\theta_{1},\beta_{1}}+\tilde{H}_{\theta_{2},\beta_{2}})$ to be a true
martingale with strictly positive values. As $[\tilde{G}_{\theta_{1},\beta
_{1}},\tilde{H}_{\theta_{2},\beta_{2}}],$ the quadratic co-variation between
$\tilde{G}_{\theta_{1},\beta_{1}}$ and $\tilde{H}_{\theta_{2},\beta_{2}},$ is
identically zero, by Yor's formula (equation II.8.19 in \cite{JaSh03}) we can
write%
\begin{equation}
\mathcal{E}(\tilde{G}_{\theta_{1},\beta_{1}}+\tilde{H}_{\theta_{2},\beta_{2}%
})(t)=\mathcal{E}(\tilde{G}_{\theta_{1},\beta_{1}})(t)\mathcal{E}(\tilde
{H}_{\theta_{2},\beta_{2}})(t),\quad t\in\lbrack0,T], \label{Equ_Yor_Formula}%
\end{equation}
and, as the stochastic exponential of a continuous process is always positive,
we just need to ensure the positivity of $\mathcal{E}(\tilde{H}_{\theta
_{2},\beta_{2}})(t).$ Assume that $\mathcal{E}(\tilde{H}_{\theta_{2},\beta
_{2}})$ is positive, then remark \ref{Remark_Stoch_Exp} yields that
$\mathcal{E}(\tilde{G}_{\theta_{1},\beta_{1}}+\tilde{H}_{\theta_{2},\beta_{2}%
})$ is a true martingale if and only if $\mathbb{E}_{P}[\mathcal{E}(\tilde
{G}_{\theta_{1},\beta_{1}}+\tilde{H}_{\theta_{2},\beta_{2}})(T)]=1.$ Using the
independence of $\tilde{G}_{\theta_{1},\beta_{1}}$ and $\tilde{H}_{\theta
_{2},\beta_{2}}$ and the identity $\left(  \ref{Equ_Yor_Formula}\right)  ,$ we
get
\[
\mathbb{E}_{P}[\mathcal{E}(\tilde{G}_{\theta_{1},\beta_{1}}+\tilde{H}%
_{\theta_{2},\beta_{2}})(T)]=\mathbb{E}_{P}[\mathcal{E}(\tilde{G}_{\theta
_{1},\beta_{1}})(T)]\mathbb{E}_{P}[\mathcal{E}(\tilde{H}_{\theta_{2},\beta
_{2}})(T)],
\]
showing that $\mathcal{E}(\tilde{G}_{\theta_{1},\beta_{1}}+\tilde{H}%
_{\theta_{2},\beta_{2}})$ is a martingale if and only if $\mathcal{E}%
(\tilde{G}_{\theta_{1},\beta_{1}})$ and $\mathcal{E}(\tilde{H}_{\theta
_{2},\beta_{2}})$ are also martingales. Hence, we can write%
\[
\left.  \frac{dQ_{\bar{\theta},\bar{\beta}}}{dP}\right\vert _{\mathcal{F}_{t}%
}=\left.  \frac{dQ_{\theta_{1},\beta_{1}}}{dP}\right\vert _{\mathcal{F}_{t}%
}\times\left.  \frac{dQ_{\theta_{2},\beta_{2}}}{dP}\right\vert _{\mathcal{F}%
_{t}},\quad t\in\lbrack0,T],
\]
where $\left.  \frac{dQ_{\theta_{1},\beta_{1}}}{dP}\right\vert _{\mathcal{F}%
_{t}}\triangleq\mathcal{E}(\tilde{G}_{\theta_{1},\beta_{1}})(t)$ and $\left.
\frac{dQ_{\theta_{2},\beta_{2}}}{dP}\right\vert _{\mathcal{F}_{t}}%
\triangleq\mathcal{E}(\tilde{H}_{\theta_{2},\beta_{2}})(t),t\in\lbrack0,T].$

The previous reasonings allow us to reduce the proof that $Q_{\bar{\theta
},\bar{\beta}}$ is a probability measure equivalent to $P,Q_{\bar{\theta}%
,\bar{\beta}}\sim P$, to prove that $\mathcal{E}(\tilde{G}_{\theta_{1}%
,\beta_{1}})$ is martingale (or $Q_{\theta_{1},\beta_{1}}\sim P$) and
$\mathcal{E}(\tilde{H}_{\theta_{2},\beta_{2}})$ is a martingale with strictly
positive values (or $Q_{\theta_{2},\beta_{2}}\sim P$). The literature on this
topic is huge, see for instance Kazamaki \cite{Ka77}, Novikov \cite{No80},
L\'{e}pingle and M\'{e}min \cite{LeMe78} and Kallsen and Shiryaev
\cite{KaSh02}. The main difficulty when trying to use the classical criteria
is that our kernels depend on the processes $X$ and $Y,$ which are unbounded.
To prove that $\mathcal{E}(\tilde{G}_{\theta_{1},\beta_{1}})$ is a martingale
one could use a localized version of Novikov's criterion. However, this
approach would entail to show that the expectation of the exponential of the
integral of a stochastic iterated integral of order two is finite. Although
these computations seem feasible, they are definitely very stodgy. On the
other hand, the most widely used sufficient criterion for martingales with
jumps is the L\'{e}pingle-M\'{e}min criterion. This criterion is very general
but the conditions obtained are far from optimal. Using this criterion we are
only able to prove the result by requiring the L\'{e}vy process $L$ to have
bounded jumps.

In a very recent paper, assuming some structure on the processes, Klebaner and
Lipster \cite{KleLi11} give a fairly general criterion which seems easier to
apply than those of Novikov and L\'{e}pingle-M\'{e}min. Although we can not
apply directly their criteria, at least not in the pure jump case, we can
reason similarly to prove the desired result for $\mathcal{E}(\tilde
{G}_{\theta_{1},\beta_{1}})$ and $\mathcal{E}(\tilde{H}_{\theta_{2},\beta_{2}%
}).$

Finally, note that these results can be extended, in a straightforward manner,
to any finite number of Langevin equations driven by Brownian motions and
L\'{e}vy processes, independent of each other. In the following two
subsections, we will drop the subindices in the parameters $\theta$ and
$\beta.$

\subsection{Brownian driven OU-process}

We first show that the process $\tilde{G}_{\theta,\beta}$ is a martingale
under $P$.

\begin{proposition}
\label{Prop_Gtilda_martingale}Let $\theta\in\mathbb{R}$ and $\beta\in
\lbrack0,1]$. Then, $\tilde{G}_{\theta,\beta}=\{\tilde{G}_{\theta,\beta
}(t)\}_{t\in\lbrack0,T]}$, defined by $\left(  \ref{Equ_G_tilda}\right)  ,$ is
a square integrable martingale under $P$.
\end{proposition}

\begin{proof}
We have to show that $G_{\theta,\beta}\in L^{2}(\Omega\times\lbrack
0,T];P\otimes\mathrm{Leb}).$ We get
\[
\mathbb{E}_{P}[\int_{0}^{T}G_{\theta,\beta}(t)^{2}dt]\leq2\sigma_{X}%
^{-2}\{\theta^{2}T+\alpha_{X}^{2}\mathbb{E}_{P}[\int_{0}^{T}X(t)^{2}dt]\}.
\]
By remark \ref{Remark_X_Gaussian} and the properties of the Gaussian
distribution, one has
\[
\mathbb{E}_{P}[\int_{0}^{T}X(t)^{2}dt]=\int_{0}^{T}(m_{t}^{2}+\Sigma_{t}%
^{2})dt\leq T\sup_{t\in\lbrack0,T]}(m_{t}^{2}+\Sigma_{t}^{2})<\infty,
\]
because $m_{t}$ and $\Sigma_{t}$ are continuous functions on $[0,T].$
\end{proof}

\begin{theorem}
\label{Theo_StochExp_Gtilda_martingale}Let $\theta\in\mathbb{R}$ and $\beta
\in\lbrack0,1]$. Then $\mathcal{E}(\tilde{G}_{\theta,\beta})=\{\mathcal{E}%
(\tilde{G}_{\theta,\beta})(t)\}_{t\in\lbrack0,T]}$ is a martingale under $P.$
\end{theorem}

\begin{proof}
As $\tilde{G}_{\theta,\beta}$ is a martingale with continuous paths, we have
that $\mathcal{E}(\tilde{G}_{\theta,\beta})$ is a positive local martingale.
By remark \ref{Remark_Stoch_Exp}, it suffices to prove that $\mathbb{E}%
_{P}[\mathcal{E}(\tilde{G}_{\theta,\beta})(T)]=1.$ Note that the sequence of
stopping times $\tau_{n}=\inf\{t:\mathcal{E}(\tilde{G}_{\theta,\beta
})>n\}\wedge T,n\geq1$ is a reducing sequence for $\mathcal{E}(\tilde
{G}_{\theta,\beta}).$ That is, $\tau_{n}$ converges a.s. to $T$ and, for every
$n\geq1$ fixed, the stopped process $\mathcal{E}(\tilde{G}_{\theta,\beta
})^{\tau_{n}}(t)\triangleq\mathcal{E}(\tilde{G}_{\theta,\beta})(t\wedge
\tau_{n})$ is a (bounded) martingale on $[0,T]$. Therefore, $\mathbb{E}%
_{P}[\mathcal{E}(\tilde{G}_{\theta,\beta})^{\tau_{n}}(T)]=\mathbb{E}%
_{P}[\mathcal{E}(\tilde{G}_{\theta,\beta})^{\tau_{n}}(0)]=1,n\geq1,$ and if we
show that
\begin{equation}
\lim_{n\rightarrow\infty}\mathbb{E}_{P}[\mathcal{E}(\tilde{G}_{\theta,\beta
})^{\tau_{n}}(T)]=\mathbb{E}_{P}[\mathcal{E}(\tilde{G}_{\theta,\beta})(T)]
\label{Equ_Conv_Expectations}%
\end{equation}
we will have finished. To show $\left(  \ref{Equ_Conv_Expectations}\right)  $
is equivalent to show the uniform integrability of the sequence of random
variable $\{\mathcal{E}(\tilde{G}_{\theta,\beta})^{\tau_{n}}(T)\}_{n\geq1},$
that is, to show
\[
\lim_{M\rightarrow\infty}\sup_{n\geq1}\mathbb{E}_{P}[\mathcal{E}(\tilde
{G}_{\theta,\beta})^{\tau_{n}}(T)\boldsymbol{1}_{\{\mathcal{E}(\tilde
{G}_{\theta,\beta})^{\tau_{n}}(T)>M\}}]=0.
\]
It is not difficult to prove that if $\Lambda(t)$ is a non-negative function
such that $\lim_{t\rightarrow\infty}\Lambda(t)/t=\infty$ and
\[
\sup_{n\geq1}\mathbb{E}_{P}[\Lambda(\mathcal{E}(\tilde{G}_{\theta,\beta
})^{\tau_{n}}(T))]<\infty,
\]
then $\{\mathcal{E}(\tilde{G}_{\theta,\beta})^{\tau_{n}}(T)\}_{n\geq1}$ is
uniformly integrable. We consider the test function $\Lambda(t)=1+t\log(t).$
Hence, it suffices to prove that
\begin{equation}
\sup_{n\geq1}\mathbb{E}_{P}[\mathcal{E}(\tilde{G}_{\theta,\beta})^{\tau_{n}%
}(T)\log(\mathcal{E}(\tilde{G}_{\theta,\beta})^{\tau_{n}}(T))]<\infty.
\label{Equ_UI_Stoch_Exp_Gtau_n}%
\end{equation}
Note that we can use the sequence of martingales on $[0,T]$ given by
$\{\mathcal{E}(\tilde{G}_{\theta,\beta})^{\tau_{n}}\}_{n\geq1}$ to define a
sequence of probability measures $\{Q_{\theta,\beta}^{n}\}_{n\geq1}$ with
Radon-Nykodim densities given by $\left.  \frac{dQ_{\theta,\beta}^{n}}%
{dP}\right\vert _{\mathcal{F}_{t}}\triangleq\mathcal{E}(\tilde{G}%
_{\theta,\beta})^{\tau_{n}}(t),t\in\lbrack0,T],n\geq1.$ In addition, one has
that
\begin{align}
\mathcal{E}(\tilde{G}_{\theta,\beta})^{\tau_{n}}(t)  &  =\exp\left(  \int
_{0}^{t\wedge\tau_{n}}G_{\theta,\beta}(s)dW(s)-\frac{1}{2}\int_{0}%
^{t\wedge\tau_{n}}G_{\theta,\beta}(s)^{2}ds\right)
\label{Equ_Stoch_Exp_Gtau_n}\\
&  =\exp\left(  \int_{0}^{t}\boldsymbol{1}_{[0,\tau_{n}]}(s)G_{\theta,\beta
}(s)dW(s)-\frac{1}{2}\int_{0}^{t}(\boldsymbol{1}_{[0,\tau_{n}]}(s)G_{\theta
,\beta}(s))^{2}ds\right) \nonumber\\
&  =\mathcal{E}(\tilde{G}_{\theta,\beta}^{n})(t),\quad t\in\lbrack
0,T],n\geq1,\nonumber
\end{align}
where $\tilde{G}_{\theta,\beta}^{n}(t)\triangleq\int_{0}^{t}\boldsymbol{1}%
_{[0,\tau_{n}]}(s)G_{\theta,\beta}(s)dW(s),t\in\lbrack0,T],n\geq1.$ On the
other hand, from $\left(  \ref{Equ_Stoch_Exp_Gtau_n}\right)  ,$ we have the
trivial bound $\log(\mathcal{E}(\tilde{G}_{\theta,\beta})^{\tau_{n}}%
(T))\leq\tilde{G}_{\theta,\beta}^{\tau_{n}}(T).$ Combining the last bound with
the change of measure given by $\{Q_{\theta,\beta}^{n}\}_{n\geq1}$ we get
that
\begin{equation}
\sup_{n\geq1}\mathbb{E}_{Q_{\theta,\beta}^{n}}[\tilde{G}_{\theta,\beta}%
^{\tau_{n}}(T)]<\infty, \label{Equ_UnifBoundGTilda}%
\end{equation}
implies that $\left(  \ref{Equ_UI_Stoch_Exp_Gtau_n}\right)  $ holds. Applying
Girsanov's Theorem, we can write
\[
\tilde{G}_{\theta,\beta}^{\tau_{n}}(T)=\int_{0}^{T}\boldsymbol{1}_{[0,\tau
_{n}]}(t)(G_{\theta,\beta}(t))^{2}dt+\int_{0}^{T}\boldsymbol{1}_{[0,\tau_{n}%
]}(t)G_{\theta,\beta}(t)dW_{Q_{\theta,\beta}^{n}}(t),
\]
where $W_{Q_{\theta,\beta}^{n}}$ is a $Q_{\theta,\beta}^{n}$-Brownian motion.
Therefore, it suffices to prove that
\begin{equation}
\sup_{n\geq1}\mathbb{E}_{Q_{\theta,\beta}^{n}}[\int_{0}^{T}\boldsymbol{1}%
_{[0,\tau_{n}]}(t)(G_{\theta,\beta}(t))^{2}dt]<\infty,
\label{EquSupBound_G_Square}%
\end{equation}
because this imply that $\int_{0}^{T\wedge\tau_{n}}G_{\theta,\beta
}(t)dW_{Q_{\theta,\beta}^{n}}(t)$ is a $Q_{\theta,\beta}^{n}$-martingale with
zero expectation and, in passing, that $\left(  \ref{Equ_UnifBoundGTilda}%
\right)  $ holds. Now we proceed as in the proof of Proposition
\ref{Prop_Gtilda_martingale}. We have that%
\[
\mathbb{E}_{Q_{\theta,\beta}^{n}}[\int_{0}^{T}\boldsymbol{1}_{[0,\tau_{n}%
]}(t)(G_{\theta,\beta}(t))^{2}dt]\leq2\sigma_{X}^{-2}\{\theta^{2}T+\alpha
_{X}^{2}\mathbb{E}_{Q_{\theta,\beta}^{n}}[\int_{0}^{T}\boldsymbol{1}%
_{[0,\tau_{n}]}(t)X(t)^{2}dt]\},
\]
but now the term with $X(t)^{2}$ is more delicate to treat. Using Remark
\ref{Remark_X_Gaussian}, we know that $X(t)$ conditioned to $\tau_{n}$ is
Gaussian, but we do not know the distribution of $\tau_{n}$ and, hence, a
direct computation of $\mathbb{E}_{Q_{\theta,\beta}^{n}}[\boldsymbol{1}%
_{[0,\tau_{n}]}(t)X(t)^{2}]$ is not possible. However, we have that
\begin{align*}
&  \mathbb{E}_{Q_{\theta,\beta}^{n}}[\int_{0}^{T}\boldsymbol{1}_{[0,\tau_{n}%
]}(t)X(t)^{2}dt]\\
&  \leq2\{\mathbb{E}_{Q_{\theta,\beta}^{n}}[\int_{0}^{T}\boldsymbol{1}%
_{[0,\tau_{n}]}(t)\left(  X(0)e^{-\alpha_{X}(1-\beta)t}+\frac{\mu_{X}+\theta
}{\alpha_{X}(1-\beta)}\left(  1-e^{-\alpha_{X}(1-\beta)t}\right)  \right)
^{2}dt]\\
&  \qquad+\sigma_{X}^{2}\mathbb{E}_{Q_{\theta,\beta}^{n}}[\int_{0}%
^{T}\boldsymbol{1}_{[0,\tau_{n}]}(t)\left(  \int_{0}^{t}e^{-\alpha_{X}%
(1-\beta)(t-u)}dW_{Q_{\theta,\beta}^{n}}(u)\right)  ^{2}dt]\}\\
&  \leq2T\{\left(  \left\vert X(0)\right\vert +\left(  \left\vert \mu
_{X}\right\vert +\left\vert \theta\right\vert \right)  T\right)  ^{2}%
+\sigma_{X}^{2}T\}<\infty,
\end{align*}
where we have used that the function $\eta(x)\triangleq(1-e^{-xa})/x\leq a$
for $x,a\geq0,$ and that
\[
\mathbb{E}_{Q_{\theta,\beta}^{n}}\left[  \left(  \int_{0}^{t}e^{-\alpha
_{X}(1-\beta)(t-u)}dW_{Q_{\theta,\beta}^{n}}(u)\right)  ^{2}\right]  =\int
_{0}^{t}e^{-2\alpha_{X}(1-\beta)(t-u)}du\leq T.
\]
Hence, we have shown $\left(  \ref{EquSupBound_G_Square}\right)  $ and the
result follows.
\end{proof}

\subsection{L\'{e}vy driven OU-processes}

First we will prove that $\tilde{H}_{\theta,\beta}$ is a square integrable martingale.

\begin{proposition}
\label{Prop_Htilda_martingale}Let $\theta\in D_{L},\beta\in\lbrack0,1]$. Then
$\tilde{H}_{\theta,\beta}=\{\tilde{H}_{\theta,\beta}(t)\}_{t\in\lbrack0,T]},$
defined by $(\ref{Equ_H_tilda})$, is a square integrable martingale under $P.$
\end{proposition}

\begin{proof}
According to Ikeda-Watanabe \cite{IkWa81}, p. 59-63, we have to check that
$\mathbb{E}_{P}[\int_{0}^{T}\int_{0}^{\infty}|H_{\theta,\beta}(s,z)-1|^{2}%
\ell(dz)dt]\allowbreak<\infty.$ We can write%
\begin{align*}
\mathbb{E}_{P}[\int_{0}^{T}\int_{0}^{\infty}|H_{\theta,\beta}(s,z)-1|^{2}%
\ell(dz)dt]  &  \leq T\int_{0}^{\infty}|e^{\theta z}-1|^{2}\ell(dz)\\
&  \qquad+\frac{\alpha_{Y}^{2}}{\left(  \kappa_{L}^{\prime\prime}%
(\theta)\right)  ^{2}}\int_{0}^{\infty}e^{2\theta z}z^{2}\ell(dz)\int_{0}%
^{T}\mathbb{E}_{P}[|Y(t)|^{2}]dt.
\end{align*}
By the mean value theorem in integral form we have that $|e^{\theta z}%
-1|^{2}=|\theta z\int_{0}^{1}e^{\lambda\theta z}d\lambda|^{2}\leq\theta
^{2}z^{2}e^{(2\theta\vee0)z}.$ Hence, as $\theta\in D_{L},$%
\[
\int_{0}^{\infty}|e^{\theta z}-1|^{2}\ell(dz)\leq\theta^{2}\int_{0}^{\infty
}z^{2}e^{2\theta z}\ell(dz)=\theta^{2}\kappa_{L}^{\prime\prime}(2\theta
\vee0)<\infty.
\]
Therefore, the result follows by showing that $\sup_{t\in\lbrack
0,T]}\mathbb{E}_{P}[\left\vert Y(t)\right\vert ^{2}]\allowbreak<\infty.$ We
have that%
\begin{align*}
\sup_{t\in\lbrack0,T]}\mathbb{E}_{P}[|Y(t)|^{2}]  &  \leq2\sup_{t\in
\lbrack0,T]}\{\left(  Y(0)e^{-\alpha_{Y}t}+\frac{\mu_{Y}+\kappa_{L}^{\prime
}(0)}{\alpha_{Y}}(1-e^{-\alpha_{Y}t})\right)  ^{2}\\
&  \qquad+\mathbb{E}_{P}[\left(  \int_{0}^{t}\int_{0}^{\infty}ze^{-\alpha
_{Y}(t-s)}\tilde{N}^{L}(ds,dz)\right)  ^{2}]\}\\
&  \leq\left(  Y(0)+\frac{\mu_{Y}+\kappa_{L}^{\prime}(0)}{\alpha_{Y}}\right)
^{2}+\sup_{t\in\lbrack0,T]}\int_{0}^{t}\int_{0}^{\infty}z^{2}e^{-2\alpha
_{Y}(t-s)}\ell(dz)ds\\
&  \leq\left(  Y(0)+\frac{\mu_{Y}+\kappa_{L}^{\prime}(0)}{\alpha_{Y}}\right)
^{2}+T\kappa_{L}^{\prime\prime}(0)<\infty.
\end{align*}

\end{proof}

Note that the stochastic exponential $\mathcal{E}(\tilde{H}_{\theta,\beta})$
satisfies the following SDE%
\[
\mathcal{E}(\tilde{H}_{\theta,\beta})(t)=1+\int_{0}^{t}\mathcal{E}(\tilde
{H}_{\theta,\beta})(s-)d\tilde{H}_{\theta,\beta}(s)=1+\int_{0}^{t}\int
_{0}^{\infty}\mathcal{E}(\tilde{H}_{\theta,\beta})(s-)\left(  \tilde
{H}_{\theta,\beta}(s,z)-1\right)  \tilde{N}^{L}(ds,dz),
\]
and it can be represented explicitly as
\begin{align}
\mathcal{E}(\tilde{H}_{\theta,\beta})(t)  &  =e^{\tilde{H}_{\theta,\beta}%
(t)}\prod_{0<s\leq t}(1+\Delta\tilde{H}_{\theta,\beta}(s))e^{-\Delta\tilde
{H}_{\theta,\beta}(s)}\label{Equ_StochExp_Explicit}\\
&  =\exp\left(  \tilde{H}_{\theta,\beta}(t)-\sum_{0\leq s\leq t}\Delta
\tilde{H}_{\theta,\beta}(s)-\log(1+\Delta\tilde{H}_{\theta,\beta}(s))\right)
,\quad t\in\lbrack0,T].\nonumber
\end{align}
\qquad Hence, a necessary and sufficient condition for the positivity of
$\mathcal{E}(\tilde{H}_{\theta,\beta})$ is that $\Delta\tilde{H}_{\theta
,\beta}>-1,$ up to an evanescent set. Moreover, by the definition of
$\tilde{H}_{\theta,\beta}(t)$ and $H_{\theta,\beta}(t,z)$ we have that
\begin{equation}
\Delta\tilde{H}_{\theta,\beta}(t)=H_{\theta,\beta}(t,\Delta L(t))-1=(e^{\theta
\Delta L(t)}-1)+\frac{\alpha_{Y}\beta}{\kappa_{L}^{\prime\prime}(\theta
)}\Delta L(t)e^{\theta\Delta L(s)}Y(t-),\quad t\in\lbrack0,T],
\label{EquJumpsHtilda}%
\end{equation}
which yields the condition
\begin{equation}
P(\frac{\alpha_{Y}\beta}{\kappa_{L}^{\prime\prime}(\theta)}(\Delta
L(t))Y(t-)>-1,t\in\lbrack0,T])=1. \label{Equ_PositivitySExp}%
\end{equation}

\begin{remark}
\label{Remark_Positivity_of_StocExpHtilda}As we assume that $L$ is a
subordinator and $Y(0)\geq0$ and $\mu\geq0$, we have that $P(Y(t)\geq
0,t\in\lbrack0,T])=1$, condition $\left(  \ref{Equ_PositivitySExp}\right)  $
is automatically satisfied and $\mathcal{E}(\tilde{H}_{\theta,\beta}),$ is
strictly positive.
\end{remark}

\begin{theorem}
\label{Theo_StochExp_Htilda_martingale}Let $\theta\in D_{L}$ and $\beta
\in\lbrack0,1].$ Then $\mathcal{E}(\tilde{H}_{\theta,\beta})=\{\mathcal{E}%
(\tilde{H}_{\theta,\beta})(t)\}_{t\in\lbrack0,T]}$ is a martingale under $P$.
\end{theorem}

\begin{proof}
As $\tilde{H}_{\theta,\beta}$ is a martingale on $[0,T]$, we have that
$\mathcal{E}(\tilde{H}_{\theta,\beta}),$ is a local martingale on $[0,T].$
Hence, there exists a sequence of increasing stopping times such that
$\tau_{n}\uparrow T,P$-$a.s.$ and the stopped processes $\mathcal{E}(\tilde
{H}_{\theta,\beta})^{\tau_{n}},n\geq1$ are martingales on $[0,T]$. By Remark
\ref{Remark_Stoch_Exp} and the same reasonings as in the proof of Theorem
\ref{Theo_StochExp_Gtilda_martingale}, to show that $\mathcal{E}(\tilde
{H}_{\theta,\beta})$ is a martingale is equivalent to show that $\mathbb{E}%
[\mathcal{E}(\tilde{H}_{\theta,\beta})(T)]=1$ and this is equivalent to prove
that the sequence $\{\mathcal{E}(\tilde{H}_{\theta,\beta})^{\tau_{n}%
}(T)\}_{n\geq1}$ is uniformly integrable. A sufficient condition for the
uniform integrability of $\{\mathcal{E}(\tilde{H}_{\theta,\beta})^{\tau_{n}%
}(T)\}_{n\geq1}$ is given by%
\begin{equation}
\sup_{n\geq1}\mathbb{E}_{P}[\mathcal{E}(\tilde{H}_{\theta,\beta})^{\tau_{n}%
}(T)\log(\mathcal{E}(\tilde{H}_{\theta,\beta})^{\tau_{n}}(T))]<\infty.
\label{Equ_Unif_Int_StochExp_Htilda}%
\end{equation}
By equation $\left(  \ref{Equ_StochExp_Explicit}\right)  $, we get%
\[
\log(\mathcal{E}(\tilde{H}_{\theta,\beta})^{\tau_{n}}(T))\leq\tilde{H}%
_{\theta,\beta}^{\tau_{n}}(T)-\sum_{0\leq t\leq\tau_{n}\wedge T}\Delta
\tilde{H}_{\theta,\beta}(t)-\log(1+\Delta\tilde{H}_{\theta,\beta}%
(t))\leq\tilde{H}_{\theta,\beta}^{\tau_{n}}(T),
\]
because the function $x-\log(1+x)\geq0$ for $x>-1.$ Hence, we can write%
\begin{align}
&  \mathbb{E}_{P}[\mathcal{E}(\tilde{H}_{\theta,\beta})^{\tau_{n}}(T)\tilde
{H}_{\theta,\beta}^{\tau_{n}}(T)]\nonumber\\
&  \qquad=\mathbb{E}_{P}[\left(  1+\int_{0}^{T\wedge\tau_{n}}\mathcal{E}%
(\tilde{H}_{\theta,\beta})(t-)d\tilde{H}_{\theta,\beta}(t)\right)  \tilde
{H}_{\theta,\beta}^{\tau_{n}}(T)]\nonumber\\
&  \qquad=\mathbb{E}_{P}[\left(  1+\int_{0}^{T}\mathcal{E}(\tilde{H}%
_{\theta,\beta})^{\tau_{n}}(t-)d\tilde{H}_{\theta,\beta}^{\tau_{n}}(t)\right)
\tilde{H}_{\theta,\beta}^{\tau_{n}}(T)]\nonumber\\
&  \qquad=\mathbb{E}_{P}[\tilde{H}_{\theta,\beta}^{\tau_{n}}(T)]+\mathbb{E}%
_{P}[\left(  \int_{0}^{T}\mathcal{E}(\tilde{H}_{\theta,\beta})^{\tau_{n}%
}(t-)d\tilde{H}_{\theta,\beta}^{\tau_{n}}(t)\right)  \left(  \int_{0}%
^{T}\boldsymbol{1}_{[0,\tau_{n}]}(t)d\tilde{H}_{\theta,\beta}^{\tau_{n}%
}(t)\right)  ]\nonumber\\
&  \qquad=\int_{0}^{T}\int_{0}^{\infty}\mathbb{E}_{P}[\boldsymbol{1}%
_{[0,\tau_{n}]}(t)\mathcal{E}(\tilde{H}_{\theta,\beta})^{\tau_{n}}(t)\left(
e^{\theta z}-1+\frac{\alpha_{Y}\beta}{\kappa_{L}^{\prime\prime}(\theta
)}e^{\theta z}zY(t)\right)  ^{2}]\ell(dz)dt\nonumber\\
&  \qquad=\mathbb{E}_{P}[\mathcal{E}(\tilde{H}_{\theta,\beta})^{\tau_{n}%
}(T)\int_{0}^{T}\int_{0}^{\infty}\boldsymbol{1}_{[0,\tau_{n}]}(t)\left(
e^{\theta z}-1+\frac{\alpha_{Y}\beta}{\kappa_{L}^{\prime\prime}(\theta
)}e^{\theta z}zY(t)\right)  ^{2}\ell(dz)dt]\nonumber\\
&  \qquad\leq2T\int_{0}^{\infty}\left\vert e^{\theta z}-1\right\vert ^{2}%
\ell(dz)+2\frac{\alpha_{Y}^{2}\kappa_{L}^{\prime\prime}(2\theta)}{(\kappa
_{L}^{\prime\prime}(\theta))^{2}}\mathbb{E}_{P}[\mathcal{E}(\tilde{H}%
_{\theta,\beta})^{\tau_{n}}(T)\int_{0}^{T\wedge\tau_{n}}Y(t)^{2}dt],
\label{Equ_Finite_Terms}%
\end{align}
where we have used that for any stopping time $\tau\leq T$ the process
$\tilde{H}_{\theta,\beta}^{\tau}(T)$ is a $P$-martingale with zero
expectation. In addition, we have used that $\forall n\geq1$\ fixed,
$\mathbb{E}_{P}[\mathcal{E}(\tilde{H}_{\theta,\beta})^{\tau_{n}}(T)]=1$ and%
\begin{align*}
&  \mathbb{E}_{P}[\mathbb{E}_{P}[\mathcal{E}(\tilde{H}_{\theta,\beta}%
)^{\tau_{n}}(T)\boldsymbol{1}_{[0,\tau_{n}]}(t)\left(  e^{\theta z}%
-1+\frac{\alpha_{Y}\beta}{\kappa_{L}^{\prime\prime}(\theta)}e^{\theta
z}zY(t)\right)  ^{2}|\mathcal{F}_{t}]]\\
&  \qquad=\mathbb{E}_{P}[\boldsymbol{1}_{[0,\tau_{n}]}(t)\mathbb{E}%
_{P}[\mathcal{E}(\tilde{H}_{\theta,\beta})^{\tau_{n}}(T)|\mathcal{F}%
_{t}]\left(  e^{\theta z}-1+\frac{\alpha_{Y}\beta}{\kappa_{L}^{\prime\prime
}(\theta)}e^{\theta z}zY(t)\right)  ^{2}]\\
&  \qquad=\mathbb{E}_{P}[\boldsymbol{1}_{[0,\tau_{n}]}(t)\mathcal{E}(\tilde
{H}_{\theta,\beta})^{\tau_{n}}(t)\left(  e^{\theta z}-1+\frac{\alpha_{Y}\beta
}{\kappa_{L}^{\prime\prime}(\theta)}e^{\theta z}zY(t)\right)  ^{2}],
\end{align*}
because $\tau_{n}$ is a reducing sequence for the local martingale
$\mathcal{E}(\tilde{H}_{\theta,\beta}).$ One can reason as in the proof of
Proposition \ref{Prop_Htilda_martingale} to show that the terms $\int
_{0}^{\infty}\left\vert e^{\theta z}-1\right\vert ^{2}\ell(dz)$ and
$\kappa_{L}^{\prime\prime}(2\theta),$ in equation $\left(
\ref{Equ_Finite_Terms}\right)  ,$ are finite. Note that $\int_{0}^{T\wedge
\tau_{n}}Y(t)^{2}dt=\int_{0}^{T\wedge\tau_{n}}Y(t\wedge\tau_{n})^{2}dt\leq
\int_{0}^{T}Y^{\tau_{n}}(t)^{2}dt,$ thus, it just remains to prove that
\[
\sup_{n\geq1}\mathbb{E}_{P}[\mathcal{E}(\tilde{H}_{\theta,\beta})^{\tau_{n}%
}(T)\int_{0}^{T}Y^{\tau_{n}}(t)^{2}dt]<\infty,
\]
to finish the proof. As $\mathcal{E}(\tilde{H}_{\theta,\beta})^{\tau_{n}}$ is
a strictly positive martingale, by Remark
\ref{Remark_Positivity_of_StocExpHtilda}, we can define the probability
measure $Q_{\theta,\beta}^{n}\sim P$ by setting $\left.  \frac{dQ_{\theta
,\beta}^{n}}{dP}\right\vert _{\mathcal{F}_{t}}\triangleq\mathcal{E}(\tilde
{H}_{\theta,\beta})^{\tau_{n}}(t),t\in\lbrack0,T],$ and, hence, it suffices to
prove that $\sup_{n\geq1}\mathbb{E}_{Q_{\theta,\beta}^{n}}[\int_{0}^{T}%
Y^{\tau_{n}}(t)^{2}dt]<\infty.$ Using Girsanov's Theorem with $Q_{\theta
,\beta}^{n}\sim P,n\geq1,$ the process $Y^{\tau_{n}}$ can be written as
\[
Y^{\tau_{n}}(t)=Y(0)+\tilde{B}^{\tau_{n}}(t)+\int_{0}^{t}\int_{0}^{\infty
}\boldsymbol{1}_{[0,\tau_{n}]}(s)z\tilde{N}_{Q_{\theta,\beta}^{n}}%
^{L}(ds,dz)\quad t\in\lbrack0,T],
\]
where%
\begin{align*}
\tilde{B}^{\tau_{n}}(t)  &  =\int_{0}^{t}\boldsymbol{1}_{[0,\tau_{n}]}%
(s)(\mu_{Y}+\kappa_{L}^{\prime}(0)-\alpha_{Y}Y(s))ds+\int_{0}^{t}%
\int_{\mathbb{R}}z\boldsymbol{1}_{[0,\tau_{n}]}(s)(H_{\theta,\beta
}(s,z)-1)\ell(dz)ds\\
&  =\int_{0}^{t}\boldsymbol{1}_{[0,\tau_{n}]}(s)\{(\mu_{Y}+\kappa_{L}^{\prime
}(0)-\alpha_{Y}Y(s))+\int_{\mathbb{R}}z(e^{\theta z}-1)\ell(dz)+\frac
{\alpha_{Y}\beta}{\kappa_{L}^{\prime\prime}(\theta)}\int_{\mathbb{R}}%
z^{2}e^{\theta z}\ell(dz)Y(s)\}ds\\
&  =\int_{0}^{t}\boldsymbol{1}_{[0,\tau_{n}]}(s)\left(  \mu_{Y}+\kappa
_{L}^{\prime}(\theta)-\alpha_{Y}(1-\beta)Y(s)\right)  ds,\quad t\in
\lbrack0,T],
\end{align*}
and $\tilde{N}_{Q_{\theta,\beta}^{n}}^{L}(ds,dz)$ is the compensated version
of the random measure $N_{Q_{\theta,\beta}^{n}}^{L}(ds,dz)$ with
$Q_{\theta,\beta}^{n}$-compensator given by $\tilde{\nu}_{Q_{\theta,\beta}%
^{n}}^{L}(ds,dz)=\{\boldsymbol{1}_{[0,\tau_{n}]}(s)(H_{\theta,\beta
}(s,z)-1)+1\}\ell(dz)ds.$ Hence,%
\begin{align}
\mathbb{E}_{Q_{\theta,\beta}^{n}}[\left(  Y^{\tau_{n}}(t)\right)  ^{2}]  &
\leq4\{Y(0)^{2}+\mathbb{E}_{Q_{\theta,\beta}^{n}}[\left(  \int_{0}%
^{t}\boldsymbol{1}_{[0,\tau_{n}]}(s)(\mu_{Y}+\kappa_{L}^{\prime}%
(\theta)+\alpha_{Y}(1-\beta)Y(s))ds\right)  ^{2}]\nonumber\\
&  \qquad+\mathbb{E}_{Q_{\theta,\beta}^{n}}[\left(  \int_{0}^{t}\int
_{0}^{\infty}\boldsymbol{1}_{[0,\tau_{n}]}(s)z\tilde{N}_{Q_{\theta,\beta}^{n}%
}^{L}(ds,dz)\right)  ^{2}]\}\nonumber\\
&  \leq4\{Y(0)^{2}+T\mathbb{E}_{Q_{\theta,\beta}^{n}}[\int_{0}^{t}%
\boldsymbol{1}_{[0,\tau_{n}]}(s)(\mu_{Y}+\kappa_{L}^{\prime}(\theta
)+\alpha_{Y}(1-\beta)Y^{\tau_{n}}(s))^{2}ds]\nonumber\\
&  \qquad+\mathbb{E}_{Q_{\theta,\beta}^{n}}[\int_{0}^{t}\int_{0}^{\infty
}\boldsymbol{1}_{[0,\tau_{n}]}(s)z^{2}\{\boldsymbol{1}_{[0,\tau_{n}%
]}(s)(H_{\theta,\beta}(s,z)-1)+1\}\ell(dz)ds]\}.\nonumber
\end{align}
On the one hand,%
\begin{align*}
&  \mathbb{E}_{Q_{\theta,\beta}^{n}}[\int_{0}^{t}\boldsymbol{1}_{[0,\tau_{n}%
]}(s)(\mu_{Y}+\kappa_{L}^{\prime}(\theta)+\alpha_{Y}(1-\beta)Y^{\tau_{n}%
}(s))^{2}ds]\\
&  \qquad\qquad\leq2T(\mu_{Y}+\kappa_{L}^{\prime}(\theta))^{2}+2\alpha_{Y}%
^{2}\int_{0}^{t}\mathbb{E}_{Q_{\theta,\beta}^{n}}[\left(  Y^{\tau_{n}%
}(s)\right)  ^{2}]ds.
\end{align*}
On the other hand,%
\begin{align*}
&  \mathbb{E}_{Q_{\theta,\beta}^{n}}[\int_{0}^{t}\int_{0}^{\infty
}\boldsymbol{1}_{[0,\tau_{n}]}(s)z^{2}\{\boldsymbol{1}_{[0,\tau_{n}%
]}(s)(H_{\theta,\beta}(s,z)-1)+1\}\ell(dz)ds]\\
&  \qquad=\mathbb{E}_{Q_{\theta,\beta}^{n}}[\int_{0}^{t}\int_{0}^{\infty
}\boldsymbol{1}_{[0,\tau_{n}]}(s)z^{2}H_{\theta,\beta}(s,z)\ell(dz)ds]\\
&  \qquad=\mathbb{E}_{Q_{\theta,\beta}^{n}}[\int_{0}^{t}\int_{0}^{\infty
}\boldsymbol{1}_{[0,\tau_{n}]}(s)z^{2}\left(  e^{\theta z}+\frac{\alpha
_{Y}\beta}{\kappa_{L}^{\prime\prime}(\theta)}e^{\theta z}zY(s-)\right)
\ell(dz)ds]\\
&  \qquad\leq T\int_{0}^{\infty}z^{2}e^{\theta z}\ell(dz)+\mathbb{E}%
_{Q_{\theta,\beta}^{n}}[\int_{0}^{t}\int_{0}^{\infty}\boldsymbol{1}%
_{[0,\tau_{n}]}(s)\frac{\alpha_{Y}\beta}{\kappa_{L}^{\prime\prime}(\theta
)}e^{\theta z}z^{3}Y^{\tau_{n}}(s)\ell(dz)ds]\\
&  \qquad\leq T\kappa_{L}^{\prime\prime}(\theta)+\frac{\alpha_{Y}\beta}%
{\kappa_{L}^{\prime\prime}(\theta)}\int_{0}^{\infty}z^{3}e^{\theta z}%
\ell(dz)\int_{0}^{t}\mathbb{E}_{Q_{\theta,\beta}^{n}}[Y^{\tau_{n}}(s)]ds\\
&  \qquad\leq T\kappa_{L}^{\prime\prime}(\theta)+\frac{\alpha_{Y}\kappa
_{L}^{(3)}(\theta)}{\kappa_{L}^{\prime\prime}(\theta)}\int_{0}^{t}%
\mathbb{E}_{Q_{\theta,\beta}^{n}}[\left(  Y^{\tau_{n}}(s)\right)  ^{2}]ds.
\end{align*}
To sum up, $\mathbb{E}_{Q_{\theta,\beta}^{n}}[\left(  Y^{\tau_{n}}(t)\right)
^{2}]\leq C_{0}+C_{1}\int_{0}^{t}\mathbb{E}_{Q_{\theta,\beta}^{n}}[\left(
Y^{\tau_{n}}(s)\right)  ^{2}]ds,$ where%
\begin{align*}
C_{0}  &  =C_{0}(Y(0),\mu_{Y},\theta,T)\triangleq4Y(0)^{2}+8T^{2}(\mu
_{Y}+\kappa_{L}^{\prime}(\theta))^{2}+4T\kappa_{L}^{\prime\prime}(\theta),\\
C_{1}  &  =C_{1}(\alpha_{Y},T)\triangleq8T\alpha_{Y}^{2}+4\frac{\alpha
_{Y}\kappa_{L}^{(3)}(\theta)}{\kappa_{L}^{\prime\prime}(\theta)},
\end{align*}
and applying Gronwall's lemma to the function $\mathbb{E}_{Q_{\theta,\beta
}^{n}}[Y^{\tau_{n}}(t)^{2}],$ we get that
\begin{equation}
\mathbb{E}_{Q_{\theta,\beta}^{n}}[Y^{\tau_{n}}(t)^{2}]\leq C_{0}e^{C_{1}T}.
\label{EquGronwall}%
\end{equation}
Finally, using Fubini-Tonelli and inequality $\left(  \ref{EquGronwall}%
\right)  $ we obtain%
\[
\sup_{n\geq1}\mathbb{E}_{Q_{\theta,\beta}^{n}}[\int_{0}^{T}Y^{\tau_{n}}%
(t)^{2}dt]\leq\sup_{n\geq1}\int_{0}^{T}\mathbb{E}_{Q_{\theta,\beta}^{n}%
}[Y^{\tau_{n}}(t)^{2}]dt\leq TC_{0}e^{C_{1}T}<\infty,
\]
and the proof is finished.
\end{proof}

\begin{remark}
\label{Remark_AlternativeChange}If $L$ has finite activity, that is
$\ell((0,\infty))<\infty,$ then one can use the kernel
\[
M_{\theta,\beta}(t,z)\triangleq e^{\theta z}\left(  1+\frac{\alpha_{Y}\beta
}{\kappa_{L}^{\prime}(\theta)}Y(t-)\right)  ,\quad t\in\lbrack0,T],z\in
\mathbb{R},
\]
and the Poisson integral
\[
\tilde{M}_{\theta,\beta}(t)\triangleq\int_{0}^{t}\int_{0}^{\infty}%
(M_{\theta,\beta}(s,z)-1)\tilde{N}^{L}(ds,dz),
\]
to define the change of measure. The results in Proposition
\ref{Prop_Htilda_martingale} and Theorem \ref{Theo_StochExp_Htilda_martingale}%
, below, also hold. Note that the change of measure with $\tilde{M}%
_{\theta,\beta}$ does not work for the infinite activity case. This is
because, in the analogous proofs of the statements in Proposition
\ref{Prop_Htilda_martingale} and Theorem \ref{Theo_StochExp_Htilda_martingale}
using the change of measure induced by $\tilde{M}_{\theta,\beta},$ it appears
the integral $\int_{0}^{\infty}e^{2\theta z}\ell(dz),$ which is divergent if
$\ell((0,\infty))=\infty.$
\end{remark}

\section{Study of the risk premium\label{SecStudyOfRiskPrem}}

We are interested in applying the previous probability measure change to study
the risk premium in electricity markets. As we discussed in the Introduction,
there are two reasonable models for the spot price $S$ in this market: the
arithmetic and the exponential model. We define the \textit{arithmetic spot
price model} by%
\begin{equation}
\label{Equ_Arith_Model}S(t)=\Lambda_{a}(t)+X(t)+Y(t),\quad t\in\lbrack
0,T^{\ast}],
\end{equation}
and the \textit{geometric spot price model} by
\begin{equation}
\label{Equ_Geom_Model}S(t)=\Lambda_{g}(t)\exp(X(t)+Y(t)),\quad t\in
\lbrack0,T^{\ast}],
\end{equation}
where $T^{\ast}>0$ is a fixed time horizon. The processes $\Lambda_{a}$ and
$\Lambda_{g}$ are assumed to be deterministic and they account for the
seasonalities observed in the spot prices.

One of the particularities of electricity markets is that power is a non
storable asset and for that reason is not a directly tradeable asset. This
entails that one can not derive the forward price of electricity from the
classical buy-and-hold hedging arguments. Using a risk-neutral pricing
argument (see Benth, \v{S}altyt\.{e} Benth and Koekebakker~\cite{BSBK}), under
the assumption of deterministic interest rates, the forward price, with time
of delivery $0<T<T^{\ast},$ at time $0<t<T$ is given by $F_{Q}(t,T)\triangleq
\mathbb{E}_{Q}[S(T)|\mathcal{F}_{t}],$ where $Q$ is any probability measure
equivalent to the historical measure $P$ and $\mathcal{F}_{t}$ is the market
information up to time $t$. In what follows we will use the probability
measure $Q$ discussed in the previous sections. However, in electricity
markets, the delivery of the underlying takes place over a period of time
$[T_{1},T_{2}],$ where $0<T_{1}<T_{2}<T^{\ast}.$ We call such contracts swap
contracts and we will denote their price at time $t<T_{1}$ by
\[
F_{Q}(t,T_{1},T_{2})\triangleq\mathbb{E}_{Q}[\frac{1}{T_{2}-T_{1}}\int_{T_{1}%
}^{T_{2}}S(T)dT|\mathcal{F}_{t}].
\]
We can use the stochastic Fubini theorem to relate the price of forwards and
swaps%
\[
F_{Q}(t,T_{1},T_{2})\triangleq\frac{1}{T_{2}-T_{1}}\int_{T_{1}}^{T_{2}}%
F_{Q}(t,T)dT.
\]
The risk premium for forward prices is defined by the following expression
$R_{Q}^{F}(t,T)\triangleq\mathbb{E}_{Q}[S(T)|\mathcal{F}_{t}]-\mathbb{E}%
_{P}[S(T)|\mathcal{F}_{t}],$ and for swap prices by%
\begin{equation}
R_{Q}^{S}(t,T_{1},T_{2})\triangleq F_{Q}(t,T_{1},T_{2})-\mathbb{E}_{Q}%
[\frac{1}{T_{2}-T_{1}}\int_{T_{1}}^{T_{2}}S(T)dT|\mathcal{F}_{t}]=\frac
{1}{T_{2}-T_{1}}\int_{T_{1}}^{T_{2}}R_{Q}^{F}(t,T)dT.
\label{Equ_Generic_Swap_RP}%
\end{equation}
In order to compute the previous quantities we need to know the dynamics of
$S$ (that is, of $X$ and $Y$) under $P$ and under $Q.$ Explicit expressions
for $X$ and $Y$ under $P$ are given in equations $\left(
\ref{Equ_X_Explicit_P}\right)  $ and $\left(  \ref{Equ_Y_Explicit_P}\right)
,$ respectively. In the rest of the paper, $Q=Q_{\bar{\theta},\bar{\beta}%
},\bar{\theta}\in\bar{D}_{L},\bar{\beta}\in\lbrack0,1]^{2}$ defined in
$\left(  \ref{Equ_Def_Q_theta_beta_p}\right)  ,$ and the explicit expressions
for $X$ and $Y$ under $Q$ are given in Remark $\ref{Remark_Dynamics_Q},$
equations $\left(  \ref{Equ_X_Explicit_Q}\right)  $ and $\left(
\ref{Equ_Y_Explicit_Q}\right)  ,$ respectively.

\begin{remark}
We will use the subindices $a$ and $g$ to denote the arithmetic and the
geometric spot models, respectively. That is, we will use the notation
$R_{a,Q}^{F}(t,T),R_{g,Q}^{F}(t,T),R_{a,Q}^{S}(t,T_{1},T_{2})$ and
$R_{g,Q}^{S}(t,T_{1},T_{2}).$
\end{remark}

\begin{remark}
In the discussion to follow, we are interested in finding values of the
parameters $\bar{\theta},\bar{\beta}$ such that some empirical features of the
observed risk premium profiles are reproduced by our pricing measure. In
particular, we show that is possible to have the sign of the risk premium
changing stochastically from positive values on the short end of the market to
negative values on the long end. This is proved for forward contracts in,
both, the arithmetic and geometric model. Equation $\left(
\ref{Equ_Generic_Swap_RP}\right)  $ just tell us that the risk premium for
swaps becomes the average of the risk premium for forwards with
fixed-delivery. Hence, we can obtain stochastic sign change also for these,
depending on the length of delivery. Worth noticing is that contracts in the
short end have short delivery (a day, or a week), while in the long end have
month/quarter/year delivery. Average for negative is negative, for the long
end, and average over short period, dominantly positive, gives positive, in
the short end.
\end{remark}

\subsection{Arithmetic spot price model}

We assume in this section that the spot price $S(t)$ is given by the dynamics
\eqref{Equ_Arith_Model}
for $0\leq t\leq T^{\ast}$, $T^{\ast}>0$, with the maturity time of the
forward contract $T$ satisfying $0<T<T^{\ast}.$ Using equations $\left(
\ref{Equ_X_Explicit_P}\right)  $ and $\left(  \ref{Equ_Y_Explicit_P}\right)
$\ and the basic properties of the conditional expectation we get%
\begin{align*}
\mathbb{E}_{P}[S(T)|\mathcal{F}_{t}]  &  =\Lambda_{a}(T)+\mathbb{E}%
_{P}[X(t)e^{-\alpha_{X}(T-t)}+\frac{\mu_{X}}{\alpha_{X}}(1-e^{-\alpha
_{X}(T-t)})|\mathcal{F}_{t}]\\
&  \qquad+\mathbb{E}_{P}[Y(t)e^{-\alpha_{Y}(T-t)}+\frac{\mu_{Y}+\kappa
_{L}^{\prime}(0)}{\alpha_{Y}}(1-e^{-\alpha_{Y}(T-t)})|\mathcal{F}_{t}]\\
&  \qquad+\mathbb{E}_{P}[\sigma_{X}\int_{t}^{T}e^{-\alpha_{X}(T-s)}%
dW(s)+\int_{t}^{T}\int_{0}^{\infty}e^{-\alpha_{Y}(T-s)}z\tilde{N}%
^{L}(ds,dz)|\mathcal{F}_{t}]\\
&  =\Lambda_{a}(T)+X(t)e^{-\alpha_{X}(T-t)}+Y(t)e^{-\alpha_{Y}(T-t)}\\
&  \qquad+\frac{\mu_{X}}{\alpha_{X}}(1-e^{-\alpha_{X}(T-t)})+\frac{\mu
_{Y}+\kappa_{L}^{\prime}(0)}{\alpha_{Y}}(1-e^{-\alpha_{Y}(T-t)})\\
&  \qquad+\mathbb{E}_{P}[\sigma_{x}\int_{t}^{T}e^{-\alpha_{X}(T-s)}%
dW(s)]+\mathbb{E}_{P}[\int_{t}^{T}\int_{0}^{\infty}e^{-\alpha_{Y}(T-u)}%
z\tilde{N}^{L}(ds,dz)]\\
&  =\Lambda_{a}(T)+X(t)e^{-\alpha_{X}(T-t)}+Y(t)e^{-\alpha_{Y}(T-t)}+\frac
{\mu_{X}}{\alpha_{X}}(1-e^{-\alpha_{X}(T-t)})\\
&  \qquad+\frac{\mu_{Y}+\kappa_{L}^{\prime}(0)}{\alpha_{Y}}(1-e^{-\alpha
_{Y}(T-t)}).
\end{align*}
Note that we have also used that $W$ and $\tilde{N}^{L}$ have independent
increments under $P$ to write conditional expectations as expectations. If we
assume that $\alpha\triangleq\alpha_{X}=\alpha_{Y},$ then%
\[
\mathbb{E}_{P}[S(T)|\mathcal{F}_{t}]=\Lambda_{a}(T)+(S(t)-\Lambda
(t))e^{-\alpha(T-t)}+\frac{\mu_{X}+\mu_{Y}+\kappa_{L}^{\prime}(0)}{\alpha
}(1-e^{-\alpha(T-t)}).
\]
This last expression for $\mathbb{E}_{P}[S(T)|\mathcal{F}_{t}]$ is
considerably simpler and depends explicitly on $S(t),$ the spot price at time
$t,$ which is directly observable in the market.

To find a similar expression for $\mathbb{E}_{Q}[S(T)|\mathcal{F}_{t}]$ we
need the following lemma.

\begin{lemma}
We have that $\int_{0}^{t}\int_{0}^{\infty}e^{\alpha_{Y}(1-\beta_{2})s}%
z\tilde{N}_{Q}^{L}(ds,dz)$ is a $Q$-martingale on $[0,T],T>0.$
\end{lemma}

\begin{proof}
We have to prove that $\mathbb{E}_{Q}[\int_{0}^{t}\int_{0}^{\infty}%
e^{\alpha_{Y}(1-\beta_{2})s}zv_{Q}^{L}(ds,dz)]<\infty.$ One has that
\begin{align*}
\mathbb{E}_{Q}[\int_{0}^{t}\int_{0}^{\infty}e^{\alpha_{Y}(1-\beta_{2})s}%
zv_{Q}^{L}(ds,dz)]  &  =\mathbb{E}_{Q}[\int_{0}^{t}\int_{0}^{\infty}%
e^{\alpha_{Y}(1-\beta_{2})s}zH_{\theta_{2},\beta_{2}}(s,z)\ell(dz)ds]\\
&  =\mathbb{E}_{Q}[\int_{0}^{t}\int_{0}^{\infty}e^{\alpha_{Y}(1-\beta_{2}%
)s}z\left(  e^{\theta_{2}z}+\frac{\alpha_{Y}\beta_{2}}{\kappa_{L}%
^{\prime\prime}(\theta_{2})}e^{\theta_{2}z}zY(s)\right)  \ell(dz)ds]\\
&  \leq e^{\alpha_{Y}T}\{T\kappa_{L}^{\prime}(\theta_{2})+\alpha_{Y}%
T\sup_{0\leq t\leq T}\mathbb{E}_{Q}[Y(t)]\},
\end{align*}
and $\kappa_{L}^{\prime}(\theta_{2})<\infty$ because $\theta_{2}\in D_{L}.$
The proof that $\sup_{0\leq s\leq T}\mathbb{E}_{Q}[Y(s)]$ is finite follows
the same lines as the last part of Theorem
\ref{Theo_StochExp_Htilda_martingale}. Using the semimartingale representation
of $Y,$ equation $\left(  \ref{Equ_GirsanovY}\right)  ,$ we obtain that there
exist constants $C_{0}$ and $C_{1}$ such that $\mathbb{E}_{Q}[Y(t)]\leq
C_{0}+C_{1}\int_{0}^{t}\mathbb{E}_{Q}[Y(s)]ds.$ Applying Gronwall's Lemma we
get that $\mathbb{E}_{Q}[Y(t)]\leq C_{0}e^{C_{1}T}$ and the result follows.
\end{proof}

\begin{remark}
\label{Remark_Q-Martingale}We need the previous lemma because Girsanov's
Theorem just ensures that
\begin{equation}
\int_{0}^{t}\int_{0}^{\infty}e^{\alpha_{Y}(1-\beta_{2})s}z\tilde{N}_{Q}%
^{L}(ds,dz) \label{Equ_ExpN^L_Q_Martingale}%
\end{equation}
is a $Q$-local martingale. We want $\left(  \ref{Equ_ExpN^L_Q_Martingale}%
\right)  $ to be a $Q$-martingale because then it follows trivially that
\[
\mathbb{E}_{Q}[\int_{t}^{T}\int_{0}^{\infty}e^{\alpha_{Y}(1-\beta_{2}%
)s}z\tilde{N}_{Q}^{L}(ds,dz)|\mathcal{F}_{t}]=0.
\]
Note that we can not reduce the previous conditional expectation (unless
$\beta_{2}=0,$ which coincides with the Esscher change of measure) to an
expectation because the compensator of $N_{Q}^{L}$ depends on $Y$ and,
therefore, $\tilde{N}_{Q}^{L}$ does not has independent increments.
\end{remark}

Using the basic properties of the conditional expectation, Remark
\ref{Remark_Q-Martingale} and equations $\left(  \ref{Equ_X_Explicit_Q}%
\right)  $ and $\left(  \ref{Equ_Y_Explicit_Q}\right)  $ we get%
\begin{align*}
\mathbb{E}_{Q}[S(T)|\mathcal{F}_{t}]  &  =\Lambda_{a}(T)+\mathbb{E}%
_{Q}[X(t)e^{-\alpha_{X}(1-\beta_{1})(T-t)}+\frac{\mu_{X}+\theta_{1}}%
{\alpha_{X}(1-\beta_{1})}(1-e^{-\alpha_{X}(1-\beta_{1})(T-t)})|\mathcal{F}%
_{t}]\\
&  +\mathbb{E}_{Q}[Y(t)e^{-\alpha_{Y}(1-\beta_{2})(T-t)}+\frac{\mu_{Y}%
+\kappa_{L}^{\prime}(\theta_{2})}{\alpha_{Y}(1-\beta_{2})}(1-e^{-\alpha
_{Y}(1-\beta_{2})(T-t)})|\mathcal{F}_{t}]\\
&  +\mathbb{E}_{Q}[\sigma_{X}\int_{t}^{T}e^{-\alpha_{X}(1-\beta_{1}%
)(T-s)}dW_{Q}(s)|\mathcal{F}_{t}]\\
&  +\mathbb{E}_{Q}[\int_{t}^{T}\int_{0}^{\infty}e^{-\alpha_{Y}(1-\beta
_{2})(T-s)}z\tilde{N}_{Q}^{L}(ds,dz)|\mathcal{F}_{t}]\\
&  =\Lambda_{a}(T)+X(t)e^{-\alpha_{X}(1-\beta_{1})(T-t)}+\frac{\mu_{X}%
+\theta_{1}}{\alpha_{X}(1-\beta_{1})}(1-e^{-\alpha_{X}(1-\beta_{1})(T-t)})\\
&  +Y(t)e^{-\alpha_{Y}(1-\beta_{2})(T-t)}+\frac{\mu_{Y}+\kappa_{L}^{\prime
}(\theta_{2})}{\alpha_{Y}(1-\beta_{2})}(1-e^{-\alpha_{Y}(1-\beta_{2})(T-t)})\\
&  +\mathbb{E}_{Q}[\sigma_{X}\int_{t}^{T}e^{-\alpha_{X}(1-\beta_{1}%
)(T-s)}dW_{Q}(s)]\\
&  +e^{-\alpha_{Y}(1-\beta_{2})T}\mathbb{E}_{Q}[\int_{t}^{T}\int_{0}^{\infty
}e^{\alpha_{Y}(1-\beta_{2})s}z\tilde{N}_{Q}^{L}(ds,dz)|\mathcal{F}_{t}]\\
&  =\Lambda_{a}(T)+X(t)e^{-\alpha_{X}(1-\beta_{1})(T-t)}+Y(t)e^{-\alpha
_{Y}(1-\beta_{2})(T-t)}\\
&  +\frac{\mu_{X}+\theta_{1}}{\alpha_{X}(1-\beta_{1})}(1-e^{-\alpha
_{X}(1-\beta_{1})(T-t)})+\frac{\mu_{Y}+\kappa_{L}^{\prime}(\theta_{2})}%
{\alpha_{Y}(1-\beta_{2})}(1-e^{-\alpha_{Y}(1-\beta_{2})(T-t)}).
\end{align*}
Therefore, we have proved the following result.

\begin{proposition}
\label{Prop_forwardprice_arithmetic} The forward price $F_{Q}(t,T)$ in the
arithmetic spot model \eqref{Equ_Arith_Model} is given by
\begin{align*}
F_{Q}(t,T)  &  =\Lambda_{a}(T)+X(t)e^{-\alpha_{X}(1-\beta_{1})(T-t)}%
+Y(t)e^{-\alpha_{Y}(1-\beta_{2})(T-t)}\\
&  \qquad+\frac{\mu_{X}+\theta_{1}}{\alpha_{X}(1-\beta_{1})}(1-e^{-\alpha
_{X}(1-\beta_{1})(T-t)})+\frac{\mu_{Y}+\kappa_{L}^{\prime}(\theta_{2})}%
{\alpha_{Y}(1-\beta_{2})}(1-e^{-\alpha_{Y}(1-\beta_{2})(T-t)}).
\end{align*}

\end{proposition}

In Lucia and Schwartz~\cite{LS} a two-factor model (among others) is proposed
as the dynamics for power spot prices in the Nordic electricity market
NordPool. Following the model of Schwartz and Smith~\cite{SS}, they consider a
non-stationary long term variation factor together with a stationary short
term variation factor. In our context, one could let the mean reversion in $X$
be zero, to obtain a non-stationary factor as a drifted Brownian motion under
the pricing measure $Q$. After doing a measure transform with $\beta_{1}=1$,
we can price forwards as in Proposition~\ref{Prop_forwardprice_arithmetic} to
find
\begin{align*}
F_{Q}(t,T)  &  =\Lambda_{a}(T)+X(t)+Y(t)e^{-\alpha_{Y}(1-\beta_{2})(T-t)}%
+(\mu_{X}+\theta_{1})(T-t)\\
&  \qquad+\frac{\mu_{Y}+\kappa_{L}^{\prime}(\theta_{2})}{\alpha_{Y}%
(1-\beta_{2})}(1-e^{-\alpha_{Y}(1-\beta_{2})(T-t)}).
\end{align*}
When $T-t$ becomes large, i.e. when we are far out on the forward curve, we
see that
\begin{equation}
F_{Q}(t,T)\sim\Lambda_{a}(T)+X(t)+(\mu_{X}+\theta_{1})(T-t)+\frac{\mu
_{Y}+\kappa_{L}^{\prime}(\theta_{2})}{\alpha_{Y}(1-\beta_{2})}\,.
\label{eq_LS-forward-asymptotic}%
\end{equation}
Thus, the forward curve moves stochastically as the non-stationary factor $X$.
If one, on the other hand, let $X$ be stationary, we find that the forward
price in Proposition ~\ref{Prop_forwardprice_arithmetic} will behave for large
time to maturities $T-t$ as
\[
F_{Q}(t,T)\sim\Lambda_{a}(T)+\frac{\mu_{X}+\theta_{1}}{\alpha_{X}(1-\beta
_{1})}+\frac{\mu_{Y}+\kappa_{L}^{\prime}(\theta_{2})}{\alpha_{Y}(1-\beta_{2}%
)}\,.
\]
The forward prices becomes constant after subtracting the seasonal function,
with no stochastic movements. This is not what is observed for forward data in
the market. However, following the empirical study in Barndorff-Nielsen, Benth
and Veraart~\cite{BNBV-spot}, electricity spot prices on the German power
exchange EEX are stationary. One way to have a stationary spot dynamics, and
still maintain forward prices which moves randomly in the long end, is to
apply our measure change to slow down the mean reversion in one or more
factors of the (stationary) spot. In the extreme case, we can let $\beta
_{1}=1$, and obtain a non-stationary factor $X$ under the pricing measure, in
which case we obtain the same long term asymptotic behaviour as in the
generalization of the Lucia and Schwartz model
\eqref{eq_LS-forward-asymptotic}. In conclusion, our pricing measure allows
for a stationary spot dynamics and a forward price dynamics which is not
constant in the long end.

Let us return back to the risk premium, which in view of
Prop.~\ref{Prop_forwardprice_arithmetic} becomes:

\begin{proposition}
\label{Prop_FRP_Arithmetic}The risk premium $R_{a,Q}^{F}(t,T)$ for the forward
price in the arithmetic spot model $\left(  \ref{Equ_Arith_Model}\right)  $ is
given by%
\begin{align*}
R_{a,Q}^{F}(t,T)  &  =X(t)e^{-\alpha_{X}(T-t)}(e^{\alpha_{X}\beta_{1}%
(T-t)}-1)+Y(t)e^{-\alpha_{Y}(T-t)}(e^{\alpha_{Y}\beta_{2}(T-t)}-1)\\
&  \qquad+\frac{\mu_{X}+\theta_{1}}{\alpha_{X}(1-\beta_{1})}(1-e^{-\alpha
_{X}(1-\beta_{1})(T-t)})+\frac{\mu_{Y}+\kappa_{L}^{\prime}(\theta_{2})}%
{\alpha_{Y}(1-\beta_{2})}(1-e^{-\alpha_{Y}(1-\beta_{2})(T-t)})\\
&  \qquad-\frac{\mu_{X}}{\alpha_{X}}(1-e^{-\alpha_{X}(T-t)})-\frac{\mu
_{Y}+\kappa_{L}^{\prime}(0)}{\alpha_{Y}}(1-e^{-\alpha_{Y}(T-t)}).
\end{align*}

\end{proposition}

We analyse different cases for the risk premium in the next subsection.

\subsubsection{Discussion on the risk premium}

The first remarkable property of this measure change is that, as long as the
parameter $\bar{\beta}\neq(0,0),$ the risk premium is stochastic. This might
be a desirable feature in view of the discussion in the Introduction where we
referred to the economical and empirical evidence in Geman and
Vasicek~\cite{GV}, Bessembinder and Lemon~\cite{BL} and Benth, Cartea and
Kiesel \cite{BCK}.
Note that when $\bar{\beta}=(0,0),$ our measure change coincides with the
Esscher transform (see Benth, \v{S}altyt\.{e} Benth and Koekebakker
\cite{BSBK}). In the Esscher case, the risk premium has a deterministic
evolution given by%
\begin{equation}
R_{a,Q}^{F}(t,T)=\frac{\theta_{1}}{\alpha_{X}}(1-e^{-\alpha_{X}(T-t)}%
)+\frac{\kappa_{L}^{\prime}(\theta_{2})-\kappa_{L}^{\prime}(0)}{\alpha_{Y}%
}(1-e^{-\alpha_{Y}(T-t)}), \label{Equ_RPA_Esscher}%
\end{equation}
an already known result, see Benth and Sgarra \cite{BeSg12}.

Another interesting feature of the empirical risk premium is that its sign
might change from positive to negative when the time to maturity
$\tau\triangleq T-t$ increases. Hence, we are interested in theoretical models
that allow to reproduce such empirical property. From now on we shall rewrite
the expressions for the risk premium in terms of the time to maturity $\tau$
and, slightly abusing the notation, we will write $R_{a,Q}^{F}(t,\tau)$
instead of $R_{a,Q}^{F}(t,t+\tau).$ We fix the parameters of the model under
the historical measure $P,$ i.e., $\mu_{X},\alpha_{X},\sigma_{X},\mu_{Y},$ and
$\alpha_{Y},$ and study the possible sign of $R_{a,Q}^{F}(t,\tau)$ in terms of
the change of measure parameters, i.e., $\bar{\beta}=(\beta_{1},\beta_{2})$
and $\bar{\theta}=(\theta_{1},\theta_{2})$ and the time to maturity $\tau.$
Note that the present time just enters into the picture through the stochastic
components $X$ and $Y.$ We are going to assume $\mu_{X}=\mu_{Y}=0.$ This
assumption is justified, from a modeling point of view, because we want the
processes $X$ and $Y$ to revert toward zero. In this way, the seasonality
function $\Lambda_{a}$ accounts completely for the mean price level. On the
other hand it is also reasonable to expect that $\alpha_{X}<\alpha_{Y},$ which
means that the component accounting for the jumps reverts the fastest (e.g.,
being the factor modelling the spikes). The factor $X$ is referred to as the
base component, modelling the normal price variations when the market is not
under particular stress. The expression for $R_{a,Q}^{F}(t,\tau)$ given in
Proposition \ref{Prop_FRP_Arithmetic} allows for a quite rich behaviour. We
are going to study the cases $\bar{\theta}=(0,0),\bar{\beta}=(0,0)$ and the
general case separately. Moreover, in order to graphically illustrate the
discussion we plot the risk premium profiles obtained assuming that the
subordinator $L$ is a compound Poisson process with jump intensity
$c/\lambda>0$ and exponential jump sizes with mean $\lambda.$ That is, $L$
will have the L\'{e}vy measure given in Example \ref{Example_Subordinators}.
We shall measure the time to maturity $\tau$ in days and plot $R_{a,Q}%
^{F}(t,\tau)$ for $\tau\in\lbrack0,360],$ roughly one year. We fix the values
of the following parameters%
\[
\alpha_{X}=0.099,\alpha_{Y}=0.3466,c=0.4,\lambda=2.
\]
The speed of mean reversion for the base component $\alpha_{X}$ yields a
half-life of seven days, while the one for the spikes $\alpha_{Y}$ yields a
half-life of two days (see e.g., Benth, Saltyte Benth and Koekebakker
\cite{BSBK} for the concept of half-life). The values for $c$ and $\lambda$
give jumps with mean $0.5$ and frequency of $5$ spikes a month.

The following lemma will help us in the discussion to follow.

\begin{lemma}
\label{Lemma_Main_RPA}If $\mu_{X}=\mu_{Y}=0$ and $\alpha_{X}<\alpha_{Y},$ we
have that the risk premium $R_{a,Q}^{F}(t,\tau)$ satisfies%
\begin{align}
R_{a,Q}^{F}(t,\tau)  &  =X(t)e^{-\alpha_{X}\tau}(e^{\alpha_{X}\beta_{1}\tau
}-1)+Y(t)e^{-\alpha_{Y}\tau}(e^{\alpha_{Y}\beta_{2}\tau}%
-1)\label{Equ_RiskPremium_Arithm_ZeroDrifts}\\
&  +\frac{\theta_{1}}{\alpha_{X}(1-\beta_{1})}(1-e^{-\alpha_{X}(1-\beta
_{1})\tau})+\frac{\kappa_{L}^{\prime}(\theta_{2})-\kappa_{L}^{\prime}%
(0)}{\alpha_{Y}(1-\beta_{2})}(1-e^{-\alpha_{Y}(1-\beta_{2})\tau})\nonumber\\
&  +\frac{\kappa_{L}^{\prime}(0)}{\alpha_{Y}}\Lambda(\alpha_{Y}\tau
,1-\beta_{2}),\nonumber
\end{align}
where%
\begin{align*}
\Lambda(x,y)  &  =\frac{1-e^{-xy}}{y}-(1-e^{-x}),\quad x\in\mathbb{R}_{+}%
,y\in\lbrack0,1],\\
\lim_{x\rightarrow\infty}\Lambda(x,y)  &  =\frac{1-y}{y},\\
\lim_{x\rightarrow0}\frac{\partial}{\partial x}\Lambda(x,y)  &  =0,
\end{align*}
is a non-negative function. Moreover,%
\begin{align}
\lim_{\tau\rightarrow\infty}R_{a,Q}^{F}(t,\tau)  &  =\frac{\theta_{1}}%
{\alpha_{X}(1-\beta_{1})}+\frac{\kappa_{L}^{\prime}(\theta_{2})-\kappa
_{L}^{\prime}(0)}{\alpha_{Y}(1-\beta_{2})}+\frac{\kappa_{L}^{\prime}%
(0)}{\alpha_{Y}}\frac{\beta_{2}}{1-\beta_{2}}%
,\label{Equ_RPremium_Infinity_Arithmetic}\\
\lim_{\tau\rightarrow0}\frac{\partial}{\partial\tau}R_{a,Q}^{F}(t,\tau)  &
=X(t)\alpha_{X}\beta_{1}+Y(t)\alpha_{Y}\beta_{2}+\theta_{1}+\kappa_{L}%
^{\prime}(\theta_{2})-\kappa_{L}^{\prime}(0).
\label{Equ_DRiskprem_Zero_Arithmetic}%
\end{align}

\end{lemma}

\begin{proof}
It follows trivially from Proposition \ref{Prop_FRP_Arithmetic} and the
assumptions on the coefficients $\mu_{X},\mu_{Y},\alpha_{X}$ and $\alpha_{Y}.$
\end{proof}

\begin{figure}[t]
\centering
\subfigure[$\theta_1=0.075,\theta_2=0$]{\includegraphics[width=2.90in]{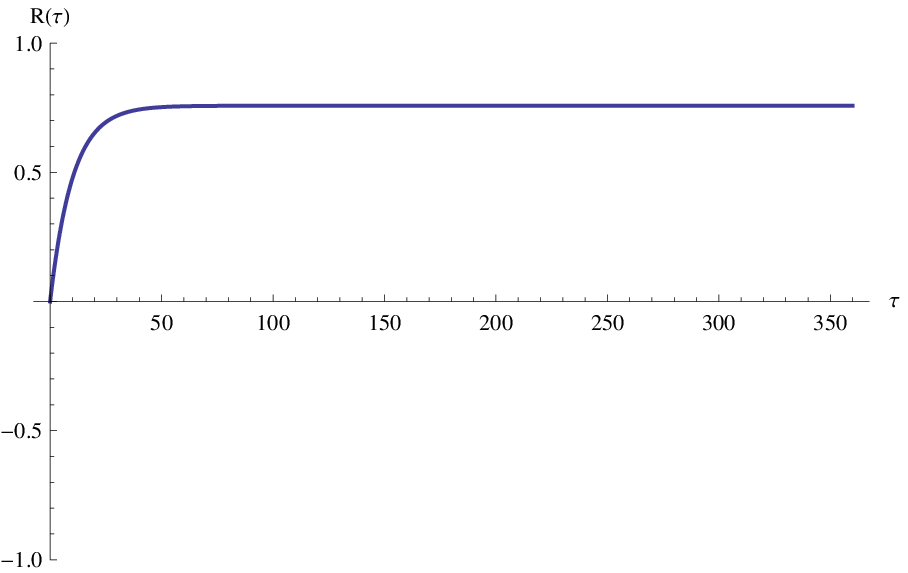}\label{FigBeta0_1_a}}
\subfigure[$\theta_1=-0.075,\theta_2=0$]{\includegraphics[width=2.90in]{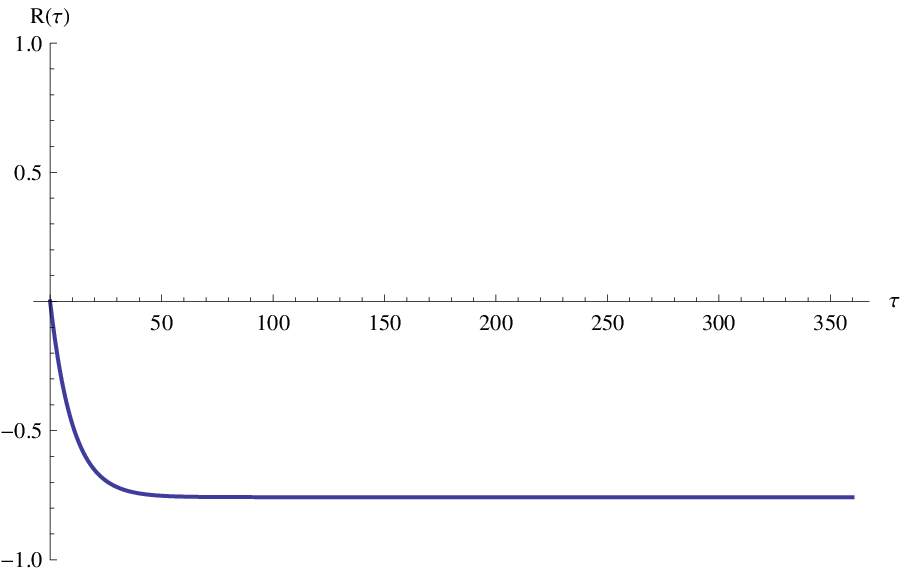}\label{FigBeta0_1_b}}\newline%
\subfigure[$\theta_1=0,\theta_2=0.75$]{\includegraphics[width=2.90in]{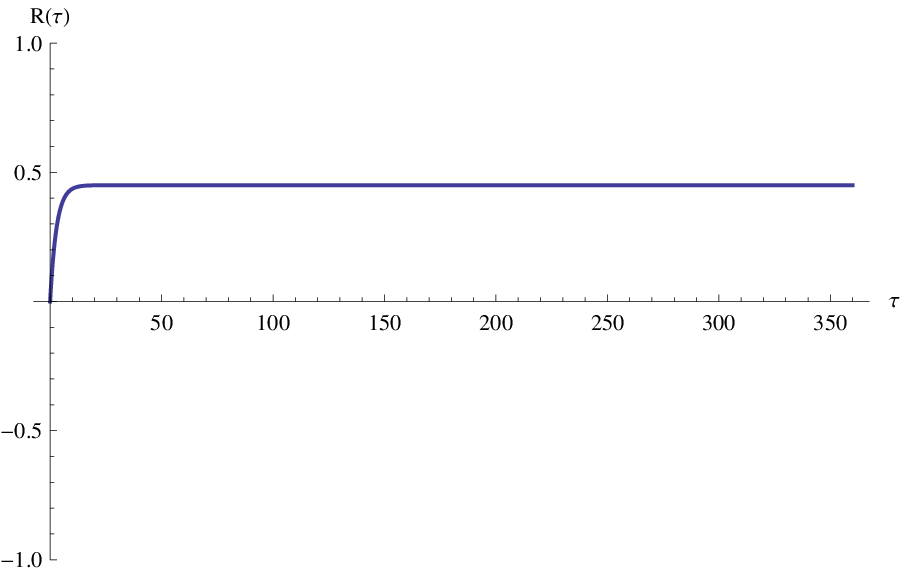}\label{FigBeta0_1_c}}
\subfigure[$\theta_1=0,\theta_2=-0.75$]{\includegraphics[width=2.90in]{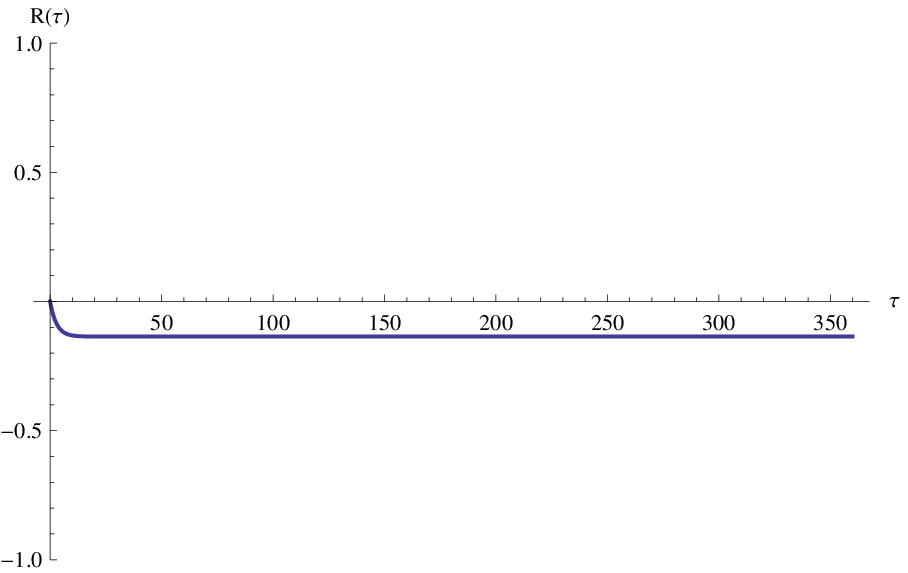}\label{FigBeta0_1_d}}
\caption{Risk premium profiles when $L$ is a compound Poisson process with
exponentially distributed jumps. Esscher transform: case $\bar{\beta}=(0,0).$
Arithmetic spot price model}%
\label{FigBeta0_1}%
\end{figure}

\begin{remark}
\label{Remark_Limit_Cond_RP_Arithm}The previous Lemma shows that the risk
premium $R_{a,Q}^{F}(t,\tau)$ vanishes with rate given by equation $\left(
\ref{Equ_DRiskprem_Zero_Arithmetic}\right)  $ at the short end of the forward
curve, when $\tau$ converges to zero, and\ approaches the value given in
equation $(\ref{Equ_RPremium_Infinity_Arithmetic})$ at long end of the forward
curve, when $\tau$ tends to infinity. It follows that the sign of $R_{a,Q}%
^{F}(t,\tau)$ in the short end of the forward curve will be positive if
$\left(  \ref{Equ_DRiskprem_Zero_Arithmetic}\right)  $ is positive and
negative if $\left(  \ref{Equ_DRiskprem_Zero_Arithmetic}\right)  $ is
negative. Hence, a sufficient condition to obtain the empirically observed
risk premium profiles (with positive values in the short end and negative
values in the long end of the forward curve) is to choose the values of the
parameters $\bar{\theta}\in\bar{D}_{L}$ and $\bar{\beta}\in\lbrack0,1]^{2}$
such that the following two conditions are simultaneously satisfied%
\begin{align*}
\frac{\theta_{1}}{\alpha_{X}(1-\beta_{1})}+\frac{\kappa_{L}^{\prime}%
(\theta_{2})-\kappa_{L}^{\prime}(0)}{\alpha_{Y}(1-\beta_{2})}+\frac{\kappa
_{L}^{\prime}(0)}{\alpha_{Y}}\frac{\beta_{2}}{1-\beta_{2}}  &  <0,\\
X(t)\alpha_{X}\beta_{1}+Y(t)\alpha_{Y}\beta_{2}+\theta_{1}+\kappa_{L}^{\prime
}(\theta_{2})-\kappa_{L}^{\prime}(0)  &  >0.
\end{align*}
We also recall here that, according to Remark \ref{Remark_D_Cumulants},
$\kappa^{\prime}(\theta)$ is positive, increasing function, so the sign of
$\kappa_{L}^{\prime}(\theta_{2})-\kappa_{L}^{\prime}(0)$ is equal to the sign
of $\theta_{2}.$ Moreover, it is easy to see that
\[
-\kappa_{L}^{\prime}(0)<\kappa_{L}^{\prime}(\theta_{2})-\kappa_{L}^{\prime
}(0)<\kappa_{L}^{\prime}(\Theta_{L}/2)-\kappa_{L}^{\prime}(0)<\infty.
\]

\end{remark}

\begin{itemize}
\item \textbf{Changing the level of mean reversion (Esscher transform)},
$\bar{\beta}=(0,0):$ Setting $\bar{\beta}=(0,0),$ the probability measure $Q$
only changes the level of mean reversion (which is assumed to be zero under
the historical measure $P$). On the other hand, the risk premium is
deterministic and cannot change with changing market conditions. From equation
$\left(  \ref{Equ_RPA_Esscher}\right)  ,$ we get that if we set $\theta
_{2}=0,$ which means that we just change the level of the regular factor $X,$
the sign of $R_{a,Q}^{F}(t,\tau)$ is the same for any time to maturity $\tau$
and it is equal to the sign of $\theta_{1},$ see Figures \ref{FigBeta0_1_a}
and \ref{FigBeta0_1_b}. The situation is similar if we set $\theta_{1}=0,$
then the sign of $R_{a,Q}^{F}(t,\tau)$ is constant over the time to maturity
$\tau$ end equal to the sign of $\kappa_{L}^{\prime}(\theta_{2})-\kappa
_{L}^{\prime}(0),$ that is to the sign of $\theta_{2},$ see Figures
\ref{FigBeta0_1_c} and \ref{FigBeta0_1_d}.

\begin{figure}[t]
\centering
\subfigure[$\theta_1=-0.1,\theta_2=0.95$]{\includegraphics[width=2.90in]{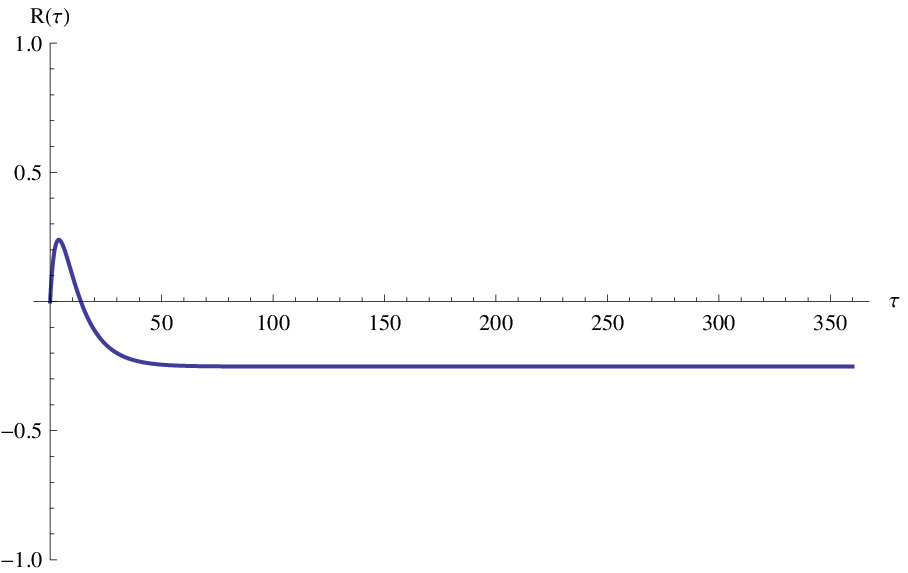}\label{FigBeta0_2_a}}
\subfigure[$\theta_1=0.02,\theta_2=-0.95$]{\includegraphics[width=2.90in]{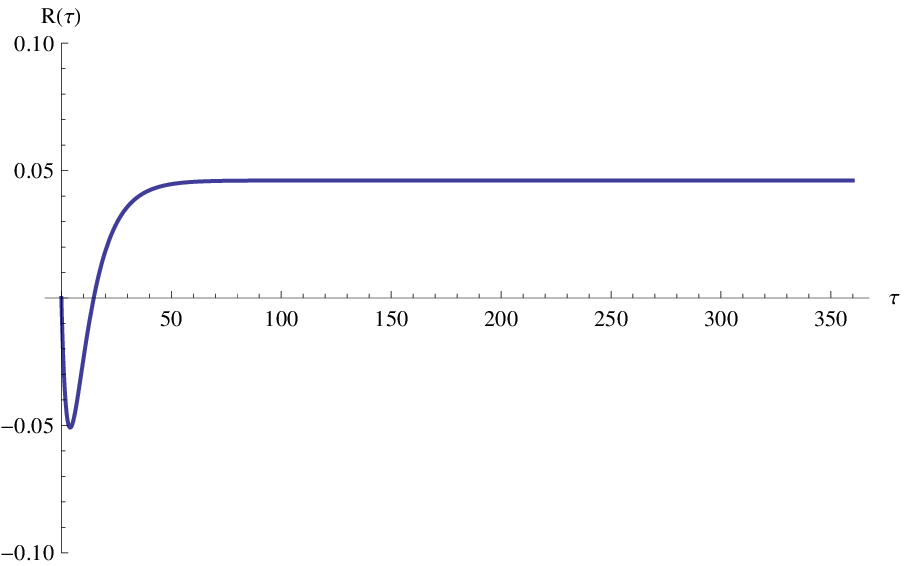}\label{FigBeta0_2_b}}\newline%
\subfigure[$\theta_1=-0.05,\theta_2=0.95$]{\includegraphics[width=2.90in]{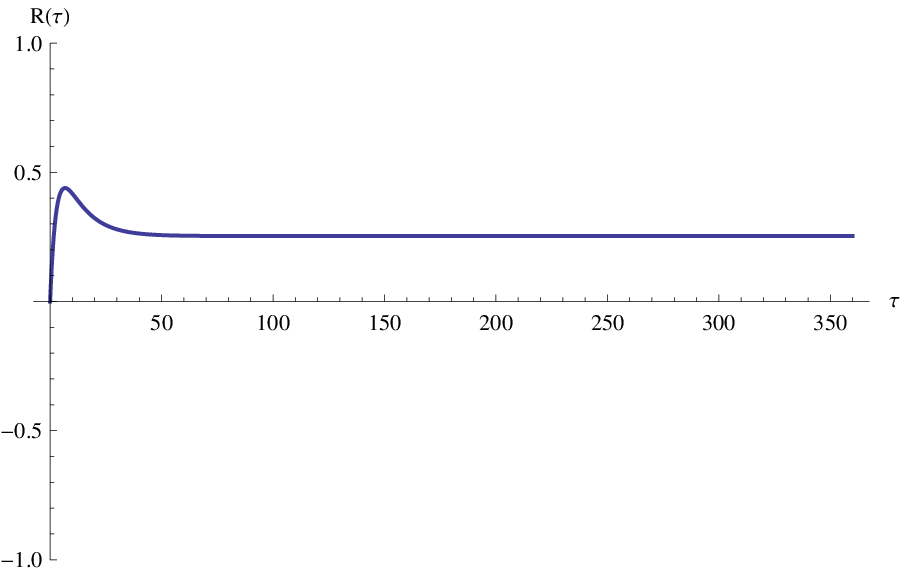}\label{FigBeta0_2_c}}
\subfigure[$\theta_1=-0.075,\theta_2=0.15$]{\includegraphics[width=2.90in]{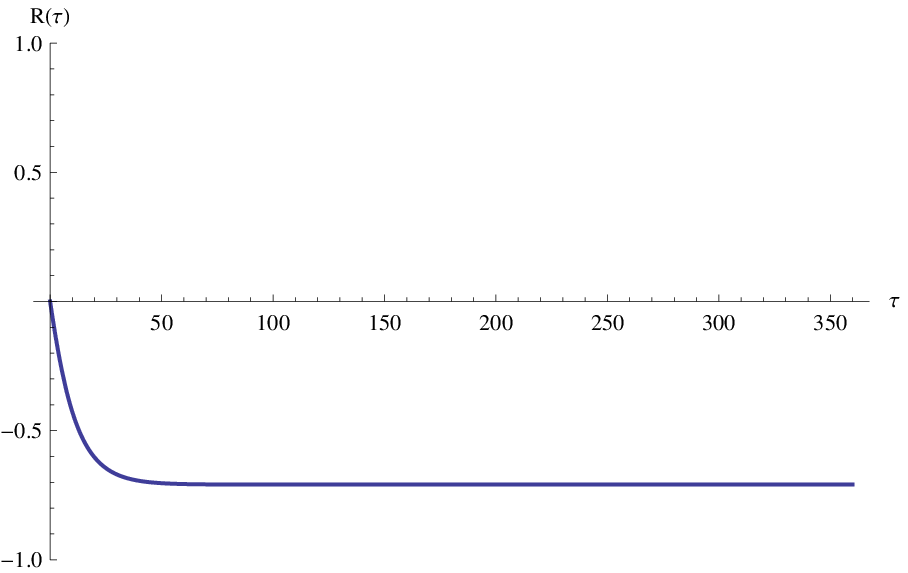}\label{FigBeta0_2_d}}\caption{Risk
premium profiles when $L$ is a compound Poisson process with exponentially
distributed jumps. Esscher transform: case $\bar{\beta}=(0,0).$ Arithmetic
spot price model}%
\label{FigBeta0_2}%
\end{figure}

When both $\theta_{1}$ and $\theta_{2}$ are different from zero the situation
is more interesting, the sign of $R_{a,Q}^{F}(t,\tau)$ may change depending on
the time to maturity. By Remark \ref{Remark_Limit_Cond_RP_Arithm} it suffices
to choose $\theta_{1}<0$ and $\theta_{2}>0$ satisfying%
\begin{align}
\frac{\theta_{1}}{\alpha_{X}}+\frac{\kappa_{L}^{\prime}(\theta_{2})-\kappa
_{L}^{\prime}(0)}{\alpha_{Y}}  &  <0,\label{Equ_RPA_Esscher_1}\\
\theta_{1}+\kappa_{L}^{\prime}(\theta_{2})-\kappa_{L}^{\prime}(0)  &  >0,
\label{Equ_RPA_Esscher_2}%
\end{align}
(these exist because $\alpha_{X}<\alpha_{Y}$ and $\kappa_{L}^{\prime}(\theta)$
is increasing) to get that $R_{a,Q}^{F}(t,\tau)>0$ for $\tau$ close to zero
and $R_{a,Q}^{F}(t,\tau)<0$ for $\tau$ large enough, see Figure
\ref{FigBeta0_2_a}. This corresponds to the situation of a premium induced
from consumers' hedging pressure on short-term contracts and long term hedging
of producers. We can also choose values for $\theta_{1}>0$ and $\theta_{2}<0$
such that equations \ref{Equ_RPA_Esscher_1} and \ref{Equ_RPA_Esscher_2} are
satisfied but with inverted inequalities. In this way, we can get that
$R_{a,Q}^{F}(t,\tau)<0$ for $\tau$ close to zero and $R_{a,Q}^{F}(t,\tau)>0$
for $\tau$ large enough, see Figure \ref{FigBeta0_2_b}. Risk premium profiles
with constant sign can also be generated, see Figures \ref{FigBeta0_2_c} and
\ref{FigBeta0_2_d}.

\item \textbf{Changing the speed of mean reversion,} $\bar{\theta}=(0,0):$
Setting $\bar{\theta}=(0,0),$ the probability measure $Q$ only changes speed
of mean reversion. Note that in this case the risk premium is stochastic and
it changes with market conditions. By Lemma \ref{Lemma_Main_RPA} we have that
the risk premium is given by
\begin{align*}
R_{a,Q}^{F}(t,\tau)  &  =X(t)e^{-\alpha_{X}\tau}(e^{\alpha_{X}\beta_{1}\tau
}-1)+Y(t)e^{-\alpha_{Y}\tau}(e^{\alpha_{Y}\beta_{2}\tau}-1)\\
&  +\frac{\kappa_{L}^{\prime}(0)}{\alpha_{Y}}\Lambda(\alpha_{Y}\tau
,1-\beta_{2}),
\end{align*}
and%
\begin{align*}
\lim_{\tau\rightarrow\infty}R_{a,Q}^{F}(t,\tau)  &  =\frac{\kappa_{L}^{\prime
}(0)}{\alpha_{Y}}\frac{\beta_{2}}{1-\beta_{2}}\geq0,\\
\lim_{\tau\rightarrow0}\frac{\partial}{\partial\tau}R_{a,Q}^{F}(t,\tau)  &
=X(t)\alpha_{X}\beta_{1}+Y(t)\alpha_{Y}\beta_{2}.
\end{align*}
Hence the risk premium will approach to a non negative value in the long end
of the market. In the short end, it can be both positive or negative and
stochastically varying with $X(t)$ and $Y(t),$ but $Y(t)$ will always
contribute to a positive sign. Actually, as the function $\Lambda(x,y)$ is
non-negative and $\kappa_{L}^{\prime}(0)$ is strictly positive, the only
negative contribution to $R_{a,Q}^{F}(t,\tau)$ comes from the term due to the
base component $X$. Hence, if $\beta_{1}=0$ or $X(t)\geq0,$ then $R_{a,Q}%
^{F}(t,\tau)$ will be positive for all times to maturity. Some of the possible
risk profiles that can be obtained are plotted in Figure \ref{FigTheta0_3}.

\begin{figure}[t]
\centering
\subfigure[$\beta_1=0.25,\beta_2=0.75,X(t)=2.5,Y(t)=2.5$]{\includegraphics[width=2.90in]{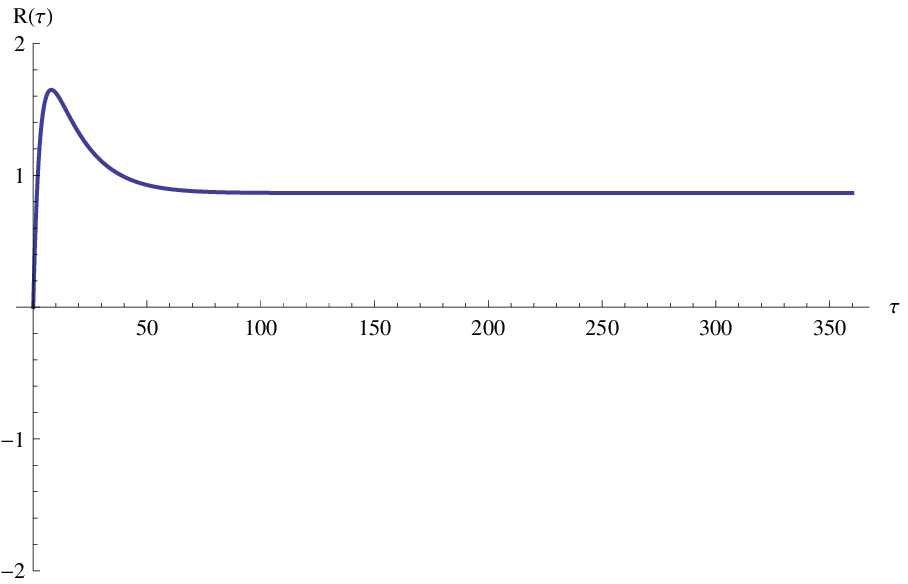}\label{FigTheta0_3_a}}
\subfigure[$\beta_1=0.75,\beta_2=0,X(t)=-2.5,Y(t)=2.5$]{\includegraphics[width=2.90in]{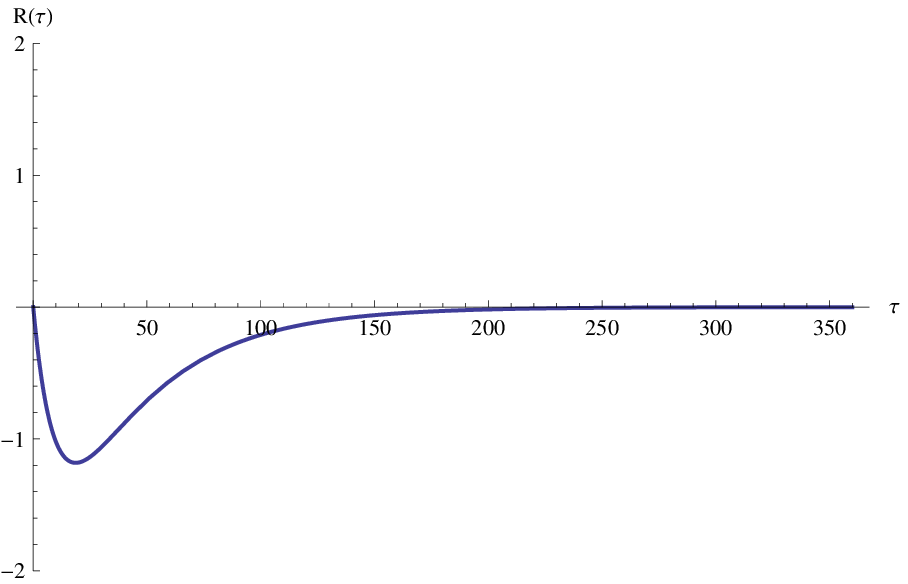}\label{FigTheta0_3_b}}\newline%
\subfigure[$\beta_1=0.75,\beta_2=0.75,X(t)=-2.5,Y(t)=0$]{\includegraphics[width=2.90in]{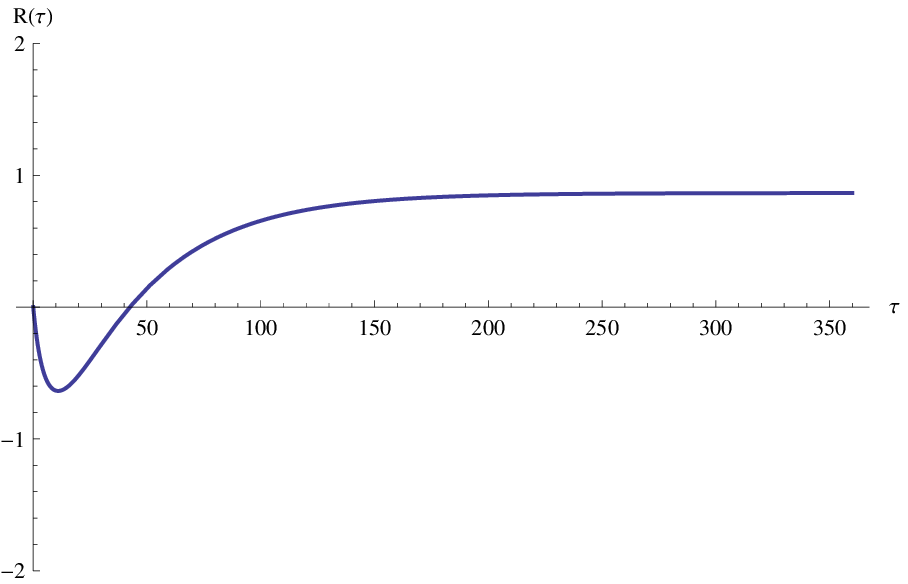}\label{FigTheta0_3_c}}
\subfigure[$\beta_1=0.5,\beta_2=0.5,X(t)=-2.5,Y(t)=2.5$]{\includegraphics[width=2.90in]{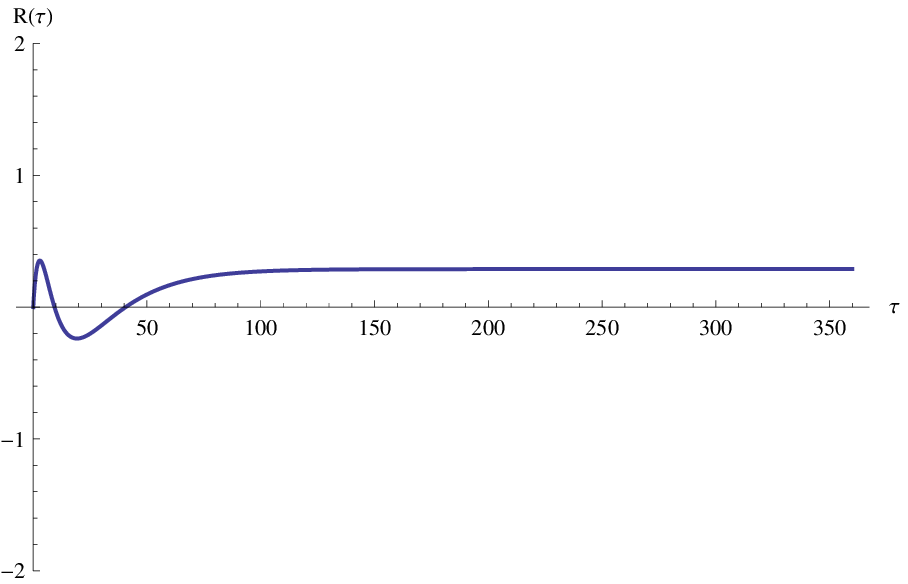}\label{FigTheta0_3_d}}\caption{Risk
premium profiles when $L$ is a compound Poisson process with exponentially
distributed jumps. Case $\bar{\theta}=(0,0).$ Arithmetic spot price model}%
\label{FigTheta0_3}%
\end{figure}\begin{figure}[t]
\centering
\subfigure[$\beta_1=0,\beta_2=0.88,\theta_1=-0.5,\theta_2=0.5,X(t)=\mathbb{R},Y(t)=5$]{\includegraphics[width=2.90in]{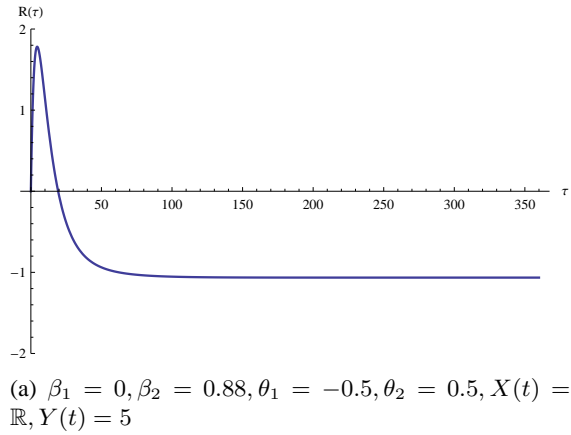}\label{Fig_4_a}}\caption{Risk
premium profiles when $L$ is a compound Poisson process with exponentially
distributed jumps. Arithmetic spot price model}%
\label{Fig_4}%
\end{figure}

\item \textbf{Changing the level and speed of mean reversion simultaneously}:
The general case is quite complex to analyse. As we are more interested in how
the change of measure $Q$ influence the component $Y(t)$, responsible for the
spikes in the prices, we are going to assume that $\beta_{1}=0.$ This means
that $Q$ may change the level of mean reversion of the regular component
$X(t),$ but not the speed at which this component reverts to that level. The
first implication of this assumption is that the possible stochastic component
in $R_{a,Q}^{F}(t,\tau)$ due to $X(t)$ vanish. This simplifies the analysis as
this term could be positive or negative. By Lemma \ref{Lemma_Main_RPA} we get
that
\begin{align*}
R_{a,Q}^{F}(t,\tau)  &  =Y(t)e^{-\alpha_{Y}\tau}(e^{\alpha_{Y}\beta_{2}\tau
}-1)+\frac{\theta_{1}}{\alpha_{X}}(1-e^{-\alpha_{X}\tau})\\
&  +\frac{\kappa_{L}^{\prime}(\theta_{2})}{\alpha_{Y}(1-\beta_{2}%
)}(1-e^{-\alpha_{Y}(1-\beta_{2})\tau})-\frac{\kappa_{L}^{\prime}(0)}%
{\alpha_{Y}}(1-e^{-\alpha_{Y}\tau}).
\end{align*}
and
\begin{align}
\lim_{\tau\rightarrow\infty}R_{a,Q}^{F}(t,\tau)  &  =\frac{\theta_{1}}%
{\alpha_{X}}+\frac{\kappa_{L}^{\prime}(\theta_{2})-\kappa_{L}^{\prime}%
(0)}{\alpha_{Y}(1-\beta_{2})}+\frac{\kappa_{L}^{\prime}(0)}{\alpha_{Y}}%
\frac{\beta_{2}}{1-\beta_{2}},\label{Equ_RPA_General_1}\\
\lim_{\tau\rightarrow0}\frac{\partial}{\partial\tau}R_{a,Q}^{F}(t,\tau)  &
=Y(t)\alpha_{Y}\beta_{2}+\theta_{1}+\kappa_{L}^{\prime}(\theta_{2})-\kappa
_{L}^{\prime}(0). \label{Equ_RPA_General_2}%
\end{align}
Note that we can make equation $\left(  \ref{Equ_RPA_General_1}\right)  $
negative by simply choosing $\theta_{1}$
\begin{equation}
\theta_{1}<-\frac{\alpha_{X}}{\alpha_{Y}(1-\beta_{2})}\left(  \kappa
_{L}^{\prime}(\theta_{2})-\kappa_{L}^{\prime}(0)+\beta_{2}\kappa_{L}^{\prime
}(0)\right)  . \label{Equ_RPA_General_3}%
\end{equation}
On the other hand, to make equation $(\ref{Equ_RPA_General_2})$ positive, we
have to choose $\theta_{1}$ satisfying%
\begin{equation}
\theta_{1}>-(\kappa_{L}^{\prime}(\theta_{2})-\kappa_{L}^{\prime}%
(0))-Y(t)\alpha_{Y}\beta_{2}. \label{Equ_RPA_General_4}%
\end{equation}
Equations $\left(  \ref{Equ_RPA_General_3}\right)  $ and $\left(
\ref{Equ_RPA_General_4}\right)  $ are compatible if the following inequality
is satisfied%
\begin{equation}
\kappa_{L}^{\prime}(\theta_{2})-\kappa_{L}^{\prime}(0)+Y(t)\alpha_{Y}\beta
_{2}>\frac{\alpha_{X}}{\alpha_{Y}(1-\beta_{2})}\left(  \kappa_{L}^{\prime
}(\theta_{2})-\kappa_{L}^{\prime}(0)+\beta_{2}\kappa_{L}^{\prime}(0)\right)  .
\label{Equ_RPA_General_5}%
\end{equation}
For any $\theta_{2}>0,$ which yields $\kappa_{L}^{\prime}(\theta_{2}%
)-\kappa_{L}^{\prime}(0)>0$ (and $\theta_{1}<0$), we have that there exists
$\beta_{2}^{\ast}\in(0,1)$ such that if $\beta_{2}<\beta_{2}^{\ast}$ equation
$\left(  \ref{Equ_RPA_General_5}\right)  $ is satisfied. Actually, the larger
the value of $Y(t),$ the larger the value of $\beta_{2}^{\ast}$. If $Y(t)$ is
close to $\kappa_{L}^{\prime}(0)/\alpha_{Y},$ then $\beta_{2}^{\ast}$ is close
to $(\alpha_{Y}-\alpha_{X})/\alpha_{Y}.$ This just says that if the speed of
mean reversion of the spikes component is large (in absolute value and
relatively to the speed of mean reversion of the base component) one can
choose $\beta_{2}$ close to one. Even in the case that $Y(t)=0$, equation
$\left(  \ref{Equ_RPA_General_5}\right)  $ is satisfied by choosing $\beta
_{2}$ small enough. To sum up, we can create a measure $Q$ that can have a
positive premium in the short end of the forward market due to sudden positive
spikes in the price (that is, $Y$ increases), whereas in the long end of the
market these spikes are not influential and we have a negative premium, see
Figure \ref{Fig_4}.
\end{itemize}

\subsection{Geometric spot price model}

We assume in this section that the spot price $S(t)$ follows the geometric
model $\eqref{Equ_Geom_Model}$ for $0\leq t\leq T^{\ast}$, $T^{\ast}>0$ and
with the maturity of the forward contract being $0<T<T^{\ast}$. In our
setting, the geometric model is harder to deal with than the arithmetic one.
The results obtained are fair less explicit and some additional integrability
conditions on $L$ are required. A first, natural, additional assumption on $L$
is that the constant $\Theta_{L}$ appearing in Assumption
\ref{Assumption_Exp_Theta_L} to be bigger than $1.$ This condition is
reasonable to expect because it just states that $\mathbb{E}[e^{L(t)}%
]<\infty,$ for all $t\in\mathbb{R}$ , and if we want $\mathbb{E}[e^{Y(t)}]$ to
be finite it seems a\ minimal assumption. Note, however that this is not
entirely obvious because the process $Y$ has a mean reversion structure that
$L$ does not have. On the other hand, the complex probabilistic structure of
the spike factor $Y$ under the new probability measure $Q,$ makes the
computations much more difficult. Still, it is possible to compute the risk
premium analytically in some cases. In general, one has to rely on numerical techniques.

In what follows, we shall compute the conditional expectations involved under
$Q$ (note that $Q=P,$ when $\theta_{1}=\theta_{2}=\beta_{1}=\beta_{2}=0$).
First, we show that the problem can be reduced to the study of the spike
component $Y.$ Due to the independence of $X$ and $Y,$ we have that
\begin{align*}
\mathbb{E}_{Q}[S(T)]  &  =\Lambda_{g}(T)\mathbb{E}_{Q}[\exp(X(T)+Y(T))]\\
&  =\Lambda_{g}(T)\mathbb{E}_{Q}[\exp(X(T))]\mathbb{E}_{Q}[\exp(Y(T))],
\end{align*}
which is finite if and only $\mathbb{E}_{Q}[\exp(X(T))]<\infty$ and
$\mathbb{E}_{Q}[\exp(Y(T))]<\infty.$ As $X(T)$ is a Gaussian random variable
it has finite exponential moments. To determine whether $\mathbb{E}_{Q}%
[\exp(Y(T))]$ is finite or not is not as straightforward. Let us assume, for
now, that it is finite. Then, it makes sense to compute the following
conditional expectation%
\begin{align*}
\mathbb{E}_{Q}[S(T)|\mathcal{F}_{t}]  &  =\Lambda_{g}(T)\mathbb{E}_{Q}%
[\exp(X(T)+Y(T))|\mathcal{F}_{t}]\\
&  =\Lambda_{g}(T)\mathbb{E}_{Q}[\mathbb{E}_{Q}[\exp(X(T))|\mathcal{F}_{t}%
\vee\sigma(\{Y(t)\}_{0\leq t\leq T})]\exp(Y(T))|\mathcal{F}_{t}].
\end{align*}
Using $\left(  \ref{Equ_X_Explicit_Q}\right)  $, the fact that $X$ is
independent of $\sigma(\{Y(t)\}_{0\leq t\leq T})$ and basic properties of the
conditional expectation we get that
\begin{align*}
&  \mathbb{E}_{Q}[\exp(X(T))|\mathcal{F}_{t}\vee\sigma(\{Y(t)\}_{0\leq t\leq
T})]\\
&  =\exp\left(  X(t)e^{-\alpha_{X}(1-\beta_{1})(T-t)}+\frac{\mu_{X}+\theta
_{1}}{\alpha_{X}(1-\beta_{1})}(1-e^{-\alpha_{X}(1-\beta_{1})(T-t)})\right) \\
&  \qquad\times\mathbb{E}_{Q}[\exp\left(  \sigma_{X}\int_{t}^{T}e^{-\alpha
_{X}(1-\beta_{1})(T-s)}dW_{Q}(s)\right)  ]\\
&  =\exp\left(  X(t)e^{-\alpha_{X}(1-\beta_{1})(T-t)}+\frac{\mu_{X}+\theta
_{1}}{\alpha_{X}(1-\beta_{1})}(1-e^{-\alpha_{X}(1-\beta_{1})(T-t)})\right) \\
&  \qquad\times\exp\left(  \frac{\sigma_{X}^{2}}{4\alpha_{X}(1-\beta_{1}%
)}(1-e^{-2\alpha_{X}(1-\beta_{1})(T-t)})\right)  .
\end{align*}
Hence, we have reduced the problem to the study of $\mathbb{E}_{Q}%
[\exp(Y(T))|\mathcal{F}_{t}].$

Let us start with the Esscher case $Q=Q_{\theta_{2},\beta_{2}}$ with
$\theta_{2}\in D_{L}$ and $\beta_{2}=0.$ We have that
\begin{align*}
\mathbb{E}_{Q}[\exp(Y(T))]  &  =\exp\left\{  Y(0)e^{-\alpha_{Y}T}+\frac
{\mu_{Y}+\kappa_{L}^{\prime}(\theta_{2})}{\alpha_{Y}}(1-e^{-\alpha_{Y}%
T})\right\} \\
&  \qquad\times\mathbb{E}_{Q}[\exp\left(  \int_{0}^{T}\int_{0}^{\infty
}ze^{-\alpha_{Y}(T-s)}\tilde{N}_{Q}^{L}(ds,dz)\right)  ]\\
&  =\exp\left\{  Y(0)e^{-\alpha_{Y}T}+\frac{\mu_{Y}}{\alpha_{Y}}%
(1-e^{-\alpha_{Y}T})\right\} \\
&  \qquad\times\mathbb{E}_{Q}[\exp\left(  \int_{0}^{T}\int_{0}^{\infty
}ze^{-\alpha_{Y}(T-s)}N_{Q}^{L}(ds,dz)\right)  ]\\
&  =\exp\left\{  Y(0)e^{-\alpha_{Y}T}+\frac{\mu_{Y}}{\alpha_{Y}}%
(1-e^{-\alpha_{Y}T})\right\} \\
&  \qquad\times\exp\left\{  \int_{0}^{T}\int_{0}^{\infty}(e^{ze^{-\alpha
_{Y}(T-s)}}-1)e^{\theta_{2}z}\ell(dz)ds\right\}  ,
\end{align*}
where we have used that the compensator of $L$ under $Q$ is $v_{Q}%
^{L}(ds,dz)=e^{\theta_{2}z}\ell(dz)ds$ (note that $e^{\theta_{2}z}\ell(dz)$ is
a L\'{e}vy measure) and Proposition 3.6 in Cont and Tankov \cite{CoTa04}. Of
course, the previous result holds as long as the integral in the exponential
is finite. A sufficient condition for the integrability of $\exp(Y(T))$
follows from%
\begin{align*}
&  \int_{0}^{T}\int_{0}^{\infty}(e^{ze^{-\alpha_{Y}(T-s)}}-1)e^{\theta_{2}%
z}\ell(dz)ds\\
&  \qquad=\int_{0}^{T}\int_{0}^{\infty}ze^{-\alpha_{Y}(T-s)}\left(  \int
_{0}^{1}e^{\lambda ze^{-\alpha_{Y}(T-s)}}d\lambda\right)  e^{\theta_{2}z}%
\ell(dz)ds\\
&  \qquad\leq\int_{0}^{T}\int_{0}^{\infty}ze^{-\alpha_{Y}(T-s)}e^{z\left(
\theta_{2}+e^{-\alpha_{Y}(T-s)}\right)  }\ell(dz)ds\leq T\kappa_{L}^{\prime
}(\theta_{2}+1).
\end{align*}
As $\theta_{2}\in D_{L},$ to have $\kappa_{L}^{\prime}(\theta_{2}+1)<\infty$
yields the condition
\[
\theta_{2}\in D_{L}^{g}\triangleq D_{L}\cap(-\infty,\Theta_{L}-1)=(-\infty
,(\Theta_{L}-1)\wedge(\Theta_{L}/2)).
\]
Note that for $\theta_{2}\in D_{L}^{g}$ to be strictly positive and,
therefore, include the case $Q=P$, we need to have $\Theta_{L}>1.$ This, of
course, is a restriction on the structure of the jumps. For instance, if $L$
is a compound Poisson process with exponentially distributed jump sizes,
Example \ref{Example_Subordinators} (Case 2), we have that the jump sizes must
have a mean less than one. Note also that, if $\Theta_{L}>2$ then $D_{L}%
^{g}=D_{L}.$

Using expression $\left(  \ref{Equ_Y_Explicit_Q}\right)  $ and repeating the
previous arguments we obtain
\begin{align*}
\mathbb{E}_{Q}[\exp(Y(T))|\mathcal{F}_{t}]  &  =\exp\left\{  Y(t)e^{-\alpha
_{Y}(T-t)}+\frac{\mu_{Y}}{\alpha_{Y}}(1-e^{-\alpha_{Y}(T-t)})\right\} \\
&  \qquad\times\exp\left\{  \int_{0}^{T-t}\int_{0}^{\infty}(e^{ze^{-\alpha
_{Y}s}}-1)e^{\theta_{2}z}\ell(dz)ds\right\}  .
\end{align*}
Hence we have proved the following result:

\begin{proposition}
\label{Prop_Esscher_Geom}In the Esscher case for the spike component $Y$,
i.e., $\theta_{2}\in D_{L}^{g},$ $\beta_{2}=0$, and assuming $\Theta_{L}>1,$
the forward price $F_{Q}(t,T)$ in the geometric spot model $\left(
\ref{Equ_Geom_Model}\right)  $ is given by
\begin{align*}
F_{Q}(t,T)  &  =\Lambda_{g}(T)\exp\left(  X(t)e^{-\alpha_{X}(1-\beta
_{1})(T-t)}+Y(t)e^{-\alpha_{Y}(T-t)}\right) \\
&  \qquad\times\exp\left(  \frac{\mu_{X}+\theta_{1}}{\alpha_{X}(1-\beta_{1}%
)}(1-e^{-\alpha_{X}(1-\beta_{1})(T-t)})+\frac{\mu_{Y}}{\alpha_{Y}%
}(1-e^{-\alpha_{Y}(T-t)})\right) \\
&  \qquad\times\exp\left(  \frac{\sigma_{X}^{2}}{4\alpha_{X}(1-\beta_{1}%
)}(1-e^{-2\alpha_{X}(1-\beta_{1})(T-t)})+\int_{0}^{T-t}\int_{0}^{\infty
}(e^{ze^{-\alpha_{Y}s}}-1)e^{\theta_{2}z}\ell(dz)ds\right)  .
\end{align*}
and the risk premium for the forward price $R_{g,Q}^{F}(t,T)$ is given by%
\begin{align*}
R_{g,Q}^{F}(t,T)  &  =\mathbb{E}_{P}[S(T)|\mathcal{F}_{t}]\{\exp(R_{a,Q}%
^{F}(t,T))\\
&  \qquad\times\exp\left(  \frac{\sigma_{X}^{2}}{4\alpha_{X}(1-\beta_{1}%
)}(1-e^{-2\alpha_{X}(1-\beta_{1})(T-t)}-\frac{\sigma_{X}^{2}}{4\alpha_{X}%
}(1-e^{-2\alpha_{X}(T-t)}\right) \\
&  \qquad\times\exp\left(  -\frac{\kappa_{L}^{\prime}(\theta_{2})-\kappa
_{L}^{\prime}(0)}{\alpha_{Y}}(1-e^{-\alpha_{Y}(T-t)})\right) \\
&  \qquad\times\exp\left(  \int_{0}^{T-t}\int_{0}^{\infty}(e^{ze^{-\alpha
_{Y}s}}-1)(e^{\theta_{2}z}-1)\ell(dz)ds\right)  -1\},
\end{align*}
where $R_{a,Q}^{F}(t,T)$ is also understood under the assumption $\beta
_{2}=0.$
\end{proposition}

\begin{corollary}
Setting $\theta_{2}=0$ in Proposition \ref{Prop_Esscher_Geom} we get%
\begin{align*}
\mathbb{E}_{P}[S(T)|\mathcal{F}_{t}]  &  =\Lambda_{g}(T)\exp\left(
X(t)e^{-\alpha_{X}(T-t)}+Y(t)e^{-\alpha_{Y}(T-t)}\right) \\
&  \qquad\times\exp\left(  \frac{\mu_{X}}{\alpha_{X}}(1-e^{-\alpha_{X}%
(T-t)})+\frac{\mu_{Y}}{\alpha_{Y}}(1-e^{-\alpha_{Y}(T-t)})\right) \\
&  \qquad\times\exp\left(  \frac{\sigma_{X}^{2}}{4\alpha_{X}}(1-e^{-2\alpha
_{X}(1-\beta_{1})(T-t)})+\int_{0}^{T-t}\int_{0}^{\infty}(e^{ze^{-\alpha_{Y}s}%
}-1)\ell(dz)ds\right)  .
\end{align*}

\end{corollary}

The previous result is as far as one can go using "basic" martingale
techniques. In the general case, in order to find conditions under which
$\mathbb{E}_{Q}[\exp(Y(T))]<\infty$, and also to compute $\mathbb{E}_{Q}%
[\exp(Y(T))|\mathcal{F}_{t}],$ it is convenient to look at $Y$ as an affine
$Q$-semimartingale process with state space $\mathbb{R}_{+}$. In the sequel we
follow the notation in Kallsen and Muhle-Karbe \cite{KaMu-ka10}, but taking
into account that in our case the L\'{e}vy characteristics do not depend on
the time parameter. The L\'{e}vy-Kintchine triplets of $Y$ are
\begin{align*}
(\beta_{0}^{1},\gamma_{0}^{11},\varphi_{0}(dz))  &  =(\mu_{Y}+\kappa
_{L}^{\prime}(\theta_{2}),0,\boldsymbol{1}_{(0,\infty)}e^{\theta_{2}z}%
\ell(dz))\\
(\beta_{1}^{1},\gamma_{1}^{11},\varphi_{1}(dz))  &  =(-\alpha_{Y}(1-\beta
_{2}),0,\frac{\alpha_{Y}\beta_{2}}{\kappa_{L}^{\prime\prime}(\theta_{2}%
)}\boldsymbol{1}_{(0,\infty)}ze^{\theta_{2}z}\ell(dz)),
\end{align*}
which, according to Definition 2.4 in Kallsen and Muhle-Karbe \cite{KaMu-ka10}%
, are (strongly) admissible. Note that, as the triplets do not depend on $t,$
we can choose any truncation function. Moreover, as $Y$ is a special
$Q$-semimartingale, we choose the (pseudo) truncation function $h(x)=x.$
Associated to the previous L\'{e}vy-Kintchine triplets we have the following
L\'{e}vy exponents%
\begin{align*}
\Lambda_{0}^{\theta_{2},\beta_{2}}(u)  &  =\left(  \mu_{Y}+\kappa_{L}^{\prime
}(\theta_{2})\right)  u+\int_{0}^{\infty}(e^{uz}-1-uz)e^{\theta_{2}z}%
\ell(dz)\\
&  =\mu_{Y}u+\int_{0}^{\infty}(e^{uz}-1)e^{\theta_{2}z}\ell(dz)\\
&  =\mu_{Y}u+\kappa_{L}(u+\theta_{2})-\kappa_{L}(\theta_{2}),\\
\Lambda_{1}^{\theta_{2},\beta_{2}}(u)  &  =-\alpha_{Y}(1-\beta_{2}%
)u+\frac{\alpha_{Y}\beta_{2}}{\kappa_{L}^{\prime\prime}(\theta_{2})}\int
_{0}^{\infty}(e^{uz}-1-uz)ze^{\theta_{2}z}\ell(dz)\\
&  =-\alpha_{Y}u+\frac{\alpha_{Y}\beta_{2}}{\kappa_{L}^{\prime\prime}%
(\theta_{2})}\int_{0}^{\infty}(e^{uz}-1)ze^{\theta_{2}z}\ell(dz)\\
&  =-\alpha_{Y}u+\frac{\alpha_{Y}\beta_{2}}{\kappa_{L}^{\prime\prime}%
(\theta_{2})}\left(  \kappa_{L}^{\prime}(u+\theta_{2})-\kappa_{L}^{\prime
}(\theta_{2})\right)  .
\end{align*}
We have the following result.

\begin{theorem}
\label{Theo_Exp_Affine}Let $\bar{\beta}\in\lbrack0,1]^{2},\bar{\theta}\in
\bar{D}_{L}^{g}\triangleq\mathbb{R}\times D_{L}^{g}$. Assume $\Theta_{L}>1,$
that $\Psi_{\theta_{2},\beta_{2}}^{0},\Psi_{\theta_{2},\beta_{2}}^{1}\in
C^{1}([0,T],\mathbb{R})$ satisfy the ODE%
\begin{equation}%
\begin{array}
[c]{ccc}%
\frac{d}{dt}\Psi_{\theta_{2},\beta_{2}}^{1}(t)=\Lambda_{1}^{\theta_{2}%
,\beta_{2}}(\Psi_{\theta_{2},\beta_{2}}^{1}(t)), &  & \Psi_{\theta_{2}%
,\beta_{2}}^{1}(0)=1,\\
\frac{d}{dt}\Psi_{\theta_{2},\beta_{2}}^{0}(t)=\Lambda_{0}^{\theta_{2}%
,\beta_{2}}(\Psi_{\theta_{2},\beta_{2}}^{1}(t)), &  & \Psi_{\theta_{2}%
,\beta_{2}}^{0}(0)=0,
\end{array}
\label{Equ_ODE}%
\end{equation}
and that the integrability condition%
\begin{equation}
\kappa_{L}^{\prime\prime}(\theta_{2}+\sup_{t\in\lbrack0,T]}\Psi_{\theta
_{2},\beta_{2}}^{1}(t))=\int_{0}^{\infty}z^{2}\exp\{(\theta_{2}+\sup
_{t\in\lbrack0,T]}\Psi_{\theta_{2},\beta_{2}}^{1}(t))z\}\ell(dz)<\infty,
\label{Equ_Integrability_ODE}%
\end{equation}
holds. Then, we have that the forward price $F_{Q}(t,T)$ in the geometric spot
model $\left(  \ref{Equ_Geom_Model}\right)  $ is given by
\begin{align*}
F_{Q}(t,T)  &  =\Lambda_{g}(T)\exp\left(  X(t)e^{-\alpha_{X}(1-\beta
_{1})(T-t)}+Y(t)\Psi_{\theta_{2},\beta_{2}}^{1}(T-t)+\Psi_{\theta_{2}%
,\beta_{2}}^{0}(T-t)\right) \\
&  \qquad\times\exp\left(  \frac{\mu_{X}+\theta_{1}}{\alpha_{X}(1-\beta_{1}%
)}(1-e^{-\alpha_{X}(1-\beta_{1})(T-t)})+\frac{\sigma_{X}^{2}}{4\alpha
_{X}(1-\beta_{1})}(1-e^{-2\alpha_{X}(1-\beta_{1})(T-t)})\right)  ,
\end{align*}
and the risk premium for the forward price $R_{g,Q}^{F}(t,T)$ is given by%
\begin{align*}
R_{g,Q}^{F}(t,T)  &  =\mathbb{E}_{P}[S(T)|\mathcal{F}_{t}]\{\exp
(X(t)e^{-\alpha_{X}(T-t)}(e^{\alpha_{X}\beta_{1}(T-t)}-1))\\
&  \qquad\times\exp(Y(t)(\Psi_{\theta_{2},\beta_{2}}^{1}(T-t)-e^{-\alpha
_{Y}(T-t)}))\\
&  \qquad\times\exp\left(  \frac{\mu_{X}+\theta_{1}}{\alpha_{X}(1-\beta_{1}%
)}(1-e^{-\alpha_{X}(1-\beta_{1})(T-t)})-\frac{\mu_{X}}{\alpha_{X}%
}(1-e^{-\alpha_{X}(T-t)})\right) \\
&  \qquad\times\exp\left(  \frac{\sigma_{X}^{2}}{4\alpha_{X}(1-\beta_{1}%
)}(1-e^{-2\alpha_{X}(1-\beta_{1})(T-t)})-\frac{\sigma_{X}^{2}}{4\alpha_{X}%
}(1-e^{-2\alpha_{X}(T-t)})\right) \\
&  \qquad\times\exp\left(  \Psi_{\theta_{2},\beta_{2}}^{0}(T-t)-\frac{\mu_{Y}%
}{\alpha_{Y}}(1-e^{-\alpha_{Y}(T-t)})-\int_{0}^{T-t}\int_{0}^{\infty
}(e^{ze^{-\alpha_{Y}s}}-1)\ell(dz)ds\right)  -1\}.
\end{align*}

\end{theorem}

\begin{proof}
We apply Theorem 5.1 in Kallsen and Muhle-Karbe \cite{KaMu-ka10}. Note that
making the change of variable $t\rightarrow T-t$ the ODE $\left(
\ref{Equ_ODE}\right)  $ is reduced to the one appearing in items 2. and 3. of
Theorem 5.1. The integrability assumption $\left(  \ref{Equ_Integrability_ODE}%
\right)  $ implies conditions 1. and 5., in Theorem 5.1, and condition 4. is
trivially satisfied because $Y(0)$ is deterministic. Hence, the conclusion of
that theorem, with $p=1,$ holds and we get%
\begin{equation}
\mathbb{E}_{Q}[\exp(Y(T))|\mathcal{F}_{t}]=\exp\left(  Y(t)\Psi_{\theta
_{2},\beta_{2}}^{1}(T-t)+\Psi_{\theta_{2},\beta_{2}}^{0}(T-t)\right)  ,\quad
t\in\lbrack0,T]. \label{Equ_Exp_Affine_Formula}%
\end{equation}
The result now follows easily.
\end{proof}

\begin{remark}
\label{RemarkPsi0}Equation $\left(  \ref{Equ_ODE}\right)  $ is called a
generalised Riccati equation in the literature. Note that the equation for
$\Psi_{\theta_{2},\beta_{2}}^{0}(t)$ is trivially solved, once we know
$\Psi_{\theta_{2},\beta_{2}}^{1}(t),$ by
\[
\Psi_{\theta_{2},\beta_{2}}^{0}(t)=\int_{0}^{t}\Lambda_{0}^{\theta_{2}%
,\beta_{2}}(\Psi_{\theta_{2},\beta_{2}}^{1}(s))ds.
\]
Hence, the problem is really reduced to study the equation for $\Psi
_{\theta_{2},\beta_{2}}^{1}(t).$
\end{remark}

\begin{remark}
The Esscher case can be obtained from Theorem \ref{Theo_Exp_Affine}, as
$\Psi_{\theta_{2},0}^{1}(t)=e^{-\alpha_{Y}t}$ and
\[
\Psi_{\theta_{2},0}^{0}(t)=\frac{\mu_{Y}}{\alpha_{Y}}(1-e^{-\alpha_{Y}t}%
)+\int_{0}^{t}\int_{0}^{\infty}(e^{ze^{-\alpha_{Y}s}}-1)e^{\theta_{2}z}%
\ell(dz)ds,
\]
solve
\[%
\begin{array}
[c]{ccc}%
\frac{d}{dt}\Psi_{\theta_{2},0}^{1}(t)=-\alpha_{Y}\Psi_{\theta_{2},0}%
^{1}(t), &  & \Psi_{\theta_{2},0}^{1}(0)=1,\\
\frac{d}{dt}\Psi_{\theta_{2},0}^{0}(t)=\mu_{Y}\Psi_{\theta_{2},0}%
^{1}(t)+\kappa_{L}(\Psi_{\theta_{2},0}^{1}(t)+\theta_{2})-\kappa_{L}%
(\theta_{2}), &  & \Psi_{\theta_{2},0}^{0}(0)=0.
\end{array}
\]
As $\sup_{t\in\lbrack0,T]}\Psi_{\theta_{2},0}^{1}(t)=1,$ the integrability
condition $\left(  \ref{Equ_Integrability_ODE}\right)  $ is satisfied because
$\theta_{2}\in D_{L}^{g}.$
\end{remark}

In general, one cannot find explicit solutions for the non-linear differential
equation $\left(  \ref{Equ_ODE}\right)  $ in Theorem \ref{Theo_Exp_Affine} and
has to rely on numerical techniques. However, the main problem that we find is
that the maximal domain of definition of $\Psi_{\theta_{2},\beta_{2}}^{0}$ and
$\Psi_{\theta_{2},\beta_{2}}^{1}$ may be a proper subset of $[0,\infty),$ in
particular when $\beta_{2}$ is close to $1$. As we are particularly interested
in the solution of $\left(  \ref{Equ_ODE}\right)  $ for large $T$, we shall
give a general sufficient criterion for global (defined for any $t>0$)
existence and uniqueness of the solution of $\left(  \ref{Equ_ODE}\right)  $.
The next theorem classifies the behaviour of the solutions of $\left(
\ref{Equ_ODE}\right)  $.

\begin{theorem}
\label{Theorem_Main_Geom}Assume that $\Theta_{L}>1.$ For any $\delta>0,$ the
system of ODEs ($\ref{Equ_ODE}$) with $\beta_{2}\in(0,1)$ and
\[
\theta_{2}\in D_{L}^{g}(\delta)\triangleq(-\infty,(\Theta_{L}-1-\delta
)\wedge(\Theta_{L}/2))
\]
admits a unique local solution $\Psi_{\theta_{2},\beta_{2}}^{0}(t)$ and
$\Psi_{\theta_{2},\beta_{2}}^{1}(t)$. In addition, let $u^{\ast}(\theta
_{2},\beta_{2})$ be the unique strictly positive solution of the following
equation%
\begin{equation}
u=\frac{\beta_{2}}{\kappa_{L}^{\prime\prime}(\theta_{2})}\left(  \kappa
_{L}^{\prime}(u+\theta_{2})-\kappa_{L}^{\prime}(\theta_{2})\right)  .
\label{Equ_ZeroVecField_Lambda1}%
\end{equation}
The behaviour of $\Psi_{\theta_{2},\beta_{2}}^{0}(t)$ and $\Psi_{\theta
_{2},\beta_{2}}^{1}(t)$ is characterised as follows:

\begin{enumerate}
\item If $u^{\ast}(\theta_{2},\beta_{2})>1,$ then $\Psi_{\theta_{2},\beta_{2}%
}^{0}(t)$ and $\Psi_{\theta_{2},\beta_{2}}^{1}(t)$ are globally defined,
satisfy
\[
0<\Psi_{\theta_{2},\beta_{2}}^{1}(t)\leq1,\qquad0\leq\Psi_{\theta_{2}%
,\beta_{2}}^{0}(t)\leq\int_{0}^{\infty}\Lambda_{0}^{\theta_{2},\beta_{2}}%
(\Psi_{\theta_{2},\beta_{2}}^{1}(s))ds<\infty,
\]
and%
\begin{align}
\lim_{t\rightarrow\infty}\frac{1}{t}\log(\Psi_{\theta_{2},\beta_{2}}^{1}(t))
&  =-\alpha_{Y}(1-\beta_{2}),\label{Equ_Exp_Psi1}\\
\lim_{t\rightarrow\infty}\Psi_{\theta_{2},\beta_{2}}^{0}(t)  &  =\int
_{0}^{\infty}\Lambda_{0}^{\theta_{2},\beta_{2}}(\Psi_{\theta_{2},\beta_{2}%
}^{1}(s))ds<\infty. \label{Equ_Infinity_Psi0}%
\end{align}

\item If $u^{\ast}(\theta_{2},\beta_{2})=1,$ then $\Psi_{\theta_{2},\beta_{2}%
}^{1}(t)\equiv1$ and $\Psi_{\theta_{2},\beta_{2}}^{0}(t)=\{\mu_{Y}+\kappa
_{L}(1+\theta_{2})-\kappa_{L}(\theta_{2})\}t.$

\item If $u^{\ast}(\theta_{2},\beta_{2})<1$, then the maximal domain of
definition of $\Psi_{\theta_{2},\beta_{2}}^{0}(t)$ and $\Psi_{\theta_{2}%
,\beta_{2}}^{1}(t)$ is $[0,t_{\infty}),$ where%
\[
0<t_{\infty}=\int_{1}^{\Theta_{L}-\theta_{2}}(\Lambda_{1}^{\theta_{2}%
,\beta_{2}}(u))^{-1}du<\infty.
\]
In addition,%
\[
\lim_{t\uparrow t_{\infty}}\Psi_{\theta_{2},\beta_{2}}^{1}(t)=\Theta
_{L}-\theta_{2},\quad\lim_{t\uparrow t_{\infty}}\Psi_{\theta_{2},\beta_{2}%
}^{0}(t)=\int_{0}^{t_{\infty}}\Lambda_{0}^{\theta_{2},\beta_{2}}(\Psi
_{\theta_{2},\beta_{2}}^{1}(s))ds,
\]
where the previous integral is non negative and may be finite or infinite.
\end{enumerate}
\end{theorem}

\begin{proof}
We have to study the vector field%
\[
\Lambda_{1}^{\theta_{2},\beta_{2}}(u)=-\alpha_{Y}u+\frac{\alpha_{Y}\beta_{2}%
}{\kappa_{L}^{\prime\prime}(\theta_{2})}\int_{0}^{\infty}(e^{uz}%
-1)ze^{\theta_{2}z}\ell(dz),\quad\beta_{2}\in\lbrack0,1],\theta_{2}\in
D_{L}^{g}.
\]
Consider
\[
\mathcal{D}(\Lambda_{1}^{\theta_{2},\beta_{2}})\triangleq\mathrm{int}%
(\{u\in\mathbb{R}:\Lambda_{1}^{\theta_{2},\beta_{2}}(u)<\infty\})=\mathrm{int}%
(\{u\in\mathbb{R}:\kappa_{L}^{\prime}(u+\theta_{2})<\infty\})=(-\infty
,\Theta_{L}-\theta_{2}),
\]
and, for any $\delta>0,$ define
\begin{align*}
\mathcal{D}_{\delta}  &  \triangleq\mathrm{int}(%
{\displaystyle\bigcap\limits_{\beta_{2}\in\lbrack0,1],\theta_{2}\in D_{L}%
^{g}(\delta)}}
\mathcal{D}(\Lambda_{1}^{\theta_{2},\beta_{2}}))=(-\infty,\Theta_{L}%
-((\Theta_{L}-1-\delta)\wedge(\Theta_{L}/2)))\\
&  =(-\infty,(1+\delta)\vee(\Theta_{L}/2))).
\end{align*}
On the other hand, for $u,v\in\mathcal{D}(\Lambda_{1}^{\theta_{2},\beta_{2}%
}),$ one has that
\[
\left\vert \Lambda_{1}^{\theta_{2},\beta_{2}}(u)-\Lambda_{1}^{\theta_{2}%
,\beta_{2}}(v)\right\vert \leq\alpha_{Y}\left\vert u-v\right\vert
+\frac{\alpha_{Y}\beta_{2}}{\kappa_{L}^{\prime\prime}(\theta_{2})}\int
_{0}^{\infty}\left\vert e^{uz}-e^{vz}\right\vert ze^{\theta_{2}z}\ell(dz),
\]
and%
\[
\int_{0}^{\infty}\left\vert e^{uz}-e^{vz}\right\vert ze^{\theta_{2}z}%
\ell(dz)\leq|u-v|\int_{0}^{\infty}e^{(u\vee v+\theta_{2})z}z^{2}\ell(dz),
\]
Moreover, note that
\[
\mathrm{int}(\{u\in\mathbb{R}:\int_{0}^{\infty}z^{2}e^{(u+\theta_{2})z}%
\ell(dz)<\infty\})=(-\infty,\Theta_{L}-\theta_{2})=\mathcal{D}(\Lambda
_{1}^{\theta_{2},\beta_{2}}).
\]
Hence, the vector field $\Lambda_{1}^{\theta_{2},\beta_{2}}(u),\theta_{2}\in
D_{L}^{g}(\delta),\beta_{2}\in\lbrack0,1]$ is well defined (i.e., finite) and
locally Lipschitz in $\mathcal{D}_{\delta}.$ As the initial condition for
$\Psi_{\theta_{2},\beta_{2}}^{1}(t)$ is $\Psi_{\theta_{2},\beta_{2}}^{1}%
(0)=1$, it is natural to require that $1\in\mathcal{D}_{\delta}$ and this is
precisely the role of $\delta>0.$ Then, by Picard-Lindel\"{o}f Theorem, see
Theorem 3.1, pag. 18, in Hale \cite{Ha69}, we have local existence and
uniqueness for $\Psi_{\theta_{2},\beta_{2}}^{1}(t)$ and $\Psi_{\theta
_{2},\beta_{2}}^{1}(0)\in\mathcal{D}_{\delta}.$ In addition, we have that
$0\in\mathcal{D}_{\delta}$ and, hence, we have local existence and uniqueness
for solutions of $\Psi_{\theta_{2},\beta_{2}}^{1}(t)$ with $\Psi_{\theta
_{2},\beta_{2}}^{1}(0)=0.$ As $\Lambda_{1}^{\theta_{2},\beta_{2}}(0)=0,$ we
have that $\Psi_{\theta_{2},\beta_{2}}^{1}(t)\equiv0$ is the unique global
solution of equation $\left(  \ref{Equ_ODE}\right)  $ starting at 0. As a
consequence, it is sufficient to study the vector field $\Lambda_{1}%
^{\theta_{2},\beta_{2}}(u)$ for $u\geq0,$ because any solution of equation
$\left(  \ref{Equ_ODE}\right)  $ with $\Psi_{\theta_{2},\beta_{2}}^{1}(0)=1$
cannot cross to the negative real line without contradicting the uniqueness
result at $0.$ The unicity of $\Psi_{\theta_{2},\beta_{2}}^{0}(t)$ trivially
follows from that of $\Psi_{\theta_{2},\beta_{2}}^{1}(t).$ The next step is to
study the zeros of $\Lambda_{1}^{\theta_{2},\beta_{2}}(u),u\in\mathcal{D}%
_{\delta}\cap\lbrack0,\infty).$ We have to solve the non-linear equation%
\begin{equation}
0=\Lambda_{1}^{\theta_{2},\beta_{2}}(u)=-\alpha_{Y}u+\frac{\alpha_{Y}\beta
_{2}}{\kappa_{L}^{\prime\prime}(\theta_{2})}\left(  \kappa_{L}^{\prime
}(u+\theta_{2})-\kappa_{L}^{\prime}(\theta_{2})\right)  .
\label{Equ_VectorField_Lamb1}%
\end{equation}

Note that equation $\left(  \ref{Equ_VectorField_Lamb1}\right)  $ has the
trivial solution $u=0.$ As the first and second derivatives of $\Lambda
_{1}^{\theta_{2},\beta_{2}}(u)$ are
\begin{align*}
\frac{d}{du}\Lambda_{1}^{\theta_{2},\beta_{2}}(u)  &  =-\alpha_{Y}%
+\frac{\alpha_{Y}\beta_{2}}{\kappa_{L}^{\prime\prime}(\theta_{2})}\kappa
_{L}^{\prime\prime}(u+\theta_{2}),\\
\frac{d^{2}}{du^{2}}\Lambda_{1}^{\theta_{2},\beta_{2}}(u)  &  =\frac
{\alpha_{Y}\beta_{2}}{\kappa_{L}^{\prime\prime}(\theta_{2})}\kappa_{L}%
^{(3)}(u+\theta_{2})>0,
\end{align*}
we have that there exists a unique $0<u^{\ast}(\theta_{2},\beta_{2}%
)<\Theta_{L}-\theta_{2}$ for $\theta_{2}\in D_{L}^{g}(\delta)$ and $\beta
_{2}\in(0,1)$ such that equation $\left(  \ref{Equ_VectorField_Lamb1}\right)
$ is satisfied. Moreover $\Lambda_{1}^{\theta_{2},\beta_{2}}(u)<0$ for
$u\in(0,u^{\ast}(\theta_{2},\beta_{2}))$ and $\Lambda_{1}^{\theta_{2}%
,\beta_{2}}(u)>0$ for $(u^{\ast}(\theta_{2},\beta_{2}),\Theta_{L}-\theta
_{2}).$ When $\beta_{2}\downarrow0,u^{\ast}(\theta_{2},\beta_{2})$ converges
to $\Theta_{L}-\theta_{2}.$ On the other hand, when $\beta_{2}\uparrow
1,u^{\ast}(\theta_{2},\beta_{2})$ converges to zero. Therefore, we have three
possible cases to discuss

\begin{itemize}
\item \textbf{Case }$\boldsymbol{1}:$ If $u^{\ast}(\theta_{2},\beta_{2})>1,$
then $\Psi_{\theta_{2},\beta_{2}}^{1}(t)$ will monotonically converge to $0$
and, by uniqueness of solutions, it will take an infinite amount of time to
reach $0.$ Hence, $\Psi_{\theta_{2},\beta_{2}}^{1}(t)$ will be a globally
defined bounded solution. The exponential rate of convergence of $\Psi
_{\theta_{2},\beta_{2}}^{1}(t)$ to zero, equation \ref{Equ_Exp_Psi1}, follows
by applying H\^{o}pital's rule to
\begin{align*}
\lim_{t\rightarrow\infty}t^{-1}\log(\Psi_{\theta_{2},\beta_{2}}^{1}(t))  &
=\lim_{t\rightarrow\infty}\frac{\frac{d}{dt}\Psi_{\theta_{2},\beta_{2}}%
^{1}(t)}{\Psi_{\theta_{2},\beta_{2}}^{1}(t)}\\
&  =\lim_{t\rightarrow\infty}\frac{\Lambda_{1}^{\theta_{2},\beta_{2}}%
(\Psi_{\theta_{2},\beta_{2}}^{1}(t))}{\Psi_{\theta_{2},\beta_{2}}^{1}(t)}\\
&  =\lim_{t\rightarrow\infty}\frac{-\alpha_{Y}\Psi_{\theta_{2},\beta_{2}}%
^{1}(t)+\frac{\alpha_{Y}\beta_{2}}{\kappa_{L}^{\prime\prime}(\theta_{2}%
)}\{\kappa_{L}^{\prime}(\Psi_{\theta_{2},\beta_{2}}^{1}(t)+\theta_{2}%
)-\kappa_{L}^{\prime}(+\theta_{2})\}}{\Psi_{\theta_{2},\beta_{2}}^{1}(t)}\\
&  =-\alpha_{Y}+\frac{\alpha_{Y}\beta_{2}}{\kappa_{L}^{\prime\prime}%
(\theta_{2})}\lim_{t\rightarrow\infty}\frac{\int_{0}^{1}\kappa_{L}%
^{\prime\prime}(\theta_{2}+\lambda\Psi_{\theta_{2},\beta_{2}}^{1}%
(t))d\lambda\Psi_{\theta_{2},\beta_{2}}^{1}(t)}{\Psi_{\theta_{2},\beta_{2}%
}^{1}(t)}\\
&  =-\alpha_{Y}(1-\beta_{2}).
\end{align*}
It follows that $\Psi_{\theta_{2},\beta_{2}}^{0}(t)$ will be also globally
defined and, as $\Lambda_{0}^{\theta_{2},\beta_{2}}(u)=\mu_{Y}u+\int_{0}%
^{1}\kappa_{L}^{\prime}(\theta_{2}+\lambda u)d\lambda>0$ for $u\in(0,1),$ by
monotone convergence%
\[
\lim_{t\rightarrow\infty}\Psi_{\theta_{2},\beta_{2}}^{0}(t)=\int_{0}^{\infty
}\Lambda_{0}^{\theta_{2},\beta_{2}}(\Psi_{\theta_{2},\beta_{2}}^{1}(s))ds.
\]
To show that the previous integral is actually finite, it suffices to prove
that $\Lambda_{0}^{\theta_{2},\beta_{2}}(\Psi_{\theta_{2},\beta_{2}}^{1}(t))$
converges to zero faster than $t^{-(1+\varepsilon)},$ for some $\varepsilon
>0,$ when $t$ tends to infinity. We have that \
\begin{align*}
\Lambda_{0}^{\theta_{2},\beta_{2}}(\Psi_{\theta_{2},\beta_{2}}^{1}(t))  &
=\mu_{Y}\Psi_{\theta_{2},\beta_{2}}^{1}(t)+\kappa_{L}(\Psi_{\theta_{2}%
,\beta_{2}}^{1}(t)+\theta_{2})-\kappa_{L}(\theta_{2})\\
&  =\{\mu_{Y}+\int_{0}^{1}\kappa_{L}^{\prime}(\theta_{2}+\lambda\Psi
_{\theta_{2},\beta_{2}}^{1}(t))d\lambda\}\Psi_{\theta_{2},\beta_{2}}^{1}(t),
\end{align*}
and%
\begin{align*}
\frac{d}{dt}\left(  \Lambda_{0}^{\theta_{2},\beta_{2}}(\Psi_{\theta_{2}%
,\beta_{2}}^{1}(t))\right)   &  =\left.  \frac{d}{du}\Lambda_{0}^{\theta
_{2},\beta_{2}}(u)\right\vert _{u=\Psi_{\theta_{2},\beta_{2}}^{1}(t)}\frac
{d}{dt}\Psi_{\theta_{2},\beta_{2}}^{1}(t)\\
&  =\{\mu_{Y}+\kappa_{L}^{\prime}(\theta_{2}+\Psi_{\theta_{2},\beta_{2}}%
^{1}(t))\}\\
&  \qquad\times\{-\alpha_{Y}\Psi_{\theta_{2},\beta_{2}}^{1}(t)+\frac
{\alpha_{Y}\beta_{2}}{\kappa_{L}^{\prime\prime}(\theta_{2})}\left(  \kappa
_{L}^{\prime}(\Psi_{\theta_{2},\beta_{2}}^{1}(t)+\theta_{2})-\kappa
_{L}^{\prime}(\theta_{2})\right)  \}\\
&  =\{\mu_{Y}+\kappa_{L}^{\prime}(\theta_{2}+\Psi_{\theta_{2},\beta_{2}}%
^{1}(t))\}\\
&  \qquad\times\{-\alpha_{Y}+\frac{\alpha_{Y}\beta_{2}}{\kappa_{L}%
^{\prime\prime}(\theta_{2})}\int_{0}^{1}\kappa_{L}^{\prime\prime}(\theta
_{2}+\lambda\Psi_{\theta_{2},\beta_{2}}^{1}(t))d\lambda\}\Psi_{\theta
_{2},\beta_{2}}^{1}(t).
\end{align*}
By H\^{o}pital's rule and equation $\left(  \ref{Equ_Exp_Psi1}\right)  $%
\begin{align*}
\lim_{t\rightarrow\infty}t^{(1+\varepsilon)}\Lambda_{0}^{\theta_{2},\beta_{2}%
}(\Psi_{\theta_{2},\beta_{2}}^{1}(t))  &  =\lim_{t\rightarrow\infty
}(1+\varepsilon)t^{\varepsilon}\left(  \frac{-\frac{d}{dt}\Lambda_{0}%
^{\theta_{2},\beta_{2}}(\Psi_{\theta_{2},\beta_{2}}^{1}(t))}{\left(
\Lambda_{0}^{\theta_{2},\beta_{2}}(\Psi_{\theta_{2},\beta_{2}}^{1}(t))\right)
^{2}}\right)  ^{-1}\\
&  =(1+\varepsilon)\frac{\mu_{Y}+\kappa_{L}^{\prime}(\theta_{2})}{\alpha
_{Y}(1-\beta_{2})}\lim_{t\rightarrow\infty}t^{\varepsilon}\Psi_{\theta
_{2},\beta_{2}}^{1}(t)=0,
\end{align*}
and we can conclude that equation $\left(  \ref{Equ_Infinity_Psi0}\right)  $ holds.

\item \textbf{Case }$\boldsymbol{2}:$ If $u^{\ast}(\theta_{2},\beta_{2})=1,$
then $\Psi_{\theta_{2},\beta_{2}}^{1}(t)\equiv1,$ will be the unique global
solution and
\[
\Psi_{\theta_{2},\beta_{2}}^{0}(t)=\int_{0}^{t}\Lambda_{0}^{\theta_{2}%
,\beta_{2}}(\Psi_{\theta_{2},\beta_{2}}^{1}(s))ds=\{\mu_{Y}+\kappa
_{L}(1+\theta_{2})-\kappa_{L}(\theta_{2})\}t.
\]

\item \textbf{Case }$\boldsymbol{3}:$ If $u^{\ast}(\theta_{2},\beta_{2})<1,$
then $\Psi_{\theta_{2},\beta_{2}}^{1}(t)$ will increase monotonically to
$\Theta_{L}-\theta_{2},$ because the vector field $\Lambda_{1}^{\theta
_{2},\beta_{2}}$ is strictly positive in $[1,\Theta_{L}-\theta_{2})$.
Separating variables an integrating the equation for $\Psi_{\theta_{2}%
,\beta_{2}}^{1}(t)$ with $\Psi_{\theta_{2},\beta_{2}}^{1}(0)=1$ we get that
the maximal domain of definition of $\Psi_{\theta_{2},\beta_{2}}^{1}(t)$ is
$[0,t_{\infty})$ with
\[
t_{\infty}\triangleq\int_{1}^{\Theta_{L}-\theta_{2}}(\Lambda_{1}^{\theta
_{2},\beta_{2}}(u))^{-1}du.
\]
To show that $t_{\infty}$ is actually finite we have to distinguish between
the case $\Theta_{L}<\infty$ and $\Theta_{L}=\infty.$ If $\Theta_{L}<\infty,$
then $(\Lambda_{1}^{\theta_{2},\beta_{2}}(u))^{-1}$ is bounded in
$[1,\Theta_{L}-\theta_{2})$ and the integral is obviously finite. If
$\Theta_{L}=\infty$ we have to ensure that $(\Lambda_{1}^{\theta_{2},\beta
_{2}}(u))^{-1}$ converges to zero fast enough when $u$ tends to infinity. Note
that, by monotone convergence, one has that
\begin{align*}
\lim_{\theta\rightarrow\infty}\kappa_{L}(\theta) &  =\int_{0}^{\infty}%
\lim_{\theta\rightarrow\infty}(e^{\theta z}-1)\ell(dz)=\infty,\\
\lim_{\theta\rightarrow\infty}\kappa_{L}^{(n)}(\theta) &  =\int_{0}^{\infty
}\lim_{\theta\rightarrow\infty}z^{n}e^{\theta z}\ell(dz)=\infty,\quad n\geq1.
\end{align*}
For any $0<\varepsilon<1,$ we have that
\begin{align*}
\lim_{u\rightarrow\infty}u^{-(1+\varepsilon)}\Lambda_{1}^{\theta_{2},\beta
_{2}}(u) &  =\lim_{u\rightarrow\infty}u^{-(1+\varepsilon)}\{-\alpha_{Y}%
u+\frac{\alpha_{Y}\beta_{2}}{\kappa_{L}^{\prime\prime}(\theta_{2})}(\kappa
_{L}^{\prime}(u+\theta_{2})-\kappa_{L}^{\prime}(\theta_{2}))\}\\
&  =\frac{\alpha_{Y}\beta_{2}}{\kappa_{L}^{\prime\prime}(\theta_{2})}%
\lim_{u\rightarrow\infty}\frac{\kappa_{L}^{\prime}(u+\theta_{2})}%
{u^{(1+\varepsilon)}}=\frac{\alpha_{Y}\beta_{2}}{(1+\varepsilon)\kappa
_{L}^{\prime\prime}(\theta_{2})}\lim_{u\rightarrow\infty}\frac{\kappa
_{L}^{\prime\prime}(u+\theta_{2})}{u^{\varepsilon}}\\
&  =\frac{\alpha_{Y}\beta_{2}}{(1+\varepsilon)\varepsilon\kappa_{L}%
^{\prime\prime}(\theta_{2})}\lim_{u\rightarrow\infty}u^{1-\varepsilon}%
\kappa_{L}^{(3)}(u+\theta_{2})=\infty,
\end{align*}
which yields that the integral defining $t_{\infty}$ is finite. According to
Remark \ref{RemarkPsi0}, we have that
\begin{equation}
\lim_{t\rightarrow t_{\infty}}\Psi_{\theta_{2},\beta_{2}}^{0}(t)=\int
_{0}^{t_{\infty}}\Lambda_{0}^{\theta_{2},\beta_{2}}(\Psi_{\theta_{2},\beta
_{2}}^{1}(s))ds,\label{EquIntegralPsi0}%
\end{equation}
which may be finite or infinite depending, of course, on how fast $\Lambda
_{0}^{\theta_{2},\beta_{2}}(\Psi_{\theta_{2},\beta_{2}}^{1}(s))$ diverges to
infinity when $s$ approaches to $t_{\infty}.$
\end{itemize}
\end{proof}

As it does not seem possible to give simple conditions for the finiteness (or
not) of the integral $\left(  \ref{EquIntegralPsi0}\right)  $\ and it is not
relevant in the discussion to follow, we do not proceed further in the analysis.

\begin{remark}
If $\beta_{2}=0,$ then $\Psi_{\theta_{2},0}^{1}(t)=e^{-\alpha_{Y}t}$ and
\begin{align*}
\Psi_{\theta_{2},0}^{0}(t) &  =\int_{0}^{t}\mu_{Y}e^{-\alpha_{Y}s}ds+\int
_{0}^{t}\kappa_{L}(e^{-\alpha_{Y}s}+\theta_{2})-\kappa_{L}(\theta_{2})ds\\
&  =\frac{\mu_{Y}}{\alpha_{Y}}(1-e^{-\alpha_{Y}t})+\int_{0}^{t}\int
_{0}^{\infty}(e^{ze^{-\alpha_{Y}s}}-1)e^{\theta_{2}z}\ell(dz)ds.
\end{align*}
Obviously $\lim_{t\rightarrow\infty}e^{\alpha_{Y}t}\Psi_{\theta_{2},0}%
^{1}(t)=1$ and
\[
\lim_{t\rightarrow\infty}\Psi_{\theta_{2},0}^{0}(t)=\frac{\mu_{Y}}{\alpha_{Y}%
}+\int_{0}^{\infty}\int_{0}^{\infty}(e^{ze^{-\alpha_{Y}s}}-1)e^{\theta_{2}%
z}\ell(dz)ds<\infty.
\]
Note that%
\begin{align*}
\int_{0}^{\infty}\int_{0}^{\infty}(e^{ze^{-\alpha_{Y}s}}-1)e^{\theta_{2}z}%
\ell(dz)ds &  =\int_{0}^{\infty}\left(  \int_{0}^{\infty}(e^{ze^{-\alpha_{Y}%
s}}-1)ds\right)  e^{\theta_{2}z}\ell(dz)\\
&  \leq\int_{0}^{\infty}\left(  \int_{0}^{\infty}\left(  \int_{0}%
^{1}e^{\lambda ze^{-\alpha_{Y}s}}d\lambda\right)  ze^{-\alpha_{Y}s}ds\right)
e^{\theta_{2}z}\ell(dz)\\
&  \leq\frac{1}{\alpha_{Y}}\int_{0}^{\infty}ze^{(1+\theta_{2})z}\ell
(dz)=\frac{\kappa_{L}^{\prime}(1+\theta_{2})}{\alpha_{Y}}<\infty.
\end{align*}
If $\beta_{2}=1,$ we have that
\[
\frac{d}{du}\Lambda_{1}^{\theta_{2},\beta_{2}}(u)=-\alpha_{Y}+\frac{\alpha
_{Y}}{\kappa_{L}^{\prime\prime}(\theta_{2})}\kappa_{L}^{\prime\prime}%
(u+\theta_{2})=\alpha_{Y}(\frac{\kappa_{L}^{\prime\prime}(u+\theta_{2}%
)}{\kappa_{L}^{\prime\prime}(\theta_{2})}-1)>0,
\]
for $u\in(0,\Theta_{L}-\theta_{2}),$ which yields that $\Psi_{\theta_{2}%
,1}^{1}(t)>1$ and monotonically diverges to infinity.
\end{remark}

Although the previous result characterizes the behaviour of the solution of
the ODE $\left(  \ref{Equ_ODE}\right)  $ for different values of $(\theta
_{2},\beta_{2})$ in terms of $u^{\ast}(\theta_{2},\beta_{2}),$ usually one
cannot find $u^{\ast}(\theta_{2},\beta_{2})$ analytically and, given
$(\theta_{2},\beta_{2}),$ equation $\left(  \ref{Equ_ZeroVecField_Lambda1}%
\right)  $ must be solved numerically to know whether the solution associated
to equation $\left(  \ref{Equ_ODE}\right)  $ is bounded or not. Hence, the
following corollary of Theorem \ref{Theorem_Main_Geom} may be helpful in practice.

\begin{corollary}
\label{Corollary_SufficientCond}Under the hypothesis of Theorem
\ref{Theorem_Main_Geom} and for $\theta_{2}\in D_{L}^{g}(\delta)$ fixed, a
sufficient condition for $u^{\ast}(\theta_{2},\beta_{2})>1$ is that
\begin{equation}
\beta_{2}<\frac{\kappa_{L}^{\prime\prime}(\theta_{2})}{\kappa_{L}^{\prime
}(1+\theta_{2})-\kappa_{L}^{\prime}(\theta_{2})}. \label{Equ_Suf_Cond_Geom}%
\end{equation}

\end{corollary}

\begin{proof}
Assume $\theta_{2}\in D_{L}^{g}(\delta)$ fixed. According to the discussion in
the proof of Theorem \ref{Theorem_Main_Geom}, for any $\theta\in D_{L}%
^{g}(\delta)$ and $\beta_{2}\in(0,1)$ there exists a unique root $u^{\ast
}=u^{\ast}(\theta_{2},\beta_{2})$ of the vector field $\Lambda_{1}^{\theta
_{2},\beta_{2}}(u)$ defined by equation $\left(
\ref{Equ_ZeroVecField_Lambda1}\right)  $ and such that $\Lambda_{1}%
^{\theta_{2},\beta_{2}}(u)<0$ if $(0,u^{\ast}(\theta_{2},\beta_{2}))$ and
$\Lambda_{1}^{\theta_{2},\beta_{2}}(u)>0$ if $(u^{\ast}(\theta_{2},\beta
_{2}),\Theta_{L}-\theta_{2}).$ Now, note that
\[
\beta_{2}^{\ast}(1)\triangleq\frac{\kappa_{L}^{\prime\prime}(\theta_{2}%
)}{\kappa_{L}^{\prime}(1+\theta_{2})-\kappa_{L}^{\prime}(\theta_{2})},
\]
is such that $1=u^{\ast}(\theta_{2},\beta_{2}^{\ast}(1)).$ If $\beta_{2}%
<\beta_{2}^{\ast}(1)$ one has that%
\[
\Lambda_{1}^{\theta_{2},\beta_{2}}(1)=\alpha_{Y}(-1+\frac{\beta_{2}}%
{\kappa_{L}^{\prime\prime}(\theta_{2})}\kappa_{L}^{\prime}(1+\theta
_{2})-\kappa_{L}^{\prime}(\theta_{2}))<0,
\]
which yields that the unique root $u^{\ast}=u^{\ast}(\theta_{2},\beta_{2})$ of
the vector field $\Lambda_{1}^{\theta_{2},\beta_{2}}(u)$ must be strictly
greater than one and, therefore, we are in the case (1) of Theorem
\ref{Theorem_Main_Geom}.
\end{proof}

Next, we present two examples where we apply the previous results.

\begin{example}
We start by the simplest possible case. Assume that the L\'{e}vy measure is
$\delta_{\{1\}}(dz),$ that is, the L\'{e}vy process $L$ has only jumps of size
$1.$ In this case $\Theta_{L}=\infty$ and, hence, $D_{L}^{g}=\mathbb{R}.$ We
have that $\kappa_{L}(\theta_{2})=e^{\theta_{2}}-1$ and $\kappa_{L}%
^{(n)}(\theta_{2})=e^{\theta_{2}},n\in\mathbb{N}.$ Therefore,
\begin{align*}
\Lambda_{0}^{\theta_{2},\beta_{2}}(u)  &  =\mu_{Y}u+\kappa_{L}(u+\theta
_{2})-\kappa_{L}(\theta_{2})=\mu_{Y}u+(e^{u+\theta_{2}}-e^{\theta_{2}}),\\
\Lambda_{1}^{\theta_{2},\beta_{2}}(u)  &  =-\alpha_{Y}u+\frac{\alpha_{Y}%
\beta_{2}}{\kappa_{L}^{\prime\prime}(\theta_{2})}\left(  \kappa_{L}^{\prime
}(u+\theta_{2})-\kappa_{L}^{\prime}(\theta_{2})\right)  =-\alpha_{Y}%
u+\alpha_{Y}\beta_{2}(e^{u}-1).
\end{align*}
First, we have to solve%
\begin{align}
\frac{d}{dt}\Psi_{\theta_{2},\beta_{2}}^{1}(t)  &  =-\alpha_{Y}\Psi
_{\theta_{2},\beta_{2}}^{1}(t)+\alpha_{Y}\beta_{2}(e^{\Psi_{\theta_{2}%
,\beta_{2}}^{1}(t)}-1),\label{Equ_Psi1_delta_1}\\
\quad\Psi_{\theta_{2},\beta_{2}}^{1}(0)  &  =1.\nonumber
\end{align}
and then integrate $\Lambda_{0}^{\theta_{2},\beta_{2}}(\Psi_{\theta_{2}%
,\beta_{2}}^{1}(s))$ from $0$ to $t.$ Although equation $\left(
\ref{Equ_Psi1_delta_1}\right)  $ can be solved analytically, its solution is
given in implicit form and a numerical method is easier to use. In this
example, equation $\left(  \ref{Equ_ZeroVecField_Lambda1}\right)  $ reads%
\begin{equation}
u=\frac{\beta_{2}}{e^{\theta_{2}}}\left(  e^{u+\theta_{2}}-e^{\theta_{2}%
}\right)  =\beta_{2}(e^{u}-1), \label{EquRootExDirac}%
\end{equation}
which can only be solved numerically. Heuristically, if $\beta_{2}$ is close
to one the solution of the previous equation must be close to zero and, hence,
the solution $\Psi_{\theta_{2},\beta_{2}}^{1}(t)$ diverges to $\infty.$
Applying Corollary \ref{Corollary_SufficientCond} we can guarantee that
$\Psi_{\theta_{2},\beta_{2}}^{1}(t)$ converges to zero if
\[
\beta_{2}<\frac{\kappa_{L}^{\prime\prime}(\theta_{2})}{\kappa_{L}^{\prime
}(1+\theta_{2})-\kappa_{L}^{\prime}(\theta_{2})}=\frac{e^{\theta_{2}}%
}{e^{1+\theta_{2}}-e^{\theta_{2}}}=(e-1)^{-1}.
\]

\end{example}

\begin{example}
Assume that the L\'{e}vy measure is $\ell(dz)=ce^{-\lambda z}\boldsymbol{1}%
_{(0,\infty)},$ that is, $L$ is a compound Poisson process with intensity
$c/\lambda$ and exponentially distributed jumps with mean $1/\lambda.$ In this
case $\Theta_{L}=\lambda$ and, hence, $D_{L}^{g}=\mathbb{(-\infty}%
,(\lambda-1)\wedge(\lambda/2).$ We have that $\kappa_{L}(\theta_{2}%
)=\frac{c\theta_{2}}{\lambda(\lambda-\theta_{2})}$ and $\kappa_{L}%
^{(n)}(\theta_{2})=\frac{cn!}{(\lambda-\theta_{2})^{n+1}},n\in\mathbb{N}.$
Therefore,
\begin{align*}
\Lambda_{0}^{\theta_{2},\beta_{2}}(u)  &  =\mu_{Y}u+\kappa_{L}(u+\theta
_{2})-\kappa_{L}(\theta_{2})\\
&  =\mu_{Y}u+\frac{c(u+\theta_{2})}{\lambda(\lambda-\theta_{2}-u)}%
-\frac{c\theta_{2}}{\lambda(\lambda-\theta_{2})},\\
\Lambda_{1}^{\theta_{2},\beta_{2}}(u)  &  =-\alpha_{Y}u+\frac{\alpha_{Y}%
\beta_{2}}{\kappa_{L}^{\prime\prime}(\theta_{2})}\left(  \kappa_{L}^{\prime
}(u+\theta_{2})-\kappa_{L}^{\prime}(\theta_{2})\right) \\
&  =-\alpha_{Y}u+\frac{\alpha_{Y}\beta_{2}(\lambda-\theta_{2})^{3}}{2}\left\{
\frac{1}{(\lambda-\theta_{2}-u)^{2}}-\frac{1}{(\lambda-\theta_{2})^{2}%
}\right\}  .
\end{align*}
Hence, we have to solve
\begin{align*}
\frac{d}{dt}\Psi_{\theta_{2},\beta_{2}}^{1}(t)  &  =-\alpha_{Y}\Psi
_{\theta_{2},\beta_{2}}^{1}(t)+\frac{\alpha_{Y}\beta_{2}(\lambda-\theta
_{2})^{3}}{2}\left\{  \frac{1}{(\lambda-\theta_{2}-\Psi_{\theta_{2},\beta_{2}%
}^{1}(t))^{2}}-\frac{1}{(\lambda-\theta_{2})^{2}}\right\}  ,\\
\Psi_{\theta_{2},\beta_{2}}^{1}(0)  &  =1,
\end{align*}
and then integrate $\Lambda_{0}^{\theta_{2},\beta_{2}}(\Psi_{\theta_{2}%
,\beta_{2}}^{1}(s))$ from $0$ to $t.$ As in the previous example, there is an
analytic solution to this equation in implicit form, but it is easier to use a
numerical method. In this example, equation $\left(
\ref{Equ_ZeroVecField_Lambda1}\right)  $ reads%
\[
u=\beta_{2}\frac{(\lambda-\theta_{2})^{3}}{2}\left(  \frac{1}{(\lambda
-\theta_{2}-u)^{2}}-\frac{1}{(\lambda-\theta_{2})^{2}}\right)  ,
\]
which has roots%
\[
(u_{0},u_{-},u_{+})=(0,\frac{\lambda-\theta_{2}}{4}\left(  4-\beta_{2}%
-\sqrt{\beta_{2}^{2}+8\beta_{2}}\right)  ,\frac{\lambda-\theta_{2}}{4}\left(
4-\beta_{2}+\sqrt{\beta_{2}^{2}+8\beta_{2}}\right)  ).
\]
We are just interested in the root $u_{-}\in(0,\lambda-\theta_{2}),$ note that
$u_{+}>\lambda-\theta_{2}.$ The inequality $\lambda-\theta_{2}>u_{-}>1$ yields%
\begin{equation}
0<\beta_{2}<2\frac{(\lambda-\theta_{2}-1)^{2}}{(\lambda-\theta_{2}%
)(2(\lambda-\theta_{2})-1)}. \label{Equ_Cond_Beta2_Example}%
\end{equation}
Hence, for any $\theta_{2}\in D_{L}^{g}(\delta)$ and $\beta_{2}$ satisfying
$\left(  \ref{Equ_Cond_Beta2_Example}\right)  $ we can ensure global existence
and boundedness of $\Psi_{\theta_{2},\beta_{2}}^{0}(t)$ and $\Psi_{\theta
_{2},\beta_{2}}^{1}(t)$.
\end{example}

\subsubsection{Discussion on the risk premium}

For the study of the sign change we are going to abuse the notation, as in the
arithmetic spot price model, and we will denote $R_{g,Q}^{F}(t,\tau)\triangleq
R_{g,Q}^{F}(t,t+\tau),$ where $\tau=T-t$ is the time to maturity. We also fix
the parameters of the model under the historical measure $P,$ i.e., $\mu
_{X},\alpha_{X},\sigma_{X},\mu_{Y},$ and $\alpha_{Y},$ and study the possible
sign of $R_{g,Q}^{F}(t,\tau)$ in terms of the change of measure parameters,
i.e., $\bar{\beta}=(\beta_{1},\beta_{2})$ and $\bar{\theta}=(\theta_{1}%
,\theta_{2})$ and the time to maturity $\tau.$ As in the arithmetic model, the
present time just enters into the picture through the stochastic components
$X$ and $Y.$ We are also going to assume $\mu_{X}=\mu_{Y}=0.$ Analogously to
the arithmetic case, in this way the seasonality function $\Lambda_{g}$
accounts completely for the mean price level. We also assume that $\alpha
_{X}<\alpha_{Y},$ which means that the component accounting for the jumps
reverts the fastest. Finally, in the sequel, we are going to assume that we
are in the Case 1 of Theorem \ref{Theorem_Main_Geom}, i.e., the values
$\theta_{2},\beta_{2}$ are such that $u^{\ast}(\theta_{2},\beta_{2})>1,$ and
$\Psi_{\theta_{2},\beta_{2}}^{0}$ and $\Psi_{\theta_{2},\beta_{2}}^{1}$ are
globally defined and the exponential affine formula $\left(
\ref{Equ_Exp_Affine_Formula}\right)  $ holds.

The following lemma will help us in the discussion to follow.

\begin{lemma}
\label{Lemma_Main_RPG}If $\mu_{X}=\mu_{Y}=0$ and $\alpha_{X}<\alpha_{Y},$ we
have that the sign of the risk premium $R_{g,Q}^{F}(t,\tau)$ will be the same
as the sign of%
\begin{align}
\Sigma(t,\tau)  &  \triangleq X(t)e^{-\alpha_{X}\tau}(e^{\alpha_{X}\beta
_{1}\tau}-1)+Y(t)(\Psi_{\theta_{2},\beta_{2}}^{1}(\tau)-\Psi_{0,0}^{1}%
(\tau))\label{Equ_SigmaTau}\\
&  +\frac{\theta_{1}}{\alpha_{X}(1-\beta_{1})}(1-e^{-\alpha_{X}(1-\beta
_{1})\tau})+\frac{\sigma_{X}^{2}}{4\alpha_{X}}\Lambda(2\alpha_{X}\tau
,1-\beta_{2})\nonumber\\
&  +\Psi_{\theta_{2},\beta_{2}}^{0}(\tau)-\Psi_{0,0}^{0}(\tau),\nonumber
\end{align}
where $\Lambda(x,y)$ is the (non-negative) function defined in Lemma
\ref{Lemma_Main_RPA}. Moreover,%
\begin{align}
\lim_{\tau\rightarrow\infty}\Sigma(t,\tau)  &  =\frac{\theta_{1}}{\alpha
_{X}(1-\beta_{1})}+\frac{\sigma_{X}^{2}}{4\alpha_{X}}\frac{\beta_{1}}%
{1-\beta_{1}}\label{Equ_Sigma_Infinity}\\
&  +\int_{0}^{\infty}\kappa_{L}(\Psi_{\theta_{2},\beta_{2}}^{1}(t)+\theta
_{2})-\kappa_{L}(\theta_{2})-\kappa_{L}(e^{-\alpha_{Y}t})dt\nonumber\\
\lim_{\tau\rightarrow0}\frac{\partial}{\partial\tau}\Sigma(t,\tau)  &
=X(t)\alpha_{X}\beta_{1}+Y(t)\alpha_{Y}\beta_{2}\frac{\kappa_{L}^{\prime
}(1+\theta_{2})-\kappa_{L}^{\prime}(\theta_{2})}{\kappa_{L}^{\prime\prime
}(\theta_{2})}\label{Equ_DSigma_Zero}\\
&  +\theta_{1}+\kappa_{L}(1+\theta_{2})-\kappa_{L}(\theta_{2})-\kappa
_{L}(1)\nonumber
\end{align}

\end{lemma}

\begin{proof}
The result follows easily from Theorem \ref{Theo_Exp_Affine} and the following
computations with $\Psi_{\theta_{2},\beta_{2}}^{1}(\tau)$ and $\Psi
_{\theta_{2},\beta_{2}}^{0}(\tau)$. We have that%
\begin{align*}
\lim_{\tau\rightarrow0}\frac{d}{d\tau}\Psi_{\theta_{2},\beta_{2}}^{1}(\tau)
&  =\lim_{\tau\rightarrow0}\Lambda_{1}^{\theta_{2},\beta_{2}}(\Psi_{\theta
_{2},\beta_{2}}^{1}(\tau))=\Lambda_{1}^{\theta_{2},\beta_{2}}(1)\\
&  =-\alpha_{Y}+\alpha_{Y}\beta_{2}\frac{\kappa_{L}^{\prime}(1+\theta
_{2})-\kappa_{L}^{\prime}(\theta_{2})}{\kappa_{L}^{\prime\prime}(\theta_{2})},
\end{align*}
and%
\begin{align*}
\lim_{\tau\rightarrow0}\frac{d}{d\tau}\Psi_{\theta_{2},\beta_{2}}^{0}(\tau)
&  =\lim_{\tau\rightarrow0}\Lambda_{0}^{\theta_{2},\beta_{2}}(\Psi_{\theta
_{2},\beta_{2}}^{1}(\tau))=\Lambda_{0}^{\theta_{2},\beta_{2}}(1)\\
&  =\kappa_{L}(1+\theta_{2})-\kappa_{L}(\theta_{2}).
\end{align*}
In Theorem \ref{Theorem_Main_Geom}, it is proved that $\Psi_{\theta_{2}%
,\beta_{2}}^{1}(\tau)$ converges to $0$ when $\tau$ tends to infinity and
\[
\lim_{\tau\rightarrow\infty}\Psi_{\theta_{2},\beta_{2}}^{0}(\tau)=\int
_{0}^{\infty}\Lambda_{0}^{\theta_{2},\beta_{2}}(\Psi_{\theta_{2},\beta_{2}%
}^{1}(t))dt.
\]
Hence, using the definitions of $\Lambda_{0}^{\theta_{2},\beta_{2}}(u)$ and
$\Lambda_{0}^{0,0}(u),$ the fact that $\Psi_{0,0}^{1}(t)=e^{-\alpha_{Y}t}$ and
$\kappa_{L}(0)=0$ we get
\[
\lim_{\tau\rightarrow\infty}(\Psi_{\theta_{2},\beta_{2}}^{0}(\tau)-\Psi
_{0,0}^{0}(\tau))=\int_{0}^{\infty}\kappa_{L}(\Psi_{\theta_{2},\beta_{2}}%
^{1}(t)+\theta_{2})-\kappa_{L}(\theta_{2})-\kappa_{L}(e^{-\alpha_{Y}t})dt.
\]

\end{proof}

The sign of $\Sigma(t,\tau)$ is more complex to analyse than the sign of
$R_{a,Q}^{F}(t,\tau),$ the risk premium in the arithmetic model. In the
Esscher case the computations can be done quite explicitly. In the general
case we shall make use of Lemma \ref{Lemma_Main_RPG} to prove that one can
generate the empirically observed risk premium profile. Moreover, some
additional information on $\Sigma(t,\tau)$ can be deduced from classical
results on comparison of solutions of ODEs. In order to graphically illustrate
the discussion we plot the risk premium profiles obtained assuming that the
subordinator $L$ is a compound Poisson process with jump intensity
$c/\lambda>0$ and exponential jump sizes with mean $\lambda.$ That is, $L$
will have the L\'{e}vy measure given in Example $\left(
\ref{Example_Subordinators}\right)  ,$ $(1).$ We shall measure the time to
maturity $\tau$ in days and plot $R_{g,Q}^{F}(t,\tau)$ for $\tau\in
\lbrack0,360],$ roughly one year. We fix the values of the following
parameters%
\[
\alpha_{X}=0.099,\sigma_{X}=0.0158,\alpha_{Y}=0.3466,c=0.4,\lambda=2.
\]
The speed of mean reversion for the base component $\alpha_{X}$ yields a
half-life of seven days, while the one for the spikes $\alpha_{Y}$ yields a
half-life of two days. The value for $\sigma_{X}$ yields an annualised
volatility of $30\%$. The values for $c$ and $\lambda$ give jumps with mean
$0.5$ and frequency of $5$ spikes a month.

\begin{itemize}
\item \textbf{Changing the level of mean reversion (Esscher transform)},
$\bar{\beta}=(0,0):$ Setting $\bar{\beta}=(0,0),$ the probability measure $Q$
only changes the level of mean reversion (which is assumed to be zero under
the historical measure $P$). Moreover, as $R_{a,Q}^{F}(t,\tau)$ is
deterministic when $\bar{\beta}=(0,0),$ we have that the randomness in
$R_{g,Q}^{F}(t,\tau)$ comes into the picture through $\mathbb{E}%
_{P}[S(T)|\mathcal{F}_{t}],$ in particular through the levels of the driving
factors $X$ and $Y.$ By Proposition \ref{Prop_Esscher_Geom} we have that
\begin{align*}
R_{g,Q}^{F}(t,\tau)  &  =\mathbb{E}_{P}[S(t+\tau)|\mathcal{F}_{t}]\\
&  \qquad\times\left\{  \exp\left(  R_{a,Q}^{F}(t,\tau)-\frac{\kappa
_{L}^{\prime}(\theta_{2})-\kappa_{L}^{\prime}(0)}{\alpha_{Y}}(1-e^{-\alpha
_{Y}\tau})\right)  \right. \\
&  \qquad\times\left.  \exp\left(  \int_{0}^{\tau}\int_{0}^{\infty}%
(e^{\theta_{2}z}-1)(e^{ze^{-\alpha_{Y}s}}-1)\ell(dz)ds\right)  -1\right\}  ,
\end{align*}
and the sign of $R_{g,Q}^{F}(t,\tau)$ is the same as the sign of
\begin{align*}
&  R_{a,Q}^{F}(t,\tau)-\frac{\kappa_{L}^{\prime}(\theta_{2})-\kappa
_{L}^{\prime}(0)}{\alpha_{Y}}(1-e^{-\alpha_{Y}\tau})+\int_{0}^{\tau}\int
_{0}^{\infty}(e^{\theta_{2}z}-1)(e^{ze^{-\alpha_{Y}s}}-1)\ell(dz)ds\\
&  \qquad=\frac{\theta_{1}}{\alpha_{X}}(1-e^{-\alpha_{X}\tau})+\int_{0}^{\tau
}\int_{0}^{\infty}(e^{\theta_{2}z}-1)(e^{ze^{-\alpha_{Y}s}}-1)\ell(dz)ds,
\end{align*}
which is equal to $\Sigma(t,\tau)$ in Lemma \ref{Lemma_Main_RPG}.

\begin{figure}[t]
\centering
\subfigure[$\theta_1=-0.3,\theta_2=0.9,X(t)=-0.5,Y(t)=0.5$]{\includegraphics[width=2.90in]{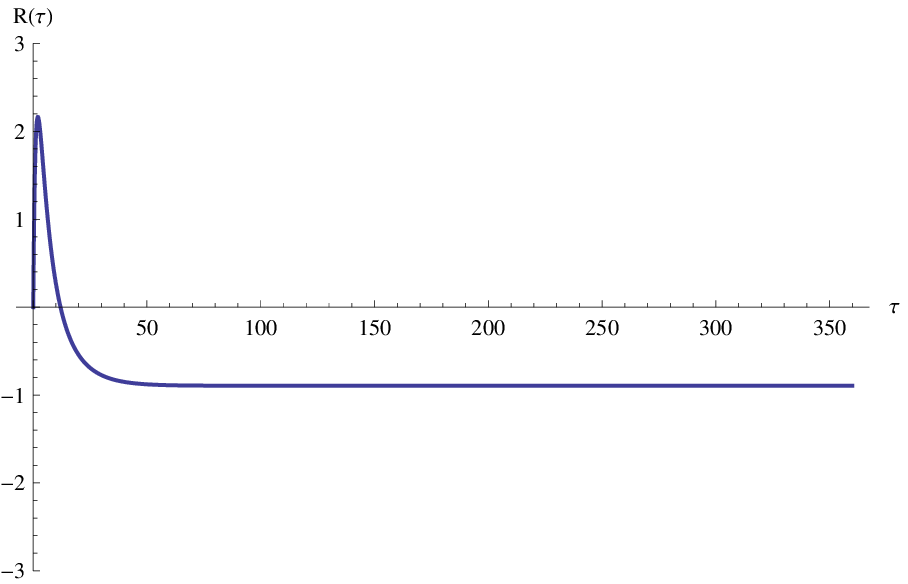}\label{FigBeta0_5_a}}
\subfigure[$\theta_1=0.03,\theta_2=-0.9,X(t)=0.5,Y(t)=0.5$]{\includegraphics[width=2.90in]{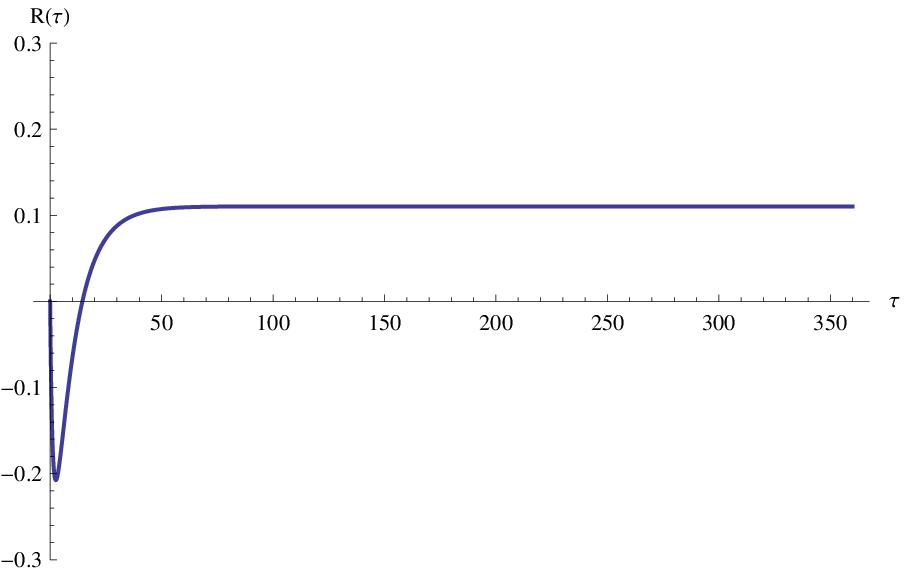}\label{FigBeta0_5_b}}\newline%
\subfigure[$\theta_1=-0.09,\theta_2=0.9,X(t)=-0.5,Y(t)=0.5$]{\includegraphics[width=2.90in]{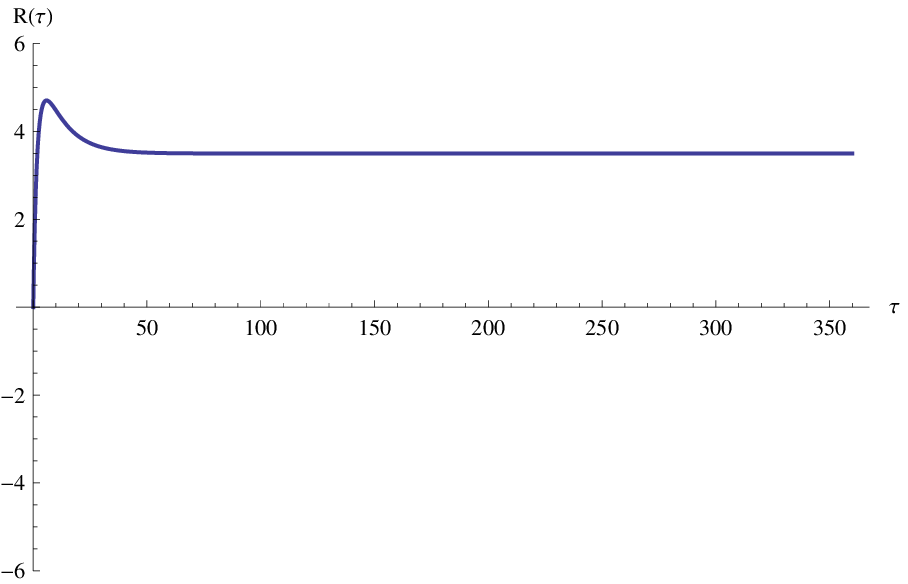}\label{FigBeta0_5_c}}
\subfigure[$\theta_1=-0.2,\theta_2=0.1,X(t)=0.5,Y(t)=0.5$]{\includegraphics[width=2.90in]{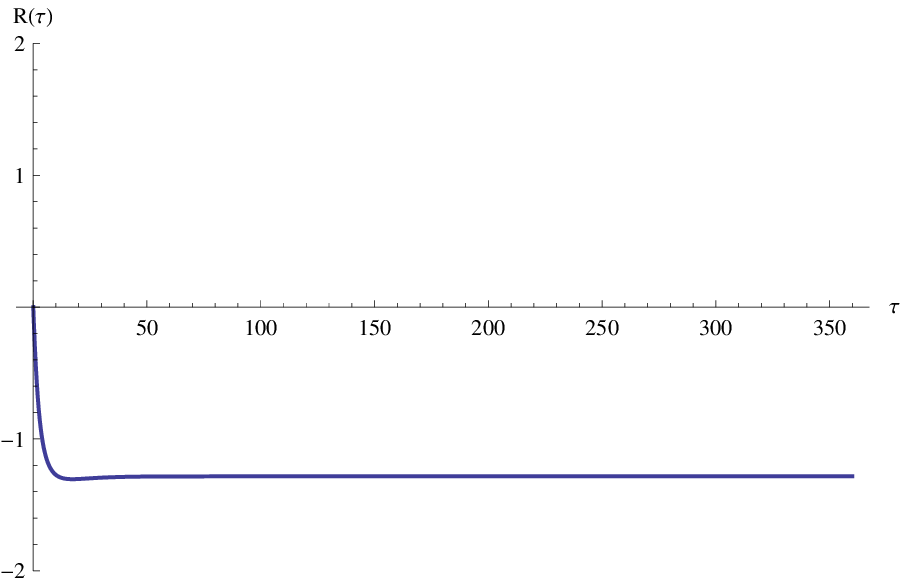}\label{FigBeta0_5_d}}\caption{Risk
premium profiles when $L$ is a compound Poisson process with exponentially
distributed jumps. Esscher transform: case $\bar{\beta}=(0,0).$ Geometric spot
model}%
\label{FigBeta0_5}%
\end{figure}If $\theta_{2}=0,$ then the sign of $R_{g,Q}^{F}(t,\tau)$ is the
same as the sign of $\theta_{1}$ and it is constant over all times to maturity
$\tau.$ Similarly, if $\theta_{1}=0,$ the sign $R_{g,Q}^{F}(t,\tau)$ is the
same as the sign of $\theta_{2}$ and it is also constant. If both $\theta_{1}$
and $\theta_{2}$ are different from zero we can get risk premium profiles with
non constant sign. By Lemma \ref{Lemma_Main_RPG}, we have that
\begin{align*}
\lim_{\tau\rightarrow0}\frac{\partial}{\partial\tau}\Sigma(t,\tau)  &
=\theta_{1}+\kappa_{L}(1+\theta_{2})-\kappa_{L}(\theta_{2})-\kappa_{L}(1)\\
&  =\theta_{1}+\int_{0}^{\infty}(e^{\theta_{2}z}-1)(e^{z}-1)\ell(dz).
\end{align*}
Hence, if we want the sign of $R_{g,Q}^{F}(t,\tau)$ to be positive when $\tau$
is close to zero we have to impose
\begin{equation}
\theta_{1}+\int_{0}^{\infty}(e^{\theta_{2}z}-1)(e^{z}-1)\ell(dz)>0.
\label{Equ_GeomSmallTau}%
\end{equation}
For large times to maturity, Lemma \ref{Lemma_Main_RPG} yields%
\begin{align*}
\lim_{\tau\rightarrow\infty}\Sigma(t,\tau)  &  =\frac{\theta_{1}}{\alpha_{X}%
}+\int_{0}^{\infty}\kappa_{L}(e^{-\alpha_{Y}t}+\theta_{2})-\kappa_{L}%
(\theta_{2})-\kappa_{L}(e^{-\alpha_{Y}t})dt\\
&  =\frac{\theta_{1}}{\alpha_{X}}+\int_{0}^{\infty}\int_{0}^{\infty}%
(e^{\theta_{2}z}-1)(e^{ze^{-\alpha_{Y}t}}-1)\ell(dz)dt.
\end{align*}
Using Fubini's theorem we get that
\begin{align*}
&  \int_{0}^{\infty}\int_{0}^{\infty}(e^{\theta_{2}z}-1)(e^{ze^{-\alpha_{Y}t}%
}-1)\ell(dz)dt\\
&  =\int_{0}^{\infty}(e^{\theta_{2}z}-1)\int_{0}^{\infty}(e^{ze^{-\alpha_{Y}%
t}}-1)dt\ell(dz)\\
&  =\int_{0}^{\infty}\frac{(e^{\theta_{2}z}-1)}{\alpha_{Y}}\left(
\operatorname{Ei}(z)-\log(z)-\gamma\right)  \ell(dz),
\end{align*}
where $\operatorname{Ei}(z)=\int_{-\infty}^{z}\frac{e^{t}}{t}dt$ is the
exponential integral function and $\gamma$ is the Euler-Mascheroni constant.
Hence, if we want $R_{g,Q}^{F}(t,\tau)$ to be negative when $\tau$ is large we
have to impose%
\begin{equation}
\theta_{1}+\frac{\alpha_{X}}{\alpha_{Y}}\int_{0}^{\infty}(e^{\theta_{2}%
z}-1)\left(  \operatorname{Ei}(z)-\log(z)-\gamma\right)  \ell(dz)<0.
\label{Equ_GeomLargeTau}%
\end{equation}
Note that $\operatorname{Ei}(z)-\log(z)-\gamma\geq0,\forall z\geq0$ and
$e^{z}-1-\frac{\alpha_{X}}{\alpha_{Y}}\left(  \operatorname{Ei}(z)-\log
(z)-\gamma\right)  >0,$ for all $z>0$ and $\alpha_{X}<\alpha_{Y}.$Therefore,
for all $\theta_{2}>0$ one has that%
\begin{equation}
0<\int_{0}^{\infty}(e^{\theta_{2}z}-1)\left(  \operatorname{Ei}(z)-\log
(z)-\gamma\right)  \ell(dz)<\frac{\alpha_{Y}}{\alpha_{X}}\int_{0}^{\infty
}(e^{\theta_{2}z}-1)\left(  e^{z}-1\right)  \ell(dz).
\label{Equ_inequality_IntegralsEi}%
\end{equation}
Combining equations $\left(  \ref{Equ_GeomSmallTau}\right)  $, $\left(
\ref{Equ_GeomLargeTau}\right)  $ and $\left(  \ref{Equ_inequality_IntegralsEi}%
\right)  $ we can conclude that it is possible to choose $\theta_{1}<0$ and
$\theta_{2}>0$ such that $R_{g,Q}^{F}(t,\tau)>0$ when the time to maturity is
close to zero and $R_{g,Q}^{F}(t,\tau)<0$ when the time to maturity is large.

\item \textbf{Changing the speed of mean reversion,} $\bar{\theta}=(0,0):$
Setting $\bar{\theta}=(0,0),$ the probability measure $Q$ only changes the
speed of mean reversion. By Lemma \ref{Lemma_Main_RPG} we have that the sign
of $R_{g,Q}^{F}(t,\tau)$ will coincide with the sign of
\begin{align*}
\Sigma(t,\tau)  &  =X(t)e^{-\alpha_{X}\tau}(e^{\alpha_{X}\beta_{1}\tau
}-1)+Y(t)\left(  \Psi_{0,\beta_{2}}^{1}(\tau)-\Psi_{0,0}^{1}(\tau)\right) \\
&  +\frac{\sigma_{X}^{2}}{4\alpha_{X}}\Lambda(2\alpha_{X}\tau,1-\beta
_{2})+\left(  \Psi_{0,\beta_{2}}^{0}(\tau)-\Psi_{0,0}^{0}(\tau)\right) \\
&  \triangleq\Sigma_{1}(t,\tau)+\Sigma_{2}(t,\tau)+\Sigma_{3}(t,\tau
)+\Sigma_{4}(t,\tau),
\end{align*}
and \begin{figure}[b]
\centering
\subfigure[$\beta_1=0.4,\beta_2=0.2,X(t)=1.0,Y(t)=0.5$]{\includegraphics[width=2.90in]{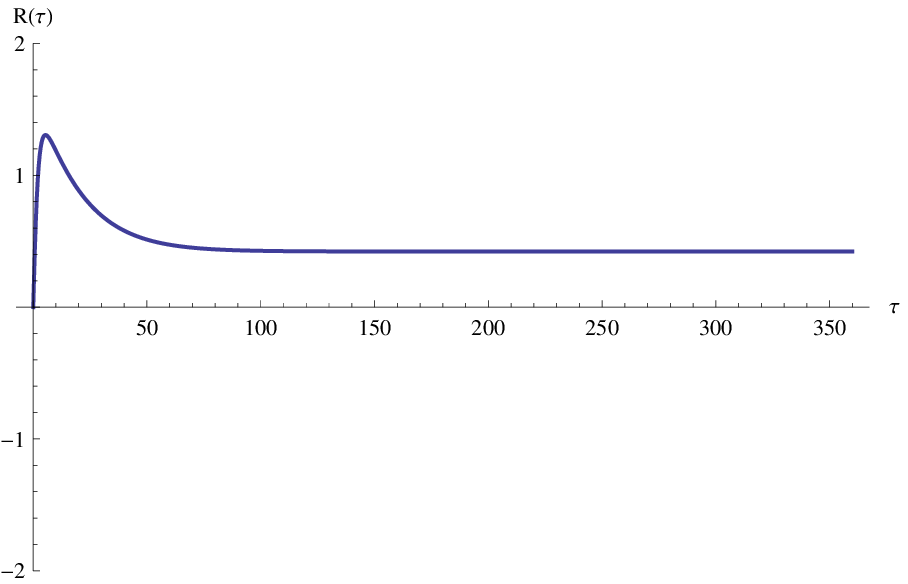}\label{FigTheta0_6_a}}
\subfigure[$\beta_1=0.75,\beta_2=0.0,X(t)=-2.5,Y(t)=0.5$]{\includegraphics[width=2.90in]{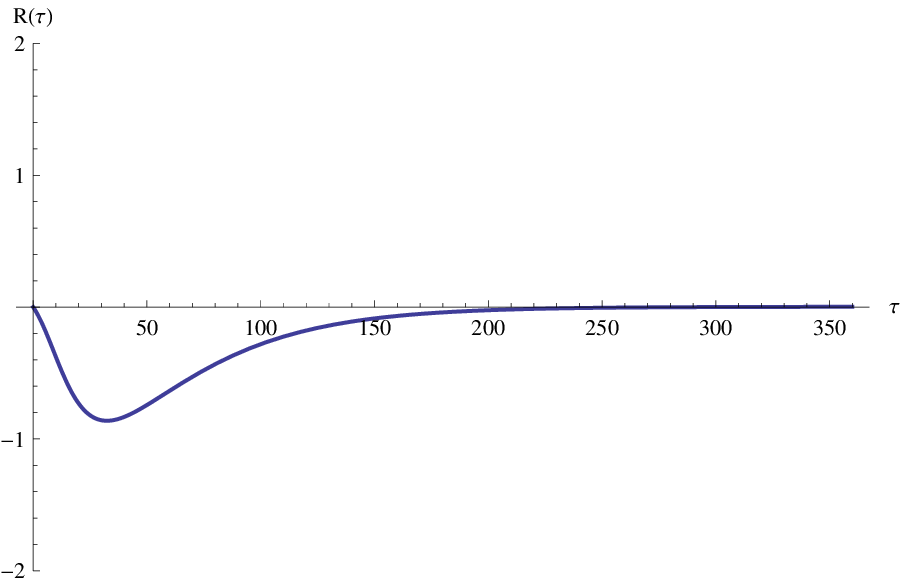}\label{FigTheta0_6_b}}\newline%
\subfigure[$\beta_1=0.75,\beta_2=0.3,X(t)=-2.5,Y(t)=0.0$]{\includegraphics[width=2.90in]{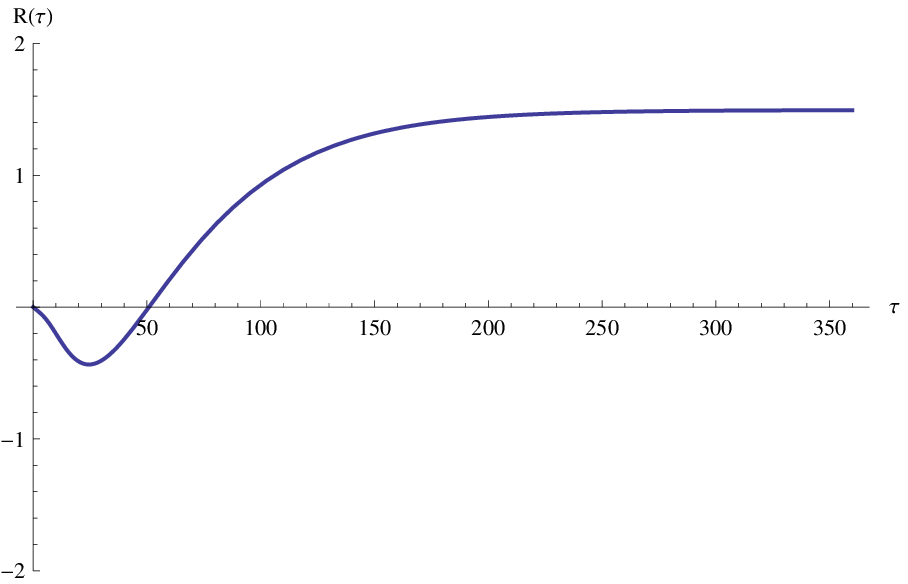}\label{FigTheta0_6_c}}
\subfigure[$\beta_1=0.5,\beta_2=0.2,X(t)=-2.5,Y(t)=2.5$]{\includegraphics[width=2.90in]{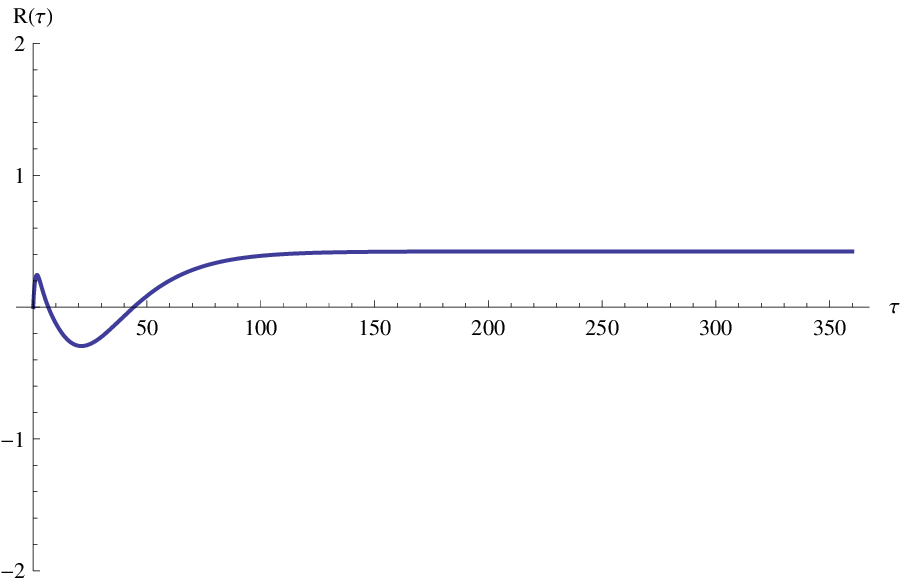}\label{FigTheta0_6_d}}\caption{Risk
premium profiles when $L$ is a compound Poisson process with exponentially
distributed jumps. Case $\bar{\theta}=(0,0).$ Geometric spot price model}%
\label{FigTheta0_6}%
\end{figure}%

\begin{align*}
\lim_{\tau\rightarrow\infty}\Sigma(t,\tau)  &  =\frac{\sigma_{X}^{2}}%
{4\alpha_{X}}\frac{\beta_{1}}{1-\beta_{1}}\geq0\\
\lim_{\tau\rightarrow0}\frac{\partial}{\partial\tau}\Sigma(t,\tau)  &
=X(t)\alpha_{X}\beta_{1}+Y(t)\alpha_{Y}\beta_{2}\frac{\kappa_{L}^{\prime
}(1)-\kappa_{L}^{\prime}(0)}{\kappa_{L}^{\prime\prime}(0)},
\end{align*}
where $\kappa_{L}^{\prime}(1)-\kappa_{L}^{\prime}(0)$ and $\kappa_{L}%
^{\prime\prime}(0)$ are strictly positive. Hence the risk premium will
approach to a non negative value in the long end of the market. In the short
end, it can be both positive or negative and stochastically varying with
$X(t)$ and $Y(t),$ but $Y(t)$ will always contribute to a positive sign. For
any $\tau,$ the sign of $\Sigma_{1}(t,\tau)$ will be the sign of $X(t),$ that
can be positive or negative. As the function $\Lambda(x,y)\ $is positive, the
term $\Sigma_{3}(t,\tau)$ is always positive. To analyse the sign of
$\Sigma_{2}(t,\tau),$ note that
\[
\Lambda_{1}^{0,\beta_{2}}(u)-\Lambda_{1}^{0,0}(u)=\frac{\alpha_{Y}\beta_{2}%
}{\kappa_{L}^{\prime\prime}(0)}\int_{0}^{\infty}(e^{uz}-1)z\ell(dz)\geq0,\quad
u\geq0,
\]
and $\Psi_{0,\beta_{2}}^{1}(1)=\Psi_{0,0}^{1}(\tau).$ Hence, applying a
comparison theorem for ODEs, see Theorem 6.1, pag.31, in Hale \cite{Ha69}, we
have that $\Psi_{0,\beta_{2}}^{1}(\tau)-\Psi_{0,0}^{1}(\tau)\geq0,$ for all
$\tau,$ and, as $Y(t)$ is always positive, the term $\Sigma_{2}(t,\tau)$ is
also always positive. Finally, as
\[
\Lambda_{0}(u)\triangleq\Lambda_{0}^{0,\beta_{2}}(u)=\Lambda_{0}^{0,0}%
(u)=\int_{0}^{\infty}(e^{uz}-1)\ell(dz),
\]
is an strictly increasing function and $\Psi_{0,\beta_{2}}^{1}(t)\geq
\Psi_{0,0}^{1}(t)$ we get that%
\[
\Sigma_{4}(t,\tau)=\Psi_{0,\beta_{2}}^{0}(\tau)-\Psi_{0,0}^{0}(\tau)=\int
_{0}^{\tau}\{\Lambda_{0}(\Psi_{0,\beta_{2}}^{1}(t))-\Lambda_{0}(\Psi_{0,0}%
^{1}(t))\}dt\geq0.
\]
Hence, if $\beta_{1}=0$ or $X(t)\geq0,$ then $R_{g,Q}^{F}(t,\tau)$ will be
positive for all times to maturity. Some of the possible risk profiles that
can be obtained are plotted in Figure \ref{FigTheta0_6}.

\item \textbf{Changing the level and speed of mean reversion simultaneously}:
We proceed as in the arithmetic case. As we are more interested in how the
change of measure $Q$ influence the component $Y(t)$, responsible for the
spikes in the prices, we are going to assume that $\beta_{1}=0.$ This means
that $Q$ may change the level of mean reversion of the regular component
$X(t),$ but not the speed at which this component reverts to that level.
According to Lemma \ref{Lemma_Main_RPG} we have that the sign of $R_{g,Q}%
^{F}(t,\tau)$ will coincide with the sign of
\begin{align*}
\Sigma(t,\tau)  &  =Y(t)(\Psi_{\theta_{2},\beta_{2}}^{1}(\tau)-e^{-\alpha
_{Y}\tau})+\frac{\theta_{1}}{\alpha_{X}}(1-e^{-\alpha_{X}\tau})+\Psi
_{\theta_{2},\beta_{2}}^{0}(\tau)\\
&  -\int_{0}^{\tau}\int_{0}^{\infty}(e^{ze^{-\alpha_{Y}s}}-1)\ell(dz)ds\\
&  =Y(t)\left(  \Psi_{\theta_{2},\beta_{2}}^{1}(\tau)-\Psi_{0,0}^{1}%
(\tau)\right)  +\frac{\theta_{1}}{\alpha_{X}}(1-e^{-\alpha_{X}\tau})+\left(
\Psi_{\theta_{2},\beta_{2}}^{0}(\tau)-\Psi_{0,0}^{0}(\tau)\right) \\
&  \triangleq\Sigma_{1}(t,\tau)+\Sigma_{2}(t,\tau)+\Sigma_{3}(t,\tau),
\end{align*}
and%
\begin{align}
\lim_{\tau\rightarrow\infty}\Sigma(t,\tau)  &  =\frac{\theta_{1}}{\alpha_{X}%
}+\int_{0}^{\infty}\kappa_{L}(\Psi_{\theta_{2},\beta_{2}}^{1}(t)+\theta
_{2})-\kappa_{L}(\theta_{2})-\kappa_{L}(e^{-\alpha_{Y}t}%
)dt\label{Equ_RPG_General_1}\\
&  =\frac{\theta_{1}}{\alpha_{X}}+\int_{0}^{\infty}\int_{0}^{1}\kappa
_{L}^{\prime}(\theta_{2}+\lambda\Psi_{\theta_{2},\beta_{2}}^{1}(t))d\lambda
\Psi_{\theta_{2},\beta_{2}}^{1}(t)dt\nonumber\\
&  -\int_{0}^{\infty}\int_{0}^{1}\kappa_{L}^{\prime}(\theta_{2}+\lambda
e^{-\alpha_{Y}t})d\lambda e^{-\alpha_{Y}t}dt\nonumber\\
\lim_{\tau\rightarrow0}\frac{\partial}{\partial\tau}\Sigma(t,\tau)  &
=Y(t)\alpha_{Y}\beta_{2}\frac{\kappa_{L}^{\prime}(1+\theta_{2})-\kappa
_{L}^{\prime}(\theta_{2})}{\kappa_{L}^{\prime\prime}(\theta_{2})}%
\label{Equ_RPG_General_2}\\
&  +\theta_{1}+\kappa_{L}(1+\theta_{2})-\kappa_{L}(\theta_{2})-\kappa
_{L}(1)\nonumber\\
&  =Y(t)\alpha_{Y}\beta_{2}\frac{\kappa_{L}^{\prime}(1+\theta_{2})-\kappa
_{L}^{\prime}(\theta_{2})}{\kappa_{L}^{\prime\prime}(\theta_{2})}+\theta
_{1}\nonumber\\
&  +\int_{0}^{1}\kappa_{L}^{\prime}(\theta_{2}+\lambda)d\lambda-\kappa
_{L}(1)\nonumber
\end{align}
\begin{figure}[t]
\centering
\subfigure[$\beta_1=0,\beta_2=0.2,\theta_1=-0.1,\theta_2=0.2,X(t)=1.0,Y(t)=1.0$]{\includegraphics[width=2.90in]{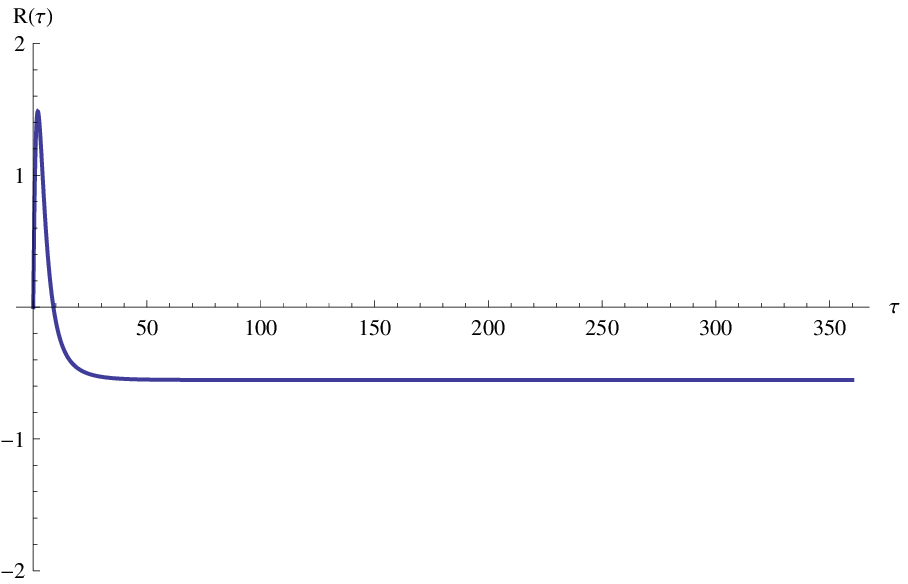}\label{Fig_7_a}}\caption{Risk
premium profiles when $L$ is a compound Poisson process with exponentially
distributed jumps. Geometric spot model}%
\label{Fig_7}%
\end{figure}Note that we can make equation $\left(  \ref{Equ_RPG_General_1}%
\right)  $ negative by simply choosing $\theta_{1}$
\begin{align}
\theta_{1}  &  <-\alpha_{X}\int_{0}^{\infty}\int_{0}^{1}\kappa_{L}^{\prime
}(\theta_{2}+\lambda\Psi_{\theta_{2},\beta_{2}}^{1}(t))d\lambda\Psi
_{\theta_{2},\beta_{2}}^{1}(t)dt\label{Equ_RPG_General_3}\\
&  +\alpha_{X}\int_{0}^{\infty}\int_{0}^{1}\kappa_{L}^{\prime}(\theta
_{2}+\lambda e^{-\alpha_{Y}t})d\lambda e^{-\alpha_{Y}t}dt.\nonumber
\end{align}
On the other hand, to make equation $(\ref{Equ_RPG_General_2})$ positive, we
have to choose $\theta_{1}$ satisfying%
\begin{equation}
\theta_{1}>-\int_{0}^{1}\kappa_{L}^{\prime}(\theta_{2}+\lambda)d\lambda
+\kappa_{L}(1)-Y(t)\alpha_{Y}\beta_{2}\frac{\kappa_{L}^{\prime}(1+\theta
_{2})-\kappa_{L}^{\prime}(\theta_{2})}{\kappa_{L}^{\prime\prime}(\theta_{2})}.
\label{Equ_RPG_General_4}%
\end{equation}
Equations $\left(  \ref{Equ_RPG_General_3}\right)  $ and $\left(
\ref{Equ_RPG_General_4}\right)  $ are compatible if the following equation is
satisfied%
\begin{align}
U_{+}(\theta_{2},\beta_{2})  &  \triangleq\int_{0}^{1}\kappa_{L}^{\prime
}(\theta_{2}+\lambda)d\lambda+\alpha_{X}\int_{0}^{\infty}\int_{0}^{1}%
\kappa_{L}^{\prime}(\theta_{2}+\lambda e^{-\alpha_{Y}t})d\lambda
e^{-\alpha_{Y}t}dt\nonumber\\
&  +Y(t)\alpha_{Y}\beta_{2}\frac{\kappa_{L}^{\prime}(1+\theta_{2})-\kappa
_{L}^{\prime}(\theta_{2})}{\kappa_{L}^{\prime\prime}(\theta_{2})}\nonumber\\
&  >\alpha_{X}\int_{0}^{\infty}\int_{0}^{1}\kappa_{L}^{\prime}(\theta
_{2}+\lambda\Psi_{\theta_{2},\beta_{2}}^{1}(t))d\lambda\Psi_{\theta_{2}%
,\beta_{2}}^{1}(t)dt+\kappa_{L}(1)\triangleq U_{-}(\theta_{2},\beta_{2}).
\label{Equ_RPG_General_5}%
\end{align}
As $e^{-\alpha_{Y}t}\leq1,\Psi_{\theta_{2},\beta_{2}}^{1}(t)\leq1,\kappa
_{L}^{\prime}(\theta)>0$ and $\kappa_{L}^{\prime\prime}(\theta)>0$ we have
that%
\begin{align*}
\frac{\kappa_{L}^{\prime}(1+\theta_{2})-\kappa_{L}^{\prime}(\theta_{2}%
)}{\kappa_{L}^{\prime\prime}(\theta_{2})}  &  =\frac{\int_{0}^{1}\kappa
_{L}^{\prime\prime}(\theta_{2}+\lambda)d\lambda}{\kappa_{L}^{\prime\prime
}(\theta_{2})}>1,\\
\int_{0}^{\infty}\int_{0}^{1}\kappa_{L}^{\prime}(\theta_{2}+\lambda
e^{-\alpha_{Y}t})d\lambda e^{-\alpha_{Y}t}dt  &  \geq\frac{\kappa_{L}^{\prime
}(\theta_{2})}{\alpha_{Y}},
\end{align*}
and%
\[
\int_{0}^{\infty}\int_{0}^{1}\kappa_{L}^{\prime}(\theta_{2}+\lambda
\Psi_{\theta_{2},\beta_{2}}^{1}(t))d\lambda\Psi_{\theta_{2},\beta_{2}}%
^{1}(t)dt\leq\int_{0}^{1}\kappa_{L}^{\prime}(\theta_{2}+\lambda)d\lambda
\int_{0}^{\infty}\Psi_{\theta_{2},\beta_{2}}^{1}(t)dt,
\]
As $\Psi_{\theta_{2},\beta_{2}}^{1}(t)$ converges to zero exponentially fast,
see equation $\left(  \ref{Equ_Exp_Psi1}\right)  $, we have that
\[
\int_{0}^{\infty}\Psi_{\theta_{2},\beta_{2}}^{1}(t)dt<\infty.
\]
Actually, as $\Lambda_{1}^{\theta_{2},\beta_{2}}(u)<\Lambda_{1}^{\theta
_{2},\beta_{2}}(1)<0,0<u<1,$ we can use a comparison theorem for ODEs to
obtain that
\[
\Psi_{\theta_{2},\beta_{2}}^{1}(t)\leq e^{\Lambda_{1}^{\theta_{2},\beta_{2}%
}(1)t}=\exp\left(  -\alpha_{Y}(1-\frac{\beta_{2}}{\kappa_{L}^{\prime\prime
}(\theta_{2})}(\kappa_{L}^{\prime}(1+\theta_{2})-\kappa_{L}^{\prime}%
(\theta_{2})))t\right)  ,
\]
which yields
\[
\int_{0}^{\infty}\Psi_{\theta_{2},\beta_{2}}^{1}(t)dt\leq\frac{1}{\alpha_{Y}%
}(1-\beta_{2}\frac{\int_{0}^{1}\kappa_{L}^{\prime\prime}(\theta_{2}%
+\lambda)d\lambda}{\kappa_{L}^{\prime\prime}(\theta_{2})})^{-1}.
\]
Hence,
\begin{align*}
U_{+}(\theta_{2},\beta_{2})  &  \geq\int_{0}^{1}\kappa_{L}^{\prime}(\theta
_{2}+\lambda)d\lambda+\frac{\alpha_{X}}{\alpha_{Y}}\kappa_{L}^{\prime}%
(\theta_{2})+Y(t)\alpha_{Y}\beta_{2}\triangleq V_{+}(\theta_{2},\beta_{2}),\\
U_{-}(\theta_{2},\beta_{2})  &  \leq\frac{\alpha_{X}}{\alpha_{Y}}\int_{0}%
^{1}\kappa_{L}^{\prime}(\theta_{2}+\lambda)d\lambda(1-\beta_{2}\frac{\int
_{0}^{1}\kappa_{L}^{\prime\prime}(\theta_{2}+\lambda)d\lambda}{\kappa
_{L}^{\prime\prime}(\theta_{2})})^{-1}+\kappa_{L}(1)\triangleq V_{-}%
(\theta_{2},\beta_{2}),
\end{align*}
and if we can find $\theta_{2}\in D_{L}^{g}(\delta)$ for some $\delta>0$ and
$\beta_{2}\in(0,1)$ such that $V_{+}(\theta_{2},\beta_{2})>V_{-}(\theta
_{2},\beta_{2})$ then equation $\left(  \ref{Equ_RPG_General_5}\right)  $ will
be satisfied. Note that the larger the value of $Y(t)$ the easier to find such
$\theta_{2}$ and $\beta_{2}.$ Even in the case that $Y(t)=0,$ by choosing
$\beta_{2}$ close to zero and $\theta_{2}$ large enough we can get
$V_{+}(\theta_{2},\beta_{2})>V_{-}(\theta_{2},\beta_{2}).$ This shows that we
can create a change of measure $Q$ generating the empirically observed risk
premium profile, see Figure \ref{Fig_7}.
\end{itemize}

\end{document}